\documentclass[letterpaper, 10pt, conference]{ieeeconf}

\IEEEoverridecommandlockouts

\newcommand{\docversion}{2}    
\usepackage{comment}  

\newif\ifextended     
\newif\ifarxiv        
\newif\ifsubmit       
\newif\iffull         

\ifnum\docversion=1
  \extendedtrue  \arxivfalse \submitfalse \fulltrue
\fi
\ifnum\docversion=2
  \extendedfalse \arxivtrue  \submitfalse \fulltrue
\fi
\ifnum\docversion=3
  \extendedfalse \arxivfalse \submittrue  \fullfalse
\fi

\ifextended
  \includecomment{extendedonly}
\else
  \excludecomment{extendedonly}
\fi

\ifextended
  \excludecomment{compactonly}
\else
  \includecomment{compactonly}
\fi

\iffull
  \includecomment{fullonly}
\else
  \excludecomment{fullonly}
\fi

\ifsubmit
  \includecomment{submitonly}
\else
  \excludecomment{submitonly}
\fi

\ifarxiv
  \includecomment{arxivonly}
\else
  \excludecomment{arxivonly}
\fi

\newcommand{\extonly}[1]{\ifextended #1\fi}

\newcommand{\fulonly}[1]{\iffull #1\fi}

\newcommand{\subonly}[1]{\ifsubmit #1\fi}

\newcommand{\arxonly}[1]{\ifarxiv #1\fi}

\newcommand{\cdconly}[1]{\ifextended\else #1\fi}

\usepackage{amsmath,amssymb,amsfonts}
\usepackage{mathtools}

\usepackage{amsthm}
\usepackage{bm}
\usepackage{xcolor}
\usepackage{todonotes}

\let\labelindent\relax
\usepackage{enumitem}
\usepackage{graphicx}
\usepackage{svg}
\usepackage{booktabs}
\usepackage{subcaption}

\usepackage{algorithm}
\usepackage{algpseudocode}

\usepackage[hidelinks,hypertexnames=false]{hyperref}
\usepackage[capitalize,noabbrev,nameinlink]{cleveref}

\usepackage{subcaption}
\usepackage{caption}
\usepackage[backend=biber,style=ieee,sorting=none,maxbibnames=6,minbibnames=1]{biblatex}
\newcommand{\R}{\mathbb{R}}
\newcommand{\N}{\mathbb{N}}
\newcommand{\X}{\mathbb{R}^{n_x}}
\newcommand{\XX}{\mathcal{X}_0}
\newcommand{\U}{\mathcal{U}}
\newcommand{\Cset}{\mathcal{C}}
\newcommand{\Kset}{\mathcal{X}}

\newcommand{\Rset}{\mathcal{R}}

\definecolor{editblue}{RGB}{0,70,170}

\newcommand{\Prob}{\mathbb{P}}

\newcommand{\ind}{\mathbf{1}}

\newcommand{\norm}[1]{\left\lVert #1 \right\rVert}

\newcommand{\ip}[2]{\left\langle #1,#2 \right\rangle}
\newcommand{\dd}{\mathrm{d}}

\let\CzNorm\FuncNorm

\newcommand{\res}{\varepsilon}      
\newcommand{\Oset}{\Omega}
\newcommand{\kappaSR}{\kappa}

\newcommand{\revised}[1]{{\color{editblue}#1}}

\DeclareMathOperator*{\argmin}{arg\,min}

\newtheorem{theorem}{Theorem}
\newtheorem{example}{Example}
\newtheorem{lemma}{Lemma}
\newtheorem{proposition}{Proposition}
\newtheorem{corollary}{Corollary}
\newtheorem{definition}{Definition}
\newtheorem{assumption}{Assumption}

\theoremstyle{remark}
\newtheorem{remark}{Remark}

\crefname{assumption}{Assumption}{Assumptions}
\crefname{proposition}{Proposition}{Propositions}
\crefname{lemma}{Lemma}{Lemmas}
\crefname{theorem}{Theorem}{Theorems}
\crefname{corollary}{Corollary}{Corollaries}
\crefname{definition}{Definition}{Definitions}
\crefname{remark}{Remark}{Remarks}
\crefname{algorithm}{Algorithm}{Algorithms}

\newcommand{\ct}{\mathrm{ct}}
\newcommand{\dt}{\mathrm{dt}}

\title{\LARGE \bf
Robust Conformal CBF and CLF Controllers\\ via Iterative Policy Updates
}

\author{Omid Mirzaeedodangeh$^{1}$, Eliot Shekhtman$^{2}$, Nikolai Matni$^{2,3}$, and Lars Lindemann$^{1}$%
\thanks{$^{1}$Automatic Control Laboratory, ETH Z\"urich, Switzerland.}%
\thanks{$^{2}$Computer and Information Science, University of Pennsylvania, USA.}%
\thanks{$^{3}$Electrical and Systems Engineering, University of Pennsylvania, USA.}%
}

\begin{document}

\maketitle
\thispagestyle{empty}
\pagestyle{empty}

\begin{abstract}

Conformal prediction (CP) has been used to obtain probabilistic bounds on the error between a learned dynamics model and the true but unknown system. Such CP bounds can then be embedded into robust control Lyapunov function (CLF) and control barrier function (CBF) frameworks. However, such an approach does not retain stability/safety guarantees because of the distribution shift between the closed-loop trajectory distribution under the deployed CLF/CBF policy and the trajectory distribution from which the CP bound and its guarantees were derived. To address this issue, we propose an episodic framework that iteratively updates the robust conformal CLF/CBF policy while maintaining stability/safety guarantees across episodes. 
We achieve this by (1) using adversarially robust conformal prediction, and (2) quantifying a distribution shift budget that allows us to control how much the model error can increase across policy updates. This distribution shift budget is derived via a closed-loop trajectory sensitivity analysis, yielding an implicit and an explicit update rule for the CP bound. We analyze convergence of our algorithm, which we demonstrate on three case studies. To the best of our knowledge, these are the first results that provide stability/safety guarantees for robust conformal CBF/CLF policies.

\end{abstract}

\vspace{0.25em}

\ifextended\input{sections/01_introduction}\fi
\cdconly{\section{Introduction}
\label{sec:intro}

Control Lyapunov functions (CLFs) and control barrier functions (CBFs) are classical tools for enforcing stability and safety in nonlinear control \cite{sontag1989universal,ames2019cbf}. When an accurate model of the dynamical system is available, CLF and CBF conditions can be enforced pointwise---most commonly via a convex quadratic program (QP)---to synthesize feedback control laws that stabilize equilibria and render safe sets forward invariant. In many cases, however, the true dynamics are not known and controllers are instead synthesized using learned or identified nominal models. This introduces a model error that can result in unstable and unsafe behavior.

Robust CLFs and CBFs were proposed to embed a worst-case model-error bound into the control design to address this issue \cite{freeman2008robust,xu2015robustness}. However, this bound depends on the system's operating regime and is usually unknown. Data-driven estimates of this bound can easily be obtained in practice, but they are often heuristic and only under-approximate the true model-error, which can render data-driven error bounds invalid and  result yet again in unstable or unsafe behavior.

\begin{figure}[t]
    \centering
    \includegraphics[
        width=1\linewidth,
    ]{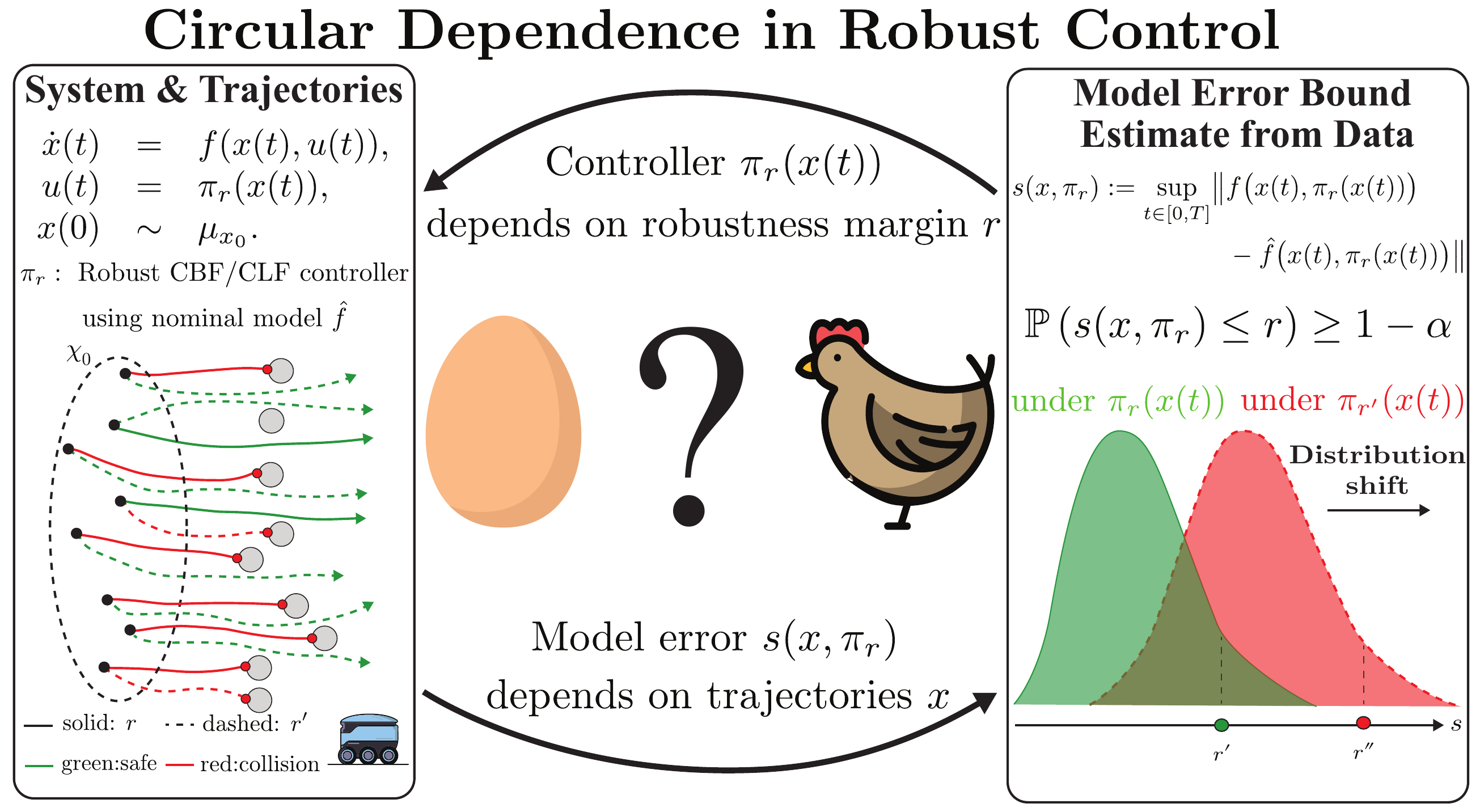}
    \caption{Circular dependence in robust conformal CLF/CBF control. The robustness margin \(r\) determines the controller \(\pi_r\), while the trajectories generated by \(\pi_r\) determine the model error $s$ that is used to estimate \(r\).}
    \label{fig:circular-dependence}
    \vspace{-.7cm}
\end{figure}

Conformal prediction (CP) is a statistical tool for uncertainty quantification in settings without distribution shift \cite{vovk2005alrw,angelopoulos2023conformal,lindemann2025formal}. The authors in \cite{hsu2025} estimate said model-error bound from an offline  calibration dataset by calibrating a trajectory-level model-error using CP. The obtained CP bound comes with probabilistic guarantees and can then be embedded into a robust CLF/CBF framework. However, this probabilistic guarantee does not carry over into stability/safety guarantees because of the distribution shift between the closed-loop trajectory distribution under the CLF/CBF policy and the distribution of the  calibration data. We study how to retain probabilistic stability/safety guarantees when using a robust conformal CLF/CBF policy. 

Our starting point is the episodic framework in \cite{mirzaeedodangeh2025safe}, which proposes iterative policy updates to address interaction-induced distribution shift in safe planning. Here, we iteratively update the CLF/CBF policy by recomputing the CP bound with calibration data collected under the current policy. This creates a circular dependence: the CLF/CBF policy depends on the CP bound, while the trajectory-level model-error and hence the CP bound depend on the policy. This circular dependence is illustrated in \Cref{fig:circular-dependence}. We break this circularity and maintain stability/safety guarantees across episodes by using adversarially robust conformal prediction (ARCP) \cite{gendler2021arcp}. Our main contributions are:
\begin{itemize}[leftmargin=1.5em]
   \item We highlight the issue of distribution shifts for conformal CLF and CBF policies and provide analytical examples where stability/safety cannot be maintained.  

  \item We propose  an episodic framework that addresses distribution shifts for conformal CLF and CBF policies by connecting adversarially robust CP to a distribution shift budget to control how much the model error can increase across a policy update. This leads to an implicit fixed-point and an explicit update rule for the CP bound.

  \item We show that our framework inherits probabilistic stability/safety guarantees across episodes. We further analyze conditions for convergence of our algorithm.
\end{itemize}



\subsection{Related work}
\label{sec:related_work}

\emph{Lyapunov and  barrier functions in control.} CLF and CBF policies are synthesized via convex quadratic programs  \cite{sontag1989universal,ames2019cbf}.  The regularity properties of such optimization-based controllers, which will be important to us later, were analyzed in \cite{mestres2025regularity}.  Robust extensions of CLFs and CBFs can ensure stability and safety even under model uncertainty and will provide the starting point for our work \cite{freeman2008robust,xu2015robustness}. 

\emph{Conformal prediction in control.} A rapidly growing body of literature uses CP for uncertainty quantification in control. For instance, CP was used for safe planning in dynamic and uncertain environments \cite{lindemann2023safeplanning,wang2024probabilistically,ren2023robots}, perception-based control under sensor uncertainty \cite{yang2023safeperception,zhang2025conformal}, and predictive runtime verification \cite{lindemann2023conformal,cairoli2023conformal}; see \cite{lindemann2025formal} for a survey. These works use CP to quantify state estimation, prediction, or model uncertainty, which is then used for downstream planning and control. These works require exchangeable data and do not apply in settings with distribution shifts. 

\emph{Conformal prediction and distribution shifts.} 
Recent work on CP has focused on dealing with distribution shifts, e.g., via adaptive CP for non-exchangeable time series data \cite{gibbs2021adaptive}, weighted CP for distribution shifts of covariates \cite{tibshirani2019covshift}, distributionally robust CP for data perturbations within an ambiguity set \cite{cauchois2024robust}, and adversarially robust CP for worst-case data perturbations \cite{gendler2021arcp}. Adaptive CP was used for control to adapt to arbitrary distribution shifts in \cite{dixit2023adaptive,sheng2024safe}, but generally fails to provide probabilistic control guarantees. Verification and control algorithms using $f$-divergence and Wasserstein distance-based ambiguity sets were proposed in \cite{rahaman2026environments,zhao2024rprvshift}, but require exact knowledge of the ambiguity set. Weighted CP was used in \cite{d2026statistical,srinivasan2026oodmpc}, but requires estimating probability ratios between calibration and deployment distributions  in practice, which can result in degrading probabilistic coverage. Conformal policy learning \cite{huang2024confpolicylearning} designs switching policies to detect and react to distribution shifts. Performative risk control \cite{li2025performativerisk} applies to settings where the CP results itself influences the data distribution, and provides iterative calibration procedures. By contrast, our focus is on policy-induced distribution shifts of CLF/CBF policies for which we use adversarial CP.  Lastly, the work in \cite{contreras2024sodampc} combines out-of-distribution detection with safe fallback controllers.

\emph{Conformal prediction for CLF/CBF policies.} CP has been integrated into CLF/CBF frameworks, e.g.,  \cite{tayal2025cpncbf,cped2025ncbf} use CP to verify learned neural CBFs, \cite{zhang2025conformal,yang2023safeperception} integrate uncertainty sets of state estimators into CBFs, and \cite{zhu2025sparc} uses CP-based prediction sets within CBFs for coupled controllable and uncontrollable agents. The authors in \cite{zhou2024adaptivecp,zhou2025safeRLacp} design learning and policy-iteration methods that combine adaptive CP with CBFs. Closest to our work is \cite{hsu2025} which obtains a CP bound on the error between a learned dynamics model and the true system that is then integrated into a robust CBF.  In contrast, we aim to retain stability/safety guarantees  and acocunt for distribution shifts by proposing an iterative policy update method. Rather than reweighting or adaptively tracking distribution shifts, we transfer probabilistic stability/safety guarantees across a policy update by using adversarially robust CP and by explicitly quantifying the worst-case distribution shift. Conceptually, our work is also related to the notion of performative prediction in supervised learning \cite{hardt2023performativepastfuture}, but our problem is grounded in a controls setting that requires stability/safety guarantees.  We further note that we view recent work on conformal policy control \cite{prinster2026cpc} as adjacent in spirit as it studies how much a policy update changes a distribution while satisfying control constraints.

}
\section{Problem formulation}
\label{sec:problem}

We study CLF/CBF policies when an approximate model of the unknown system dynamics is available. The main text focuses on the continuous-time setting, while discrete-time analogues are provided in the \subonly{extended version~\cite{our_extended}.}%
\fulonly{Appendix.}
\subonly{The extended version also includes the omitted proofs, additional simulations, and larger figures.}
Let $x(t)\in\X$ and $u(t)\in\U$ denote  state and control inputs of the system at time $t\ge 0$, where $\U\subseteq\R^{n_u}$ are input constraints. The  system dynamics are described by
\begin{equation}
  \dot x(t) = f(x(t),u(t)) = \hat f(x(t),u(t)) + \res(x(t),u(t)),
  \label{eq:ct_true}
\end{equation}
where $f:\X\times\U\to \X$ is an unknown continuous function, while $\hat f:\X\times\U\to \X$ is a known nominal model (e.g., learned from data) that induces the unknown model-error $\res(x,u) := f(x,u) - \hat f(x,u)$. A state-feedback policy is a measurable map $\pi:\X\to\U$, and we write $u(t):=\pi(x(t))$ for brevity. For a fixed time horizon $T>0$ and an initial condition $x(0)\in \X$, let $x(0:T):=\{x(t)\}_{t=0}^T$ and $u(0:T):= \{u(t)\}_{t=0}^T$ denote the corresponding system and input trajectories. For convenience, let us also define the joint state-input trajectory $\tau := \bigl(x(0:T),u(0:T)\bigr)$.

\subsection{Robust Control Barrier and Control Lyapunov Functions}
\label{sec:robust_clf_cbf}

A well-known approach to account for unknown model-errors is to enforce \emph{robustified versions} of standard CLF/CBF inequalities. These account for the model error via a robustness margin $r\ge 0$ which should capture a bound on the error $\|\res(x,u)\|\le r$, where $\|\cdot\|$ denotes the Euclidean norm.

\begin{definition}[Robust CBF]
\label{def:cbf}
Let $\Kset\subseteq\X$ be an open set and  $h:\X\to\R$ be a continuously differentiable function. Define the  set $\Cset:=\{x\in\X:h(x)\ge 0\}$ and assume that $\Cset\subseteq\Kset$. We say that $h(x)$ is a \emph{robust control barrier function (rCBF) on $\Kset$} with decay rate $\gamma>0$ and robustness margin $r\ge 0$ if, for every $x\in\Kset$, there exists  $u\in\U$ satisfying
\begin{equation}
  \ip{\nabla h(x)}{\hat f(x,u)} + \gamma h(x) \ge \|\nabla h(x)\|\, r.
  \label{eq:ct_cbf}
\end{equation}
\end{definition}

\begin{definition}[Robust CLF]
\label{def:clf}
Let $\Kset\subseteq\X$ be an open set, $V:\X\to\R_{\ge 0}$ be a continuously differentiable and positive-definite function\fulonly{\footnote{$V(\cdot)$ is positive definite if $V(0)=0$ and $V(x)>0$ for $x\neq 0$.}}, $V_{\max}>0$ be a constant, and  $\underline\alpha_V(\cdot),\overline\alpha_V:\R\to \R$ be locally Lipschitz continuous class-$\mathcal{K}_\infty$ functions\fulonly{\footnote{$\alpha_V(\cdot)$ is a class-$\mathcal{K}_\infty$ function if it is strictly increasing with $\alpha_V(0)=0$.}} such that
$\underline\alpha_V(\|x\|)\le V(x)\le\overline\alpha_V(\|x\|)$
for all $x\in\Kset$. Define the set $\mathcal{V}:=\{x\in\X:V(x)\le V_{\max}\}$  and assume that $\mathcal{V}\subseteq\Kset$. We say that $V(x)$ is a \emph{robust control Lyapunov function (rCLF) on $\Kset$} with decay rate $c>0$ and robustness margin $r\ge 0$ if, for every $x\in\Kset$, there exists $u\in\U$ satisfying
\begin{equation}
  \ip{\nabla V(x)}{\hat f(x,u)} + c V(x) \le -\|\nabla V(x)\|\, r.
  \label{eq:ct_clf}
\end{equation}
\end{definition}

For any $(x,u)\in\Kset\times\U$ satisfying $\|\res(x,u)\|\le r$, the Cauchy--Schwarz inequality yields $-\ip{\nabla h(x)}{\res(x,u)} \le \|\nabla h(x)\|\,r$ and $\ip{\nabla V(x)}{\res(x,u)} \le  \|\nabla V(x)\|\,r$. Hence, if \eqref{eq:ct_cbf} and \eqref{eq:ct_clf} hold, then the true dynamics \eqref{eq:ct_true} satisfy the nominal conditions $\dot{h}(x):=\ip{\nabla h(x)}{ f(x,u)}  \ge -\gamma h(x)$ and $\dot{V}(x):=\ip{\nabla V(x)}{ f(x,u)}  \le -c V(x)$, respectively.

Hence, one usually assumes that $\|\res(x,u)\|\le r$ holds for all $(x,u)\in \Kset \times \U$, but the trajectory-level error bound
\begin{equation}
  \sup_{t\in[0,T]}\|\res(x(t),u(t))\| \le r
  \label{eq:traj_residual_bound}
\end{equation}
actually suffices to certify stability and safety of a solution $x(t)$. In both cases, it can be shown that: (1) $x(t)\in\mathcal{C}$ for all $t\in[0,T]$ if $\mathcal{C}$ is compact and $x(0)\in\mathcal{C}$, and (2)  $\|x(t)\|\le\underline\alpha_V^{-1}\bigl(e^{-ct}\overline\alpha_V(\|x(0)\|)\bigr)$ for all $t\in[0,T]$ if  $x(0)\in\mathcal{V}$. See \subonly{Appendix~\ref{app:det} of~\cite{our_extended} for more details}%
\fulonly{Appendix~\ref{app:det} for more details}.

CBF/CLF policies are  implemented via convex quadratic programs (QPs) for which we consider control-affine nominal models $\hat f(x,u):=\hat f_0(x)+\hat g(x)u$, which render the constraints in \eqref{eq:ct_cbf} and \eqref{eq:ct_clf} affine in $u$. For safety, and given nominal input $u_{\mathrm{nom}}(x)$, the robust CBF-QP policy is:
\begin{equation}
\begin{aligned}
  \pi^{\mathrm{cbf}}_r(x)
  \in\argmin_{u\in\U}\;& \tfrac12\|u-u_{\mathrm{nom}}(x)\|^2\;
  \text{s.t. } \eqref{eq:ct_cbf}.
\end{aligned}
\label{eq:cbf_qp}
\end{equation}
For stability, the robust CLF-QP policy is:
\begin{equation}
\begin{aligned}
  \pi^{\mathrm{clf}}_r(x)
  \in\argmin_{u\in\U}\;& \tfrac12\|u\|^2\;
  \text{s.t. } \eqref{eq:ct_clf}.
\end{aligned}
\label{eq:clf_qp}
\end{equation}
When  safety is desired we set $\pi_r(x):=\pi^{\mathrm{cbf}}_r(x)$, and when  stability is desired we set $\pi_r(x):=\pi^{\mathrm{clf}}_r(x)$. We assume that the policy $\pi_r(x)$ is continuous in $x$; conditions guaranteeing continuity of $\pi_r(x)$ are presented in \cite{mestres2025regularity}.

The certificates above require choosing a margin $r$ that satisfies \eqref{eq:traj_residual_bound}, yet the residual $\res$ is unknown. Classical robust designs require a known deterministic bound on $\|\res(x,u)\|$ over the domain $\Kset\times\U$ \cite{freeman2008robust,xu2015robustness}.  Data-driven estimates of $r$ are either heuristic or  under-approximate the model-error.

\subsection{Conformal rCBF and rCLF  Induce Distribution Shifts}
\label{sec:baseline_cp}


 Ideally, we would like to ensure that \eqref{eq:traj_residual_bound} holds for all trajectories $x(t)$ with initial state such that $x(0)\in\mathcal{X}$. As this is difficult to achieve, we instead sample initial states from a set $\XX\subseteq \mathcal{X}$. Therefore, let $\mu_{\XX}$ be a probability distribution with support over
$\XX$, and sample initial conditions as
$x(0)\sim\mu_{\XX}$. {For any set $\mathcal A\subseteq\X$, the condition $\mu_{\XX}(\mathcal A)=1$ means that an initial condition $x(0)$ drawn from $\mu_{\XX}$ lies in $\mathcal A$ almost surely.} In theory, one can freely choose $\XX$ (e.g., $\XX=\mathcal{C}$ or $\XX=\mathcal{V}$) and $\mu_{\XX}$ (e.g., a uniform distribution), whereas in practice both $\XX$ and $\mu_{\XX}$ are implicitly determined by the available data.  Since the dynamics \eqref{eq:ct_true} are deterministic, a realization of $x(0)$ uniquely determines the 
trajectory $\tau$; consequently, all randomness originates from
$x(0)$. We further limit our attention to trajectories over the time interval $[0,T]$ and write $\tau\sim\mathcal{D}_\pi$ for the induced trajectory distribution.  In this way, we will be able to obtain \emph{probabilistic} safety and stability certificates over the set of
initial conditions $\XX$, i.e., we will be able to guarantee $\Prob \big(x(t)\in\Cset,\ \forall t\in[0,T]\big)\ge 1-\alpha$ and $\Prob \big(\|x(t)\|\le\underline\alpha_V^{-1}\bigl(e^{-ct}\overline\alpha_V(\|x(0)\|)\bigr),\ \forall t\in[0,T]\big)\ge 1-\alpha$ for a miscoverage level $\alpha\in(0,1)$.

Conformal robustness, as presented in \cite{hsu2025}, uses conformal prediction to obtain a probabilistic bound for the trajectory-level error $\res(x(t),u(t))$ in \eqref{eq:traj_residual_bound}. To apply conformal prediction, let
$x^{(1)}(0),\dots,x^{(n)}(0)\sim \mu_{\XX}$ be $n$ independent
and identically distributed (i.i.d.) random variables, often referred to as the calibration data. Given a generic policy $\pi$, we denote the calibration trajectories that follow from the system \eqref{eq:ct_true} under the policy $u(t)=\pi(x(t))$ by
$\tau^{(i)}:=\bigl(x^{(i)}(0:T),u^{(i)}(0:T)\bigr)$ for $i=1,\dots,n$.   Based on this, we define the continuous-time nonconformity score
\begin{equation}
  s^{\ct}(\tau) := \sup_{t\in[0,T]} \|\res(x(t),u(t))\|
  \label{eq:score_ct}
\end{equation}
and let $s^{(i)}:=s^{\ct}(\tau^{(i)})$ for  calibration trajectories $i=1,\dots,n$ and
$s:=s^{\ct}(\tau)$ for the test
trajectory.
{To have well-defined nonconformity scores, we have to assume that the policy $\pi$ ensures that $\tau,\tau^{(1)},\hdots,\tau^{(n)}$ are defined over $[0,T]$, i.e., $\tau(t),\tau^{(1)}(t),\hdots,\tau^{(n)}(t)$ exist for all $t\in[0,T]$.} We note that the supremum operator in \eqref{eq:score_ct} is measurable
under mild regularity conditions, see%
\subonly{ Appendix~\ref{app:ct_meas} of~\cite{our_extended}}%
\fulonly{ Appendix~\ref{app:ct_meas}}. Since the calibration and test trajectories are i.i.d. random variables following $\mathcal{D}_{\pi}$, the
associated nonconformity scores are also i.i.d. following some induced distribution $\mathcal{S}_{\pi}$. In practice, one typically obtains sampled  data, making the computation of the continuous-time  score $s^{\ct}(\tau)$ difficult. \subonly{Appendix~\ref{app:ct_meas} of~\cite{our_extended}}%
\fulonly{Appendix~\ref{app:ct_meas}} provides a method to compute an upper bound in  this case. 

Split conformal prediction \cite{vovk2005alrw,lindemann2025formal} constructs a probabilistic upper bound of $s$ from the calibration data
$D^{\mathrm{cal}}:=\{s^{(1)},\dots,s^{(n)}\}$. For a
miscoverage level $\alpha\in(0,1)$, the split conformal threshold is
$s^{[k]}$ with
$k:=\lceil(1-\alpha)(n+1)\rceil$, where $s^{[k]}$ denotes the $k$th
order statistic of $\{s^{(1)},\dots,s^{(n)},+\infty\}$.\footnote{The $k$th order statistic $s^{[k]}$ of $\{s^{(1)},\dots,s^{(n)},+\infty\}$ is equivalent to the $k$th smallest value of $\{s^{(1)},\dots,s^{(n)},\infty\}$.} The next validity guarantee follows standard arguments from
\cite{vovk2005alrw,angelopoulos2023conformal}; see%
\subonly{ Appendix~\ref{app:arcp} of~\cite{our_extended}}%
\fulonly{ Appendix~\ref{app:arcp}} for a proof.
\begin{lemma}[Split conformal prediction {\cite{vovk2005alrw}}]
\label{thm:split_conformal_validity} Given the i.i.d. random trajectories
 $\tau,\tau^{(1)},\dots,\tau^{(n)}\sim \mathcal D_{\pi}$ under a policy $\pi$, the induced nonconformity scores $s,s^{(1)},\dots,s^{(n)}\sim \mathcal{S}_{\pi}$, and  the
miscoverage level $\alpha\in(0,1)$. Then, it holds that
$\Prob_{n+1}\bigl(s \le  s^{[k]}\bigr) \ge 1-\alpha$ with $k:=\lceil(1-\alpha)(n+1)\rceil$.\footnote{We note that $\Prob_{n+1}(\cdot)$ is the $(n+1)$-fold product probability measure of $\Prob(\cdot)$, which is known to approximate the calibration-conditional probability measure $\Prob(\cdot|D^{\mathrm{cal}})$. Indeed, it is known that $\Prob_{n+1}(\cdot)$ approximates  $\Prob(\cdot|D^{\mathrm{cal}})$ with increasing accuracy as $n$ increases \cite{angelopoulos2023conformal,lindemann2025formal}. We later use a variant of conformal prediction that more directly captures $\Prob(\cdot|D^{\mathrm{cal}})$.}
\end{lemma}
By construction, the event $\{s\le r\}$ is equivalent to the condition in
\eqref{eq:traj_residual_bound}. Therefore, one could think that a high-probability bound on $s$, as obtained in \Cref{thm:split_conformal_validity}, could yield probabilistic stability and safety certificates. Motivated by this observation,  conformal robustness from \cite{hsu2025} sets $r:=s^{[k]}$ and
enforces the CLF and CBF constraints in \eqref{eq:ct_cbf} and \eqref{eq:ct_clf} with this margin. The
next result combines the conformal validity of
\Cref{thm:split_conformal_validity} with CLF and CBF constraints in \eqref{eq:ct_cbf} and \eqref{eq:ct_clf} to yield
safety and stability certificates under some rather restrictive assumptions, which we discuss thereafter. 
\begin{lemma}[Conformal control certificates for a fixed policy]
\label{thm:baseline_cr}
Given the i.i.d. random trajectories
 $\tau,\tau^{(1)},\dots,\tau^{(n)}\sim \mathcal D_{\pi}$ under a policy $\pi$, the miscoverage level $\alpha\in(0,1)$, and the robustness margin
$r:=s^{[k]}$ with $k:=\lceil(1-\alpha)(n+1)\rceil$.

 \emph{(Safety).} Let $h(x)$ be a robust CBF on $\Kset$ with decay rate $\gamma$ and margin
$r$. Assume that $\mathcal{C}$ is compact. Furthermore, let $\pi(x)$ be a continuous function that enforces the CBF constraint \eqref{eq:ct_cbf} for all $x\in\Kset$. 
{If $\mu_{\XX}(\Cset)=1$}, then
\begin{extendedonly}
\begin{align*}
  \Prob_{n+1}(x(t)\in\Cset,\ \forall t\in[0,T])\ge 1-\alpha.
\end{align*}
\end{extendedonly}
\begin{arxivonly}
\begin{equation*}
  \Prob_{n+1}(x(t)\in\Cset,\ \forall t\in[0,T])
  \ge 1-\alpha.
\end{equation*}
\end{arxivonly}
\begin{submitonly}
$\Prob_{n+1}(x(t)\in\Cset,\ \forall t\in[0,T])\ge 1-\alpha$.
\end{submitonly}

\emph{(Stability).}
Let $V(x)$ be a robust CLF on $\Kset$  with decay rate $c$ and margin
$r$. Assume that $\mathcal{V}$ is compact. Furthermore, let $\pi(x)$ be a continuous function that enforces the CLF constraint \eqref{eq:ct_clf} for all $x\in\Kset$.
{If $\mu_{\XX}(\mathcal{V})=1$}, then
\begin{extendedonly}
\begin{align*}
  \Prob_{n+1}(\|x(t)\|\le\underline\alpha_V^{-1}\bigl(e^{-ct}\overline\alpha_V(\|x(0)\|)\bigr),\  \forall t\in[0,T])\ge 1-\alpha.
\end{align*}
\end{extendedonly}
\begin{arxivonly}
\begin{equation*}
  \Prob_{n+1}\bigl(\|x(t)\|\le\underline\alpha_V^{-1}(e^{-ct}\overline\alpha_V(\|x(0)\|)),
  \forall t\in[0,T]\bigr)\ge 1-\alpha.
\end{equation*}
\end{arxivonly}
\begin{submitonly}
$\Prob_{n+1}(\|x(t)\|\le\underline\alpha_V^{-1}(e^{-ct}\overline\alpha_V(\|x(0)\|)),\ \forall t\in[0,T])\ge 1-\alpha$.
\end{submitonly}
\end{lemma}
\Cref{thm:baseline_cr} first appeared in \cite{hsu2025} and is here presented in slightly different form. 
A proof is provided in%
\subonly{ Appendix~\ref{app:event} of~\cite{our_extended}}%
\fulonly{ Appendix~\ref{app:event}}. 
Notably, \Cref{thm:baseline_cr} assumes that the calibration and test trajectories are i.i.d. random trajectories, which  can in practice only be achieved when the policy $\pi$ is fixed a-priori so that calibration and test trajectories follow the same distribution $\mathcal{D}_{\pi}$. This is difficult to achieve as it requires $\pi$
to enforce the CBF constraint \eqref{eq:ct_cbf} (or the CLF constraint \eqref{eq:ct_clf}) with a margin
$r:=s^{[k]}$. The issue is circular: the policy  \(\pi\)  depends on the margin \(r\), while the margin $r$ itself depends on the trajectory distribution $\mathcal{D}_{\pi}$ and thereby also on the policy \(\pi\).
Additionally, fixing the policy \(\pi\) a priori defeats the purpose, as it prevents policy synthesis via the CBF-QP \(\pi_r^{\mathrm{cbf}}(x)\) in \eqref{eq:cbf_qp} and the CLF-QP \(\pi_r^{\mathrm{clf}}(x)\) in \eqref{eq:clf_qp}. Replacing \(r\) in \eqref{eq:cbf_qp} (or \eqref{eq:clf_qp}) by the conformal threshold $s^{[k]}$ therefore violates the fixed-policy assumption in \Cref{thm:baseline_cr}. The next example, for which we provide detailed derivations in%
\subonly{ Appendix~\ref{app:counterexample} of~\cite{our_extended}}%
\fulonly{ Appendix~\ref{app:counterexample}}, shows that this issue is not merely theoretical, but that a CBF-QP $\pi_r^{\mathrm{cbf}}(x)$ that uses the robustness margin
$r:=s^{[k]}$ can actually violate safety guarantees.
\begin{example}\label{prop:counterexample_main}
\begin{fullonly}
Consider system~\eqref{eq:ct_true} with model $\hat{f}(x(t),u(t)):=u(t)$ and model error $\varepsilon(x(t),u(t)):=-(2+x(t))u(t)$, so that $\dot x(t)=-(1+x(t))u(t)$. Let the safe set $\Cset:=\{h(x)\ge 0\}$ be defined by $h(x)=x$, $\U:=\R$, and $x(0)\sim\mu_{\XX}:=\mathrm{Unif}[0,1]$. Fix $T>0$, $\alpha\in(0,1)$, and consider the calibration policy $\pi(x):=-u_0$ with $u_0>0$. Under this calibration policy, we have $\Prob_{\tau\sim\mathcal D_\pi}\bigl(x(t)\in\Cset,\ \forall t\in[0,T]\bigr)=1$, i.e., the system is safe almost surely. Let now $r$ denote the $(1-\alpha)$-quantile of $s^{\ct}(\tau)$ under $\tau\sim\mathcal D_\pi$. Indeed, we obtain $r=u_0\bigl(1+(2-\alpha)e^{u_0T}\bigr)$, see \Cref{app:counterexample} for detailed derivations. If $r$ is now used as the robustness margin in the CBF-QP~\eqref{eq:cbf_qp} with nominal input $u_{\mathrm{nom}}(x)\equiv -u_0$, then the synthesized controller is $\pi_r^{\mathrm{cbf}}(x)=\max\{-u_0,\ r-\gamma x\}$. Moreover, if $0<\gamma<r/2$, then
$\Prob_{\tau\sim\mathcal D_{\pi_r^{\mathrm{cbf}}}}\bigl(s^{\ct}(\tau)\le r\bigr)=0$.
If, in addition, $T(r-\gamma)\ge 1$, then $\Prob_{\tau\sim\mathcal D_{\pi_r^{\mathrm{cbf}}}}\bigl(x(t)\in\Cset,\ \forall t\in[0,T]\bigr)=0$, i.e., the controller renders the system unsafe almost surely, as illustrated in \Cref{fig:prop1}.
\end{fullonly}
\begin{submitonly}
Consider system~\eqref{eq:ct_true} with $\hat{f}(x,u):=u$ and $\varepsilon(x,u):=-(2+x)u$, so that $\dot x(t)=-(1+x(t))u(t)$. Let $h(x)=x$, $\Cset:=\{x:h(x)\ge 0\}$, $\U:=\R$, and $x(0)\sim\mathrm{Unif}[0,1]$. Fix $T>0$, $\alpha\in(0,1)$, and a calibration policy $\pi(x):=-u_0$ with $u_0>0$, under which the system is safe almost surely. Compute the analytical $(1-\alpha)$-quantile of $s^{\ct}(\tau)$  as $r=u_0(1+(2-\alpha)e^{u_0T})$ and use $r$ in the CBF-QP~\eqref{eq:cbf_qp} with $u_{\mathrm{nom}}(x)\equiv -u_0$, yielding $\pi_r^{\mathrm{cbf}}(x)=\max\{-u_0,\,r-\gamma x\}$. Detailed derivations are provided in Appendix~\ref{app:counterexample} of the extended version~\cite{our_extended}. If $0<\gamma<r/2$ and $T(r-\gamma)\ge 1$, then $\Prob_{\tau\sim\mathcal{D}_{\pi_r^{\mathrm{cbf}}}}(x(t)\in\Cset,\,\forall\,t\in[0,T])=0$, i.e., the controller renders the system unsafe almost surely, see \Cref{fig:prop1}. 
\end{submitonly}
 \begin{arxivonly}
    \vspace{-.5cm}
\end{arxivonly}
\end{example}
\begin{submitonly}
    \vspace{-0.4cm}
\end{submitonly}
\begin{compactonly}
    \begin{figure}[htbp!]
        \centering        \captionsetup{font=small,skip=2pt}
        \begin{subfigure}[t]{0.52\linewidth}           \includegraphics[width=\linewidth]{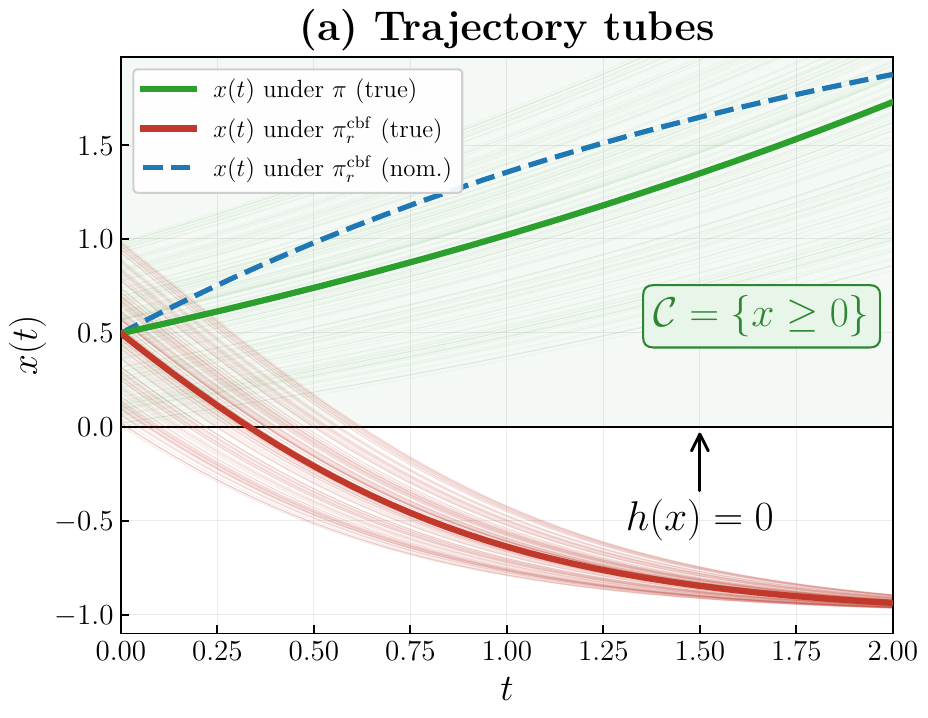}
        \end{subfigure}\hfill
        \begin{subfigure}[t]{0.47\linewidth}
\includegraphics[width=\linewidth]{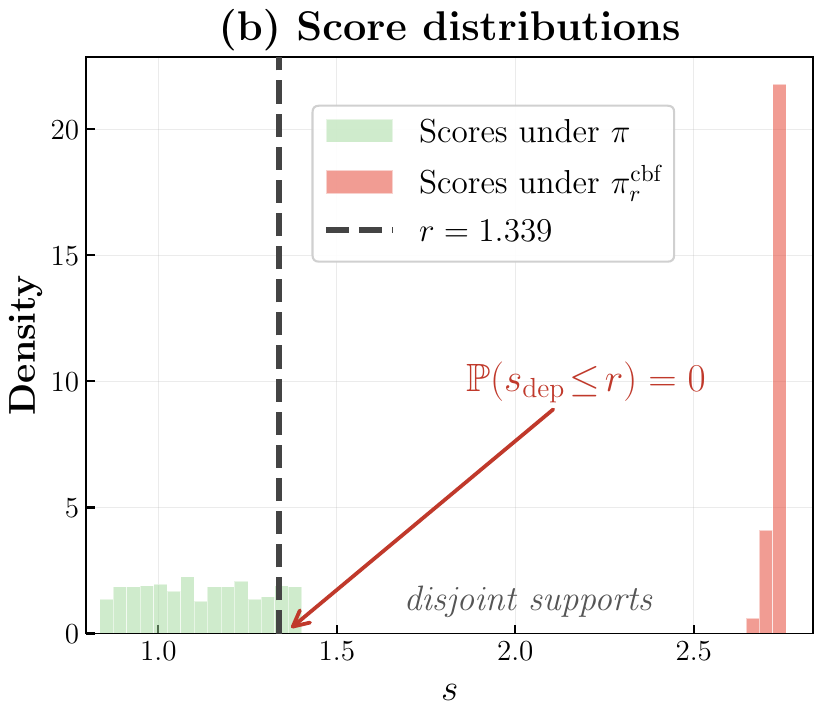}
        \end{subfigure}       \caption{Example~\ref{prop:counterexample_main} with $u_0=0.3$, $T=2$, $\alpha=0.1$, $\gamma=0.5$. \textbf{(a)}~Trajectories under $\pi$ (green, safe) and $\pi_r^{\mathrm{cbf}}$ (red, unsafe). \textbf{(b)}~Calibration and deployment scores have disjoint support. }
        \label{fig:prop1}
        \begin{submitonly}
    \vspace{-.5cm}
\end{submitonly}
    \end{figure}
     \begin{arxivonly}
    \vspace{-.6cm}
\end{arxivonly}
\end{compactonly}

\begin{extendedonly}
    \begin{figure}[htbp!]
        \centering      \captionsetup{font=small,skip=2pt}
        \begin{subfigure}[t]{0.34\linewidth}
    \includegraphics[width=\linewidth]{figures/final/ex1_trajectories.pdf}
        \end{subfigure}\hfill
        \begin{subfigure}[t]{0.31\linewidth}
    \includegraphics[width=\linewidth]{figures/final/ex1_scores.pdf}
        \end{subfigure} \hfill 
        \begin{subfigure}[t]{0.34\linewidth}
     \includegraphics[width=\linewidth]{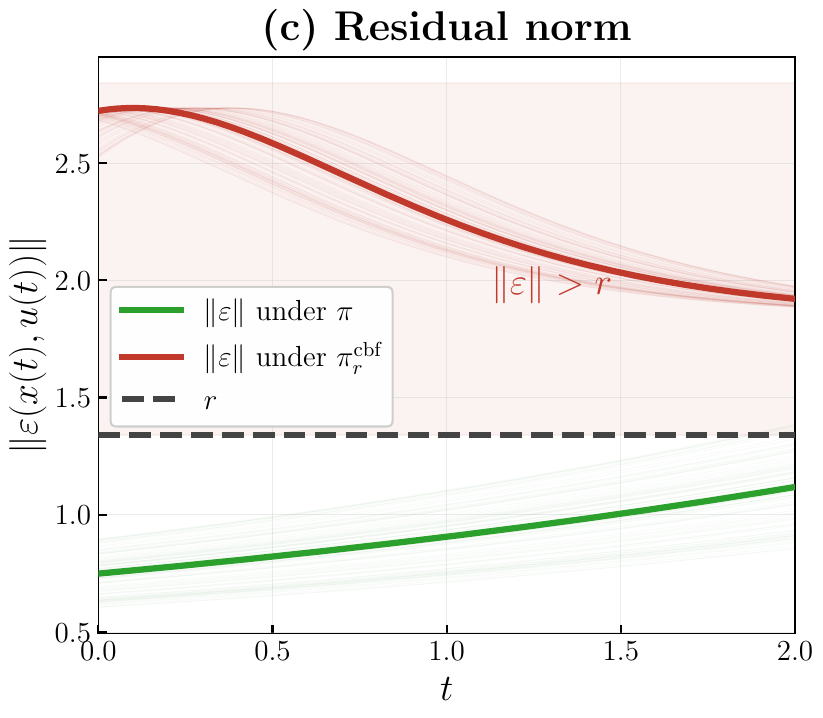}
        \end{subfigure}
\caption{Example~\ref{prop:counterexample_main} with $u_0=0.3$, $T=2$, $\alpha=0.1$, $\gamma=0.5$. \textbf{(a)}~Trajectories under $\pi$ (green, safe) and $\pi_r^{\mathrm{cbf}}$ (red, unsafe). \textbf{(b)}~Calibration and deployment scores have disjoint support. \textbf{(c)}~Residual norm under $\pi_r^{\mathrm{cbf}}$ exceeds $r$ from $t=0$.}
        \label{fig:prop1}
    \end{figure}
\end{extendedonly}
\subsection{Adversarially Robust Conformal Prediction}
\label{sec:arcp}

 To transfer control certificates beyond the  setting in \Cref{thm:baseline_cr}, we
use adversarially robust conformal prediction (ARCP)
\cite{gendler2021arcp}. ARCP provides guarantees under
bounded perturbations of the nonconformity score; here, the ``perturbation'' is the change in
nonconformity score induced by the policy. ARCP does so by increasing the conformal threshold by a nonnegative, possibly calibration-data-dependent, perturbation term $M$ that upper-bounds the change in the nonconformity score. We present a calibration-conditional version of ARCP, relying on results from calibration-conditional CP 
\cite{vovk12conditional,duchi2025sample}, summarized in
\cite[Lemma 2]{lindemann2025formal}.

\begin{lemma}[Adversarially robust conformal prediction
  {\cite{gendler2021arcp,duchi2025sample}}]
\label{lem:arcp}
Let $s,s^{(1)},\dots,s^{(n)}\sim\mathcal{S}$ be i.i.d. random variables,  $D^{\mathrm{cal}}:=\{s^{(i)}\}_{i=1}^n$,  $\alpha,\delta\in(0,1)$, and 
$\bar\alpha:=\alpha-\sqrt{\ln(1/\delta)/(2n)}>0$. Let $M(D^{\mathrm{cal}})\ge 0$ be a non-negative function of the calibration data.  Suppose that $\tilde s$ is a random variable such that $\tilde s\le s+M(D^{\mathrm{cal}})$ almost surely. Then, with $k:=\lceil(1-\bar\alpha)n\rceil$, it holds that
\begin{equation}\label{eq:ARCP}
\Prob_{n}\Bigl\{
    \Prob\bigl(
      \tilde s \le s^{[k]}+M(D^{\mathrm{cal}})
    \mid D^{\mathrm{cal}}\bigr)
    \ge 1-\alpha
  \Bigr\}\ge 1-\delta .
  \footnote{We note that $\Prob_{n}\{\cdot\}$ is the $n$-fold product probability measure of $\Prob(\cdot)$, while $\Prob(\cdot \mid D^{\mathrm{cal}})$ is the probability measure of $\Prob(\cdot)$ conditioned on $D^{\mathrm{cal}}$.} 
\end{equation}
\end{lemma}

ARCP, as originally presented in \cite{gendler2021arcp}, gives marginal guarantees of the form $\Prob_{n+1}(\tilde s\le s^{[k]}+M(D^{\mathrm{cal}}))\ge1-\alpha$. In contrast, \Cref{lem:arcp} gives calibration-conditional guarantees of the form \eqref{eq:ARCP}.
These guarantees are conditional on $D^{\mathrm{cal}}$ and in the form needed for our convergence analysis. The proof, given \fulonly{in Appendix~\ref{app:arcp}}\subonly{in Appendix~\ref{app:arcp} of \cite{our_extended}}, uses  conditional split conformal prediction from \cite{duchi2025sample}, as summarized in \cite[Lemma~2]{lindemann2025formal}, and then uses an  event inclusion induced by the bound $\tilde s\le s+M(D^{\mathrm{cal}})$. Accordingly, the order statistic is taken at the sample-conditional level $k=\lceil(1-\bar\alpha)n\rceil$, rather than at the marginal split conformal level $\lceil(1-\alpha)(n+1)\rceil$. We note that \Cref{thm:split_conformal_validity,thm:baseline_cr} could also be stated in calibration-conditional terms, but this alone would not address the circularity and distribution-shift issues discussed above.

\section{Conformal rCBF and rCLF  via Iterative Policy Updates}
\label{sec:algorithm}


To address the aforementioned circular dependency between the policy $\pi$ and the margin $r$, we reframe the problem into an iterative policy update framework. We start by considering an episode index
$j\in \{0\}\cup \N$. At episode $j$, a robustness margin $r_j\ge 0$ is used --- important details on the construction of $r_j$  follow below --- to
synthesize the policy $\pi_j:=\pi_{r_j}$ via the CBF-QP \eqref{eq:cbf_qp} or
the CLF-QP \eqref{eq:clf_qp}.  In each episode, we then sample
$n_j$ initial states 
$x_j(0),x_j^{(1)}(0),\dots,x_j^{(n_j)}(0)\stackrel{\mathrm{i.i.d.}}{\sim}
\mu_{\XX}$ from which we collect the  state-input 
trajectories $\tau_j,\tau_j^{(1)},\hdots, \tau_j^{(n_j)}\stackrel{\mathrm{i.i.d.}}{\sim}
\mathcal{D}_{\pi_j}$, where $\mathcal{D}_{\pi_j}$ denotes again the trajectory distribution under policy $\pi_j$.  We hence make the following standing assumption.

\begin{assumption}[]
\label{ass:within_episode_iid}
At each episode $j\in \{0\}\cup \N$,  the  initial conditions and the associated state-input trajectories are such that
$x_j(0),x_j^{(1)}(0),\dots,x_j^{(n_j)}(0)\stackrel{\mathrm{i.i.d.}}{\sim}\mu_{\XX}$ and  $\tau_j,\tau_j^{(1)},\hdots, \tau_j^{(n_j)}\stackrel{\mathrm{i.i.d.}}{\sim}
\mathcal{D}_{\pi_j}$, respectively. The  state-input trajectories $\tau_j,\tau_j^{(1)},\ldots,\tau_j^{(n_j)}$ are defined over the time interval $[0,T]$, i.e., $\tau_j(t),\tau_j^{(1)}(t),\ldots,\tau_j^{(n_j)}(t)$ exist for all $t\in[0,T]$.
\end{assumption}
Following the ARCP methodology from \Cref{sec:arcp},  we define the test nonconformity score $s_j := s^{\ct}(\tau_{j})\sim \mathcal{S}_{\pi_j}$ and the calibration nonconformity scores $s_j^{(i)} := s^{\ct}(\tau_{j}^{(i)})\sim \mathcal{S}_{\pi_j}$, where $\mathcal{S}_{\pi_j}$ denotes the nonconformity score distributions. We then define the calibration dataset $D_j^{\mathrm{cal}} := \{s_j^{(i)}\}_{i=1}^{n_j}$ and compute the conformal threshold $q_j := s_j^{[k_j]}$, where $s_j^{[k_j]}$ is the $k_j$th order statistic of the finite set $D_j^{\mathrm{cal}}$ and $k_j:=\lceil(1-\bar{\alpha}_j)n_j\rceil$ with $\bar\alpha_j
  := \alpha - \sqrt{\ln(1/\delta)/(2n_j)}$.

At episode $j=0$, we assume to be given a margin $r_0$ along with the corresponding policy $\pi_0$, which can but is not required to result in safety or stability initially. The main problem that we address in this paper is a framework for iteratively updating $r_{j+1}$
\begin{extendedonly}
such that
\begin{align*}
&\textbf{Episodic validity:}\;\Prob_{n_j}\Bigl\{
    \Prob\bigl(
      s_{j+1} \le r_{j+1}\mid D_j^{\mathrm{cal}}\bigr)
    \ge 1-\alpha
  \Bigr\}\ge 1-\delta,\quad \forall j\in\N, \\
&\textbf{Minimized margins:}\; r_{j+1}\text{ is minimized for all } j\in\N, \\
&\textbf{Convergence of margins:}
\; r_{j+1}\to r^*\text{ eventually for $r^*>0$.}
\end{align*}
\end{extendedonly}
\begin{arxivonly}
such that
\begin{align*}
&\textbf{Episodic validity:}\;\Prob_{n_j}\Bigl\{
    \Prob\bigl(
      s_{j+1} \le r_{j+1}\mid D_j^{\mathrm{cal}}\bigr)
    \ge 1-\alpha
  \Bigr\}\\[-0.2em]
&\qquad \qquad\qquad\qquad\qquad\qquad\qquad\ge 1-\delta,\quad \forall j\in\N, \\
&\textbf{Minimized margins:}\; r_{j+1}\text{ is minimized for all } j\in\N, \\
&\textbf{Convergence of margins:}
\; r_{j+1}\to r^*\text{ eventually for $r^*>0$.}
\end{align*}
\end{arxivonly}
\begin{submitonly}
to achieve: (i)~\emph{episodic validity}: $\Prob_{n_j}\bigl\{
    \Prob(
      s_{j+1} \le r_{j+1}\mid D_j^{\mathrm{cal}})
    \ge 1-\alpha
  \bigr\}\ge 1-\delta$  for all $j\in\N$; (ii)~\emph{minimized margins}: $r_{j+1}$ is as small as possible for all $j\in\N$; and (iii)~\emph{convergence}: $r_{j+1}\to r^*$ for some $r^*>0$.
\end{submitonly}

  When per-episode validity holds, we can achieve episodic safety/stability with high probability, similarly to \Cref{thm:baseline_cr}.


\subsection{Bounding Policy-Induced Distribution Shifts}
\label{sec:policy_shift}

The conformal threshold $q_j$, obtained from calibration data under the policy $\pi_j$, is valid only for  $\pi_j$  in the current episode  so that $\Prob_{n_j}\Bigl\{
    \Prob\bigl(
       s_j \le q_j
    \mid D_j^{\mathrm{cal}}\bigr)
    \ge 1-\alpha
  \bigr\}\ge 1-\delta$. However, $q_j$  is not directly valid for the
policy $\pi_{j+1}$ in the next episode, because the trajectory distribution changes from
$\mathcal{D}_{\pi_j}$ to $\mathcal{D}_{\pi_{j+1}}$. We now instantiate \Cref{lem:arcp} in this episodic framework  to obtain probabilistic guarantees for $s_{j+1}$ in terms of $q_j$ and an offset $M_{j+1}$ that describes the discrepancy between the distributions $\mathcal{S}_{\pi_j}$ and $\mathcal{S}_{\pi_{j+1}}$. 

It remains to construct $M_{j+1}$.  Recall that the nonconformity score $s^{\ct}(\tau)$ is a function of the state-input trajectory $\tau$, which is  a function of the initial state $x(0)$ and the policy $\pi$. We make this dependency explicit by writing $s^{\ct}(x(0),\pi)$.
\begin{assumption}
\label{ass:domination}
Let the initial condition be $x_0\sim \mu_{\XX}$. For any two policies $\pi'$ and $\pi''$, there exists a nonnegative, bounded function $\rho(\pi',\pi'')$ such that
\begin{equation*}
 s^{\ct}(x_0,\pi') \le s^{\ct}(x_0,\pi'') + \rho(\pi',\pi'')
 \;\;\text{almost surely}.
\end{equation*}
\end{assumption}
A sufficient condition to ensure \Cref{ass:domination} is when the functions $f(x,u)$, $\hat{f}(x,u)$, and $\pi_r(x)$ are Lipschitz continuous (details are provided in \Cref{ass:ct_lip}). The key idea for transferring guarantees from episode $j$ to episode $j+1$ is to use \Cref{lem:arcp} and compare the policies $\pi_{j+1}$ and $\pi_j$ using the constant $M_{j+1}:=\rho(\pi_{j+1},\pi_j)$. As $\pi_j$ will depend on the calibration data $D_j^{\mathrm{cal}}$, the constant $M_{j+1}$ will also depend on $D_j^{\mathrm{cal}}$.  The next result is proven in%
\subonly{ Appendix~\ref{app:event} of~\cite{our_extended}}%
\fulonly{ Appendix~\ref{app:event}}.



\begin{theorem}
\label{thm:srcr_validity}
Let \Cref{ass:within_episode_iid,ass:domination} hold. Given an episode $j\in\{0\}\cup \N$, the induced nonconformity scores $s_j,s^{(1)}_j,\dots,s^{(n_j)}_j\sim \mathcal{S}_{\pi_j}$, and the miscoverage levels $\alpha,\delta\in(0,1)$. If $M_{j+1}:=\rho(\pi_{j+1},\pi_j)$, then 
 \begin{align*}
\Prob_{n_j}\Bigl\{
    \Prob\bigl(
      s_{j+1} \le q_j+M_{j+1}
    \mid D^{\mathrm{cal}}_j\bigr)
    \ge 1-\alpha
  \Bigr\}\ge 1-\delta.
\end{align*}
Consequently, every robustness margin $r_{j+1}$ satisfying $r_{j+1} \ge q_j + M_{j+1}$ guarantees
\begin{equation}
\Prob_{n_j}\Bigl\{
    \Prob\bigl(
      s_{j+1} \le r_{j+1}
    \mid D^{\mathrm{cal}}_j\bigr)
    \ge 1-\alpha
  \Bigr\}\ge 1-\delta.
  \label{eq:srcr_validity}
\end{equation}
\end{theorem}
\subsection{Algorithmically Updating the Robustness Margin}
\label{sec:margin_update}

The choice of $r_{j+1}$ from \Cref{thm:srcr_validity} provides validity guarantees for the policy $\pi_{j+1}$. However, solving the inequality $r_{j+1} \ge q_j + M_{j+1}$ for $r_{j+1}$ is generally challenging due to the dependence of $M_{j+1}=\rho(\pi_{j+1},\pi_j)$ on $\pi_{j+1}$ and consequently also on $r_{j+1}$. This leads to an implicit fixed-point problem since $r_{j+1}$ depends on $\pi_{j+1}$, which itself is a function of $r_{j+1}$. In the remainder, we propose an algorithmic solution for the case where the functions $f(x,u)$, $\hat{f}(x,u)$, and $\pi_r(x)$ are Lipschitz continuous. 

\begin{assumption}
\label{ass:ct_lip}
There exist compact sets $\Rset:=[r_{\min},r_{\max}]\subseteq\R_{\ge 0}$ and $ \Oset\subseteq \Kset$ with $\mathrm{int}(\Oset)\supseteq \XX $ such that, for  $r\in \Rset$ and $x(0)\in\XX$, the solution $x(t)$ to \eqref{eq:ct_true} under $\pi_r(x)$   is such that $x(t)\in\Oset$ for all $t\in[0,T]$. Additionally, the functions $f(x,u)$, $\hat{f}(x,u)$, and $\pi_r(x)$ are Lipschitz continuous on $\Oset\times\U$, i.e., there exist Lipschitz constants $L_x,L_u,L_{\hat{f},x},L_{\hat{f},u},L_\pi\ge 0$ such that:
\begin{fullonly}
\begin{enumerate}[label=(\roman*),leftmargin=2.0em]
  \item 
  $\norm{f(x,u)-f(x',u')} \le L_x\norm{x-x'}+L_u\norm{u-u'}$ for all $(x,u),(x',u')\in \Oset\times\U$,
  \item $\norm{\hat{f}(x,u)-\hat{f}(x',u')} \le L_{\hat{f},x}\norm{x-x'}+L_{\hat{f},u}\norm{u-u'}$ for all $(x,u),(x',u')\in \Oset\times\U$, and
  \item $\norm{\pi_r(x)-\pi_r(x')}\le L_{\pi}\norm{x-x'}$ for all $x,x'\in\Oset$.
\end{enumerate}
\end{fullonly}
\begin{submitonly}
(i)~$\norm{f(x,u)-f(x',u')} \le L_x\norm{x-x'}+L_u\norm{u-u'}$, (ii)~$\norm{\hat{f}(x,u)-\hat{f}(x',u')} \le L_{\hat{f},x}\norm{x-x'}+L_{\hat{f},u}\norm{u-u'}$, and (iii)~$\norm{\pi_r(x)-\pi_r(x')}\le L_{\pi}\norm{x-x'}$, for all $(x,u),(x',u')\in \Oset\times\U$ and $x,x'\in\Oset$.
\end{submitonly}
\end{assumption}

\begin{remark}[Role of $\Oset$]
\label{rem:Omega_X0}
The compact set $\Oset$ contains trajectories $x(t)$, which are generated by the policy $\pi_r$  for any $r\in\Rset$, for all times $t\in[0,T]$.  The reason for introducing $\Rset$ is that, in order to compute $r_{j+1}$, we must first quantify how the nonconformity score changes when the policy changes from $\pi_{r_j}$ to $\pi_{r}$ for arbitrary values of $r\in\Rset$. Consequently, trajectories generated by any policy $\pi_r$ with $r\in\Rset$ remain in $\Oset\subseteq\Kset$ in which the rCBF and rCLF constraints as well as the Lipschitz conditions (i), (ii), and (iii) from \Cref{ass:ct_lip} are evaluated. The set $\Oset$ is chosen to be compact so that the Lipschitz constants are finite; requiring global Lipschitz bounds on all of $\X$ would be unnecessary conservative.
\end{remark}

Note that conditions (i) and (ii) in \Cref{ass:ct_lip} imply that the model error $\res(x,u) = f(x,u) - \hat{f}(x,u)$ is also Lipschitz continuous on $\Oset\times\U$ with constants $L_{\res,x}\le L_x + L_{\hat{f},x}$ and $L_{\res,u}\le L_u + L_{\hat{f},u}$. Additionally, condition (iii) holds for strongly regular quadratic programs, which can be shown via the QP sensitivity analysis provided in%
\subonly{ Appendix~\ref{app:qp} of~\cite{our_extended}}%
\fulonly{ Appendix~\ref{app:qp}}. 

We can now compute a constant $\beta_T>0$ such that, for every $r,r'\in\Rset$ and every initial condition $x(0)\sim\mu_{\XX}$,
\begin{equation}
  \bigl|s^{\ct}(x(0),\pi_{r'})-s^{\ct}(x(0),\pi_r)\bigr|
  \le \beta_T\,\CzNorm{\pi_{r'}-\pi_r}{\Oset},
  \label{eq:betaT_main}
\end{equation}
where $\CzNorm{\pi_{r'}-\pi_r}{\Oset}:=\sup_{x\in\Oset}\norm{\pi_{r'}(x)-\pi_r(x)}$. The constant $\beta_T$ depends on the horizon $T$ and the Lipschitz constants from \Cref{ass:ct_lip}; its derivation via a Gr\"onwall inequality argument is shown in%
\subonly{ Appendix~\ref{app:sensitivity} of~\cite{our_extended}}%
\fulonly{ Appendix~\ref{app:sensitivity}}. The constant $\beta_T$ grows exponentially in $T$, so it can be conservative for long horizons. One may instead estimate $\hat\beta_T\ge\beta_T$ from calibration data~\cite{huang2023sample,fazlyab2019efficient} with confidence $1-\delta_\beta$; since $M_{j+1}$ may depend on $D_j^{\mathrm{cal}}$ (see \Cref{thm:srcr_validity}), this is valid and degrades only the outer confidence from $1-\delta$ to $1-\delta-\delta_\beta$.

\textbf{Implicit solution.} Following~\eqref{eq:betaT_main}, the choice of $\rho(\pi_{j+1},\pi_j):=\beta_T\,\CzNorm{\pi_{j+1}-\pi_j}{\Oset}$ satisfies \Cref{ass:domination}. Motivated by \Cref{thm:srcr_validity}, the smallest valid robustness margin $r_{j+1}$ for the next episode is any minimizer of
\begin{equation}
  r_{j+1}
  \in\argmin\bigl\{r:\ r\ge q_j
  + \beta_T\,\CzNorm{\pi_{r}-\pi_{r_j}}{\Oset}\bigr\}.
  \label{eq:implicit_update}
\end{equation}
As the right-hand side of~\eqref{eq:implicit_update} depends on $r$ through $\CzNorm{\pi_r-\pi_{r_j}}{\Oset}$, which is a continuous function of $r$, the feasible set is determined by $g_j(r):=r-q_j-\beta_T\,\CzNorm{\pi_r-\pi_{r_j}}{\Oset}$ on the compact set $\Rset$. Thus, it can be found by any standard scalar root-finding method, such as grid search.

\textbf{Explicit solution.} We can convert the  inequality~\eqref{eq:implicit_update} to an explicit rule if the policy $\pi_r$ is Lipschitz continuous in $r$. 
\begin{assumption}\label{ass4}
     The function $\pi_r(x)$ is Lipschitz continuous on $\Rset$, i.e., there exists a Lipschitz constant $L_U\ge 0$ such that $\CzNorm{\pi_r-\pi_{r'}}{\Oset}
  \le L_U\,|r-r'|$ for all $r,r'\in\Rset$.
\end{assumption}
In%
\subonly{ Appendix~\ref{app:qp} of~\cite{our_extended},}%
\fulonly{ Appendix~\ref{app:qp},} we provide sufficient conditions on quadratic programs such that Assumption~\ref{ass4} holds for $\pi_r$. Inserting the Lipschitz bound into~\eqref{eq:betaT_main} yields the  bound
\begin{fullonly}
\begin{equation}
  \bigl|s^{\ct}(x(0),\pi_{r'})-s^{\ct}(x(0),\pi_r)\bigr|
  \le \kappaSR\, |r'-r|,
  \label{eq:kappa_score_shift}
\end{equation}
where $\kappaSR:=\beta_TL_U$ (see \subonly{ Appendix~\ref{app:sensitivity} of~\cite{our_extended}}%
\fulonly{ Appendix~\ref{app:sensitivity}}).
Equation~\eqref{eq:kappa_score_shift} now provides an alternative to equation~\eqref{eq:implicit_update} via the update
\end{fullonly}
\begin{submitonly}
\refstepcounter{equation}\label{eq:kappa_score_shift}%
$|s^{\ct}(x(0),\pi_{r'})-s^{\ct}(x(0),\pi_r)|
  \le \kappaSR\, |r'-r|$,
where $\kappaSR:=\beta_TL_U$.
This provides an alternative to~\eqref{eq:implicit_update} via the update
\end{submitonly}
\begin{fullonly}
\begin{equation}
  r_{j+1}\in\argmin\ \bigl\{r:\ r\ge q_j+\kappaSR\,|r-r_j|\bigr\}.
  \label{eq:implicit_solver}
\end{equation}
\end{fullonly}
\begin{submitonly}
\refstepcounter{equation}\label{eq:implicit_solver}%
$r_{j+1}\in\argmin\{r:\ r\ge q_j+\kappaSR\,|r-r_j|\}$.
\end{submitonly}
Every solution of\fulonly{~\eqref{eq:implicit_solver}}\subonly{ this update} satisfies the condition in \Cref{thm:srcr_validity} and therefore provides the same guarantees. When $\kappaSR<1$, the solution $r_{j+1}$\fulonly{ to~\eqref{eq:implicit_solver}} can be computed in closed-form as
\begin{fullonly}
\begin{equation}
  r_{j+1}
  :=
  \begin{cases}
    \dfrac{q_j-\kappaSR r_j}{1-\kappaSR}, & q_j\ge r_j,\\[0.85em]
    \dfrac{q_j+\kappaSR r_j}{1+\kappaSR}, & q_j<r_j.
  \end{cases}
  \label{eq:explicit_update}
\end{equation}
\end{fullonly}
\begin{submitonly}
\begin{equation}
  r_{j+1}
  := \frac{q_j-\kappaSR r_j}{1-\kappaSR}\ \text{if } q_j\ge r_j,
  \quad
  \frac{q_j+\kappaSR r_j}{1+\kappaSR}\ \text{if } q_j<r_j,
  \label{eq:explicit_update}
\end{equation}
\end{submitonly}
see \cite[Lemma 3]{mirzaeedodangeh2025safe}. Whenever we use \eqref{eq:explicit_update}, we assume that $r_{j+1}$ remains in $\Rset$, i.e., $r_{j+1} \in \Rset$ for all $j \in \{0\}\cup\mathbb{N}$.

\section{Theoretical Guarantees and  Convergence}
\label{sec:guarantees}

This section equips the iterative algorithm
from \Cref{sec:algorithm} with safety and stability guarantees. We further analyze convergence of the explicit update rule in \Cref{eq:explicit_update}.

\subsection{Conformal Safety and Stability Certificates}
\label{sec:certificates}


Using the update conditions for $r_{j+1}$ from \Cref{thm:srcr_validity}, we show invariance of $\mathcal{C}$ and stability of $\mathcal{V}$ in
\Cref{thm:ct_cbf_certificate}. The proofs are presented in%
\subonly{ Appendix~\ref{app:event} of~\cite{our_extended}}%
\fulonly{ Appendix~\ref{app:event}}.
\begin{theorem}[Conformal safety and stability certificates]
\label{thm:ct_cbf_certificate}
Let the conditions of \Cref{thm:srcr_validity} hold and $r_{j+1} \ge q_j + M_{j+1}$.

 \emph{(Safety).} Let $h(x)$ be a robust CBF on $\Kset$ with decay rate $\gamma$ and margin
$r_{j+1}$. Assume that $\mathcal{C}$ is compact. Furthermore, let $\pi_{{j+1}}(x)$ be a continuous function that enforces the CBF constraint \eqref{eq:ct_cbf} for all $x\in\Kset$. If $\mu_{\XX}(\Cset)=1$, then
\begin{extendedonly}
\begin{equation*}
  \Prob_{n_j}\Bigl\{
    \Prob\bigl(x_{j+1}(t)\in\Cset\text{ for all } t\in[0,T]
    \mid D_j^{\mathrm{cal}}\bigr)\ge 1-\alpha
  \Bigr\}\ge 1-\delta.
\end{equation*}
\end{extendedonly}
\begin{arxivonly}
\begin{multline*}
  \Prob_{n_j}\!\Bigl\{
    \Prob\bigl(x_{j+1}(t)\in\Cset,
    \ \forall t\in[0,T]\mid D_j^{\mathrm{cal}}\bigr)\\
    \ge 1-\alpha
  \Bigr\}\ge 1-\delta .
\end{multline*}
\end{arxivonly}
\begin{submitonly}
$\Prob_{n_j}\{
    \Prob(x_{j+1}(t)\in\Cset,\
    \forall t\in[0,T]
    \mid D_j^{\mathrm{cal}})
    \ge 1-\alpha
  \} \ge 1-\delta$.
\end{submitonly}
 \emph{(Stability).} Let $V(x)$ be a robust CLF on $\Kset$  with decay rate $c$ and margin
$r_{j+1}$. Assume that $\mathcal{V}$ is compact. Furthermore, let $\pi_{{j+1}}(x)$ be a continuous function that enforces the CLF constraint \eqref{eq:ct_clf} for all $x\in\Kset$. If $\mu_{\XX}(\mathcal V)=1$, then
\begin{extendedonly}
\begin{equation*}
  \Prob_{n_j}\Bigl\{
    \Prob\bigl(
      \|x_{j+1}(t)\|
      \le\underline\alpha_V^{-1}\bigl(
        e^{-ct}\overline\alpha_V(\|x_{j+1}(0)\|)\bigr)\text{ for all } t\in[0,T]
    \mid D_j^{\mathrm{cal}}\bigr)\ge 1-\alpha
  \Bigr\}\ge 1-\delta.
\end{equation*}
\end{extendedonly}
\begin{arxivonly}
\begin{multline*}
  \Prob_{n_j}\!\Bigl\{
    \Prob\bigl(
      \|x_{j+1}(t)\|
      \le\underline\alpha_V^{-1}\bigl(
        e^{-ct}\overline\alpha_V(\|x_{j+1}(0)\|)\bigr),\\
    \forall t\in[0,T]
    \mid D_j^{\mathrm{cal}}\bigr)
    \ge 1-\alpha
  \Bigr\}
  \ge 1-\delta.
\end{multline*}
\end{arxivonly}
\begin{submitonly}
$\Prob_{n_j}\{
    \Prob(
      \|x_{j+1}(t)\|
      \le\underline\alpha_V^{-1}(
        e^{-ct}\overline\alpha_V(\|x_{j+1}(0)\|)),\
    \forall t\in[0,T]
    \mid D_j^{\mathrm{cal}})
    \ge 1-\alpha
  \} \ge 1-\delta$.
\end{submitonly}
\subonly{We present discrete-time analogs in Appendix~\ref{app:dt_certificates} of~\cite{our_extended}.}
\end{theorem}
\fulonly{We present the discrete-time analog in Appendix~\ref{app:dt_certificates}.} 

\subsection{Convergence of the Explicit Update Rule}
\label{sec:convergence}

We analyze the explicit update rule~\eqref{eq:explicit_update} as a stochastic recursion. For each $r\in\Rset$, recall that $\tau_r\sim\mathcal D_{\pi_r}$ denotes the trajectory of the system \eqref{eq:ct_true}  from the initial condition $x(0)\sim\mu_{\XX}$ under the policy $\pi_r$. Let 
\begin{equation}
    S_r:=s^{\ct}(\tau_r)
    \label{eq:score_random_variable_r}
\end{equation}
be the corresponding nonconformity score random variable, where $s^{\ct}(\cdot)$ is the nonconformity score in \eqref{eq:score_ct}. For each $p\in(0,1)$, define the population $p$-quantile of $S_r$ by
\begin{equation}
    Q_r(p):=\inf\{z\in\R:\Prob(S_r\le z)\ge p\}.
    \label{eq:population_quantile_r}
\end{equation}
With this notation, $S_{r_j}$ has the same distribution as the test nonconformity score $s_j$ at episode $j$. Additionally, the calibration nonconformity scores $s_j^{(1)},\ldots,s_j^{(n_j)}$ are i.i.d. samples from the same distribution. The sensitivity bound \fulonly{in~\eqref{eq:kappa_score_shift}}\subonly{$|s^{\ct}(x(0),\pi_{r'})-s^{\ct}(x(0),\pi_r)|\le\kappa|r'-r|$} and the quantile perturbation argument \subonly{in \Cref{cor:Q_lip} of Appendix~\ref{app:convergence} in~\cite{our_extended}}\fulonly{in \Cref{cor:Q_lip} of Appendix~\ref{app:convergence}} imply that the  population quantiles satisfy
\begin{equation}
  |Q_{r'}(1-\alpha)-Q_r(1-\alpha)|\le\kappa\,|r'-r|
  \qquad\forall\, r,r'\in\Rset.
  \label{eq:Q_Lipschitz}
\end{equation}
When $\kappa<1$, the function $r\mapsto Q_r(1-\alpha)$ is a contraction on $\Rset$. Hence, any fixed point in $\Rset$ is unique. We let $r_\star\in\Rset$ denote the fixed point satisfying
\begin{equation*}
    r_\star=Q_{r_\star}(1-\alpha),
\end{equation*}
i.e., the population $(1-\alpha)$-quantile of the nonconformity score equals the robustness margin $r_\star$ at this fixed point.

The conformal threshold $q_j$ is an empirical quantile and  generally differs from $Q_{r_j}(1-\alpha)$. We decompose $q_j$ as
\begin{equation}
    q_j=Q_{r_j}(1-\alpha)+\eta_j,  
    \label{eq:eta_decomposition_main}
\end{equation}
where $\eta_j:=\xi_j+b_j$ quantifies this deviation with
\begin{equation*}
    \xi_j:=q_j-Q_{r_j}(1-\bar\alpha_j),
    \;
    b_j:=Q_{r_j}(1-\bar\alpha_j)-Q_{r_j}(1-\alpha).
\end{equation*}
The term $\xi_j$ is the empirical quantile error at level $1-\bar\alpha_j$, while $b_j$ is the deterministic bias introduced by the tightened calibration level $1-\bar\alpha_j$. To control these terms, we use the following regularity condition.

\begin{assumption}
\label{ass:quantile_regularity}
There exist a neighborhood $\mathcal R_\star\subseteq\Rset$ around $r_\star$, an open interval $I\subseteq\R$, and a constant $m>0$ such that, for every $r\in\mathcal R_\star$, the cumulative distribution function
\begin{equation}
    F_r(z):=\Prob(S_r\le z)
    \label{eq:cdf_score_r_main}
\end{equation}
admits a density function $f_r(z)$ satisfying $f_r(z)\ge m$ for all $z\in I$. Furthermore, there exists an episode $J_0\in\mathbb N$ such that, for every $j\ge J_0$ and every $r\in\mathcal R_\star$, it holds that $Q_r(1-\alpha)\in I$ and $Q_r(p)\in I$ for all probability levels $p\in
    [1-\bar\alpha_j-\Delta_j,
    1-\bar\alpha_j+\Delta_j]\cap(0,1)$ 
where  $\Delta_j:=\sqrt{\frac{\ln(2/\delta_j)}{2n_j}}$.
\end{assumption}

The density bound $f_r(z)\ge m$ in  \Cref{ass:quantile_regularity} enables us to derive probabilistic bounds on the deviation $\eta_j$.  
Indeed, our quantile-error analysis \subonly{in Appendix~\ref{app:convergence} of~\cite{our_extended}}\fulonly{in \Cref{cor:eta_dkw} of Appendix~\ref{app:convergence}} implies that, for every episode $j\ge J_0$ with $r_j\in\mathcal R_\star$ and $1-\bar\alpha_j-\Delta_j,1-\bar\alpha_j+\Delta_j\in(0,1)$, we are guaranteed that
\begin{equation*}
\begin{gathered}
\Prob_{n_j}\!\left\{
|\eta_j|\le
m^{-1}\!\left(\sqrt{\frac{\ln(2/\delta_j)}{2n_j}}
+\alpha-\bar\alpha_j\right)
\right\}
\ge 1-\delta_j .
\end{gathered}
\end{equation*}
We next quantify the deviation of $r_j$ from $r_\star$.

\begin{theorem}
\label{thm:convergence}
Let the conditions of \Cref{thm:srcr_validity} and Assumptions~\ref{ass:ct_lip} and~\ref{ass4} hold. Suppose that $\kappa<1$ and let $r_\star\in\Rset$ be the unique fixed point satisfying $r_\star=Q_{r_\star}(1-\alpha)$. Define the error $e_j:=|r_j-r_\star|$ and variables $\lambda_\kappa:=2\kappa/(1-\kappa)$ and $B_\kappa:=1/(1-\kappa)$.
Then, for every $j\ge 0$, it holds that
\begin{fullonly}
\begin{equation} e_{j+1}\le\lambda_\kappa\,e_j+B_\kappa\,|\eta_j|.
  \label{eq:contraction}
\end{equation}
\end{fullonly}
\begin{submitonly}
\refstepcounter{equation}\label{eq:contraction}%
$e_{j+1}\le\lambda_\kappa\,e_j+B_\kappa\,|\eta_j|$.
\end{submitonly}
Consequently,  $e_{j+1}
    \le \lambda_\kappa^{j+1}e_0
    +B_\kappa\sum_{m=0}^{j}\lambda_\kappa^{j-m}|\eta_m|$.
If $\kappa<1/3$ and $\sup_j|\eta_j|\le C$, then $\limsup_{j\to\infty}e_j\le \frac{C}{1-3\kappa}$.
\end{theorem}

In equation%
\fulonly{~\eqref{eq:contraction}}%
\subonly{ $e_{j+1}\le\lambda_\kappa\,e_j+B_\kappa\,|\eta_j|$}, $\lambda_\kappa e_j$ captures the contraction of the map $r\mapsto Q_r(1-\alpha)$ with constant $\kappa<1$, while $B_\kappa|\eta_j|$ captures the finite-sample quantile-estimation perturbation caused by the deviation of $q_j$ from $Q_{r_j}(1-\alpha)$. When $\kappa<1/3$, the margin $r_j$ tracks $r_\star$ up to the  error\subonly{; see Appendix~\ref{app:convergence} of \cite{our_extended} for the proof}. Finally, we provide guarantees across multiple episods under \Cref{ass:quantile_regularity}.

\begin{corollary}
\label{cor:hp_tracking_convergence}
Let the conditions of \Cref{thm:srcr_validity} and Assumptions~\ref{ass:ct_lip}, \ref{ass4}, and~\ref{ass:quantile_regularity} hold. For an episode $J\in\mathbb N$ with $J\ge J_0$, assume that $r_j\in\mathcal R_\star$ and $1-\bar\alpha_j-\Delta_j,1-\bar\alpha_j+\Delta_j\in(0,1)$ hold for all $j=J_0,\ldots,J$ where  $\Delta_j:=\sqrt{\frac{\ln(2/\delta_j)}{2n_j}}$. Define the  event  $\mathcal H_J:=\bigcap_{j=J_0}^{J}
\left\{e_{j+1}\le
\lambda_\kappa e_j+B_\kappa\varepsilon_j\right\}$ where $\varepsilon_j:=m^{-1}\!\left(
\sqrt{\frac{\ln(2/\delta_j)}{2n_j}}
+\alpha-\bar\alpha_j\right)$.
Then, it holds that\footnote{Here, $\Prob_{0:J}\{\cdot\}$ denotes the joint probability measure over the probability measures $\Prob_{n_0}\{\cdot\},\hdots, \Prob_{n_J}\{\cdot\}$ at individual episodes.}
\begin{equation*}
    \Prob_{0:J}\{\mathcal H_J\}\ge 1-\sum_{j=J_0}^{J}\delta_j.
\end{equation*}
 If the above conditions hold for all $j\ge J_0$, $\kappa<1/3$, $\sum_{j=J_0}^\infty\delta_j<1$, and $\sup_{j\ge J_0}\varepsilon_j\le C$, then we have
\begin{equation*}
\begin{gathered}
\Prob_{0:\infty}\!\left\{
\limsup_{j\to\infty}|r_j-r_\star|
\le \frac{C}{1-3\kappa}
\right\}
\ge 1-\sum_{j=J_0}^{\infty}\delta_j .
\end{gathered}
\end{equation*}
If, in addition, $\varepsilon_j\to0$, then
\begin{equation*}
\Prob_{0:\infty}\!\left\{\lim_{j\to\infty}r_j=r_\star\right\}
\ge 1-\sum_{j=J_0}^{\infty}\delta_j .
\end{equation*}
\subonly{\vspace{-.1cm}}
\end{corollary}
\fulonly{The proofs are given in Appendix~\ref{app:convergence}.}

\section{Case Studies}
\label{sec:case_studies}
\begin{extendedonly}
\revised{We validate our method on the active benchmarks included in this draft: stability of an inverted pendulum via rCLFs (\Cref{sec:inverted_pend}), collision avoidance in a multi-obstacle planar maze via rCBFs (\Cref{sec:maze}), and quadcopter obstacle-avoidance via higher-order rCBFs (Appendix~\ref{app:navigation_details}). The quadcopter obstacle-avoidance task is built on the QuadSwarm simulator~\cite{quadswarm}; see Appendix~\ref{app:navigation_details} for details on all case studies. In all experiments we compare four baselines: \emph{robust} (our iterative update rule~\eqref{eq:explicit_update}), \emph{naive} ($r_j=q_j$ each episode), \emph{calibrate-once} ($r_j=q_0$ for all $j$), and \emph{non-robust} ($r_j=0$).}
\end{extendedonly}
\begin{arxivonly}
We validate our method on three benchmarks of increasing complexity: stability of an inverted pendulum via rCLFs (\Cref{sec:inverted_pend}), collision avoidance in a multi-obstacle planar maze via rCBFs (\Cref{sec:maze}), and quadcopter obstacle-avoidance via higher-order rCBFs (Appendix~\ref{app:navigation_details}). The quadcopter obstacle-avoidance task is built on the QuadSwarm simulator~\cite{quadswarm}; see Appendix~\ref{app:navigation_details}  for details on all case studies. In all experiments we compare four baselines: \emph{robust} (our iterative update~\eqref{eq:explicit_update}), \emph{naive} ($r_j=q_j$ each episode), \emph{calibrate-once} ($r_j=q_0$ for all $j$), and \emph{non-robust} ($r_j=0$).
\end{arxivonly}
\begin{submitonly}
We validate our method on three benchmarks: stability of an inverted pendulum via rCLFs (\Cref{sec:inverted_pend}), collision avoidance in a multi-obstacle planar maze via rCBFs (\Cref{sec:maze}), and quadcopter obstacle-avoidance via higher-order rCBFs (presented in Appendix~\ref{app:navigation_details} of~\cite{our_extended}).  In all experiments we compare four baselines: \emph{robust} (our iterative update~\eqref{eq:explicit_update}), \emph{naive} ($r_j=q_j$ in each episode), \emph{calibrate-once} ($r_j=q_0$ for all $j$), and \emph{non-robust} ($r_j=0$).
\end{submitonly}
\subsection{Inverted Pendulum}
\label{sec:inverted_pend}
We consider the inverted pendulum setup of~\cite{hsu2025}.
Unlike~\cite{hsu2025}, which uses a degree-5 polynomial regressor,
we use a degree-3 nominal model
$\hat{f}(x,u) = M_1 \phi(x) + M_2 \phi(x)\, u$,
where $M_1, M_2 \in \mathbb{R}^{2 \times 10}$ are learned weight matrices and
$\phi(x) = [1 \ \theta \ \dot\theta \ \theta^2 \ \theta\dot\theta
\ \dot\theta^2 \ \theta^3 \cdots \dot\theta^3]$.
The quadratic CLF $V(x) = x^\top P x$ is obtained via the Lyapunov equation
as in~\cite{hsu2025} (using feedback gain $K = [6 \ 1]$ and $Q = 0.5\, I_{2\times 2}$).
We sample initial states from
$\mathcal{X}_0 \triangleq \{ x \in \mathbb{R}^2 \mid V(x) = 1.3 \}$
with $T = 5$, collect $n_j = 200$ trajectories per episode
for estimating $q_j$ ($\alpha = 0.1$, $\delta = 0.1$)
and $100$ for validation.
\Cref{fig:ip-results} compares our algorithm against the baselines.
The calibrate-once baseline performs well here and $r_{\text{calibrate-once}}$ effectively lower-bounds
$\{ s_j^{(i)} \}_{i=1}^n$ across episodes,
as expected
since this setting was chosen in~\cite{hsu2025} to demonstrate this particular baseline. This baseline will not perform well in the second case study, similar to  \Cref{prop:counterexample_main}.
Importantly, \Cref{fig:ip-results}(a) shows that
$r_{\text{robust}}$ converges to the $1{-}\alpha$ score quantile
while maintaining a positive margin in empirical score coverage
over all baselines in~(b).
As in~\cite{hsu2025}, the nonrobust case fails
to achieve stability (see \Cref{fig:ip-trajectories}),
underscoring the necessity of robustification under dynamics mismatch.

\begin{figure}[!tbp]
    \centering
    \captionsetup{font=small,skip=2pt}
    \begin{subfigure}[t]{0.49\linewidth}        \includegraphics[width=\linewidth]{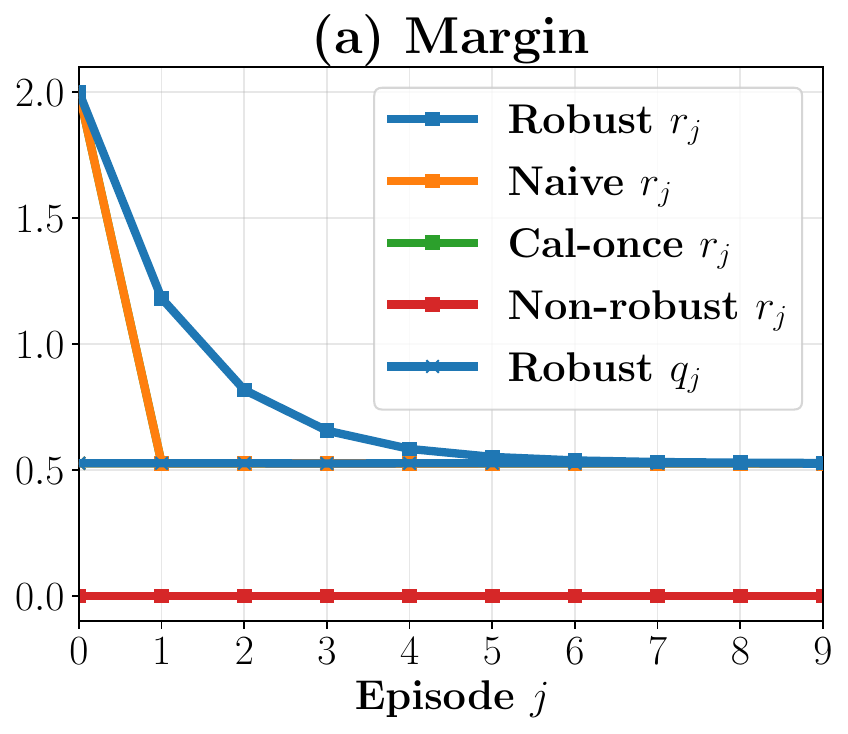}
    \end{subfigure}\hfill
    \begin{subfigure}[t]{0.49\linewidth}        \includegraphics[width=\linewidth]{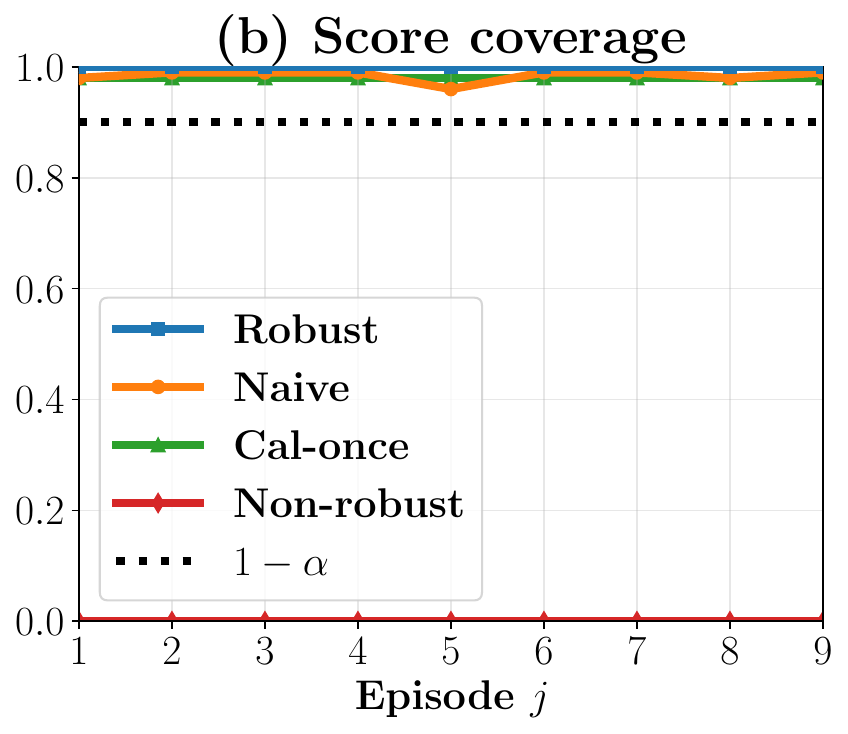}
    \end{subfigure}\hfill
    \caption{Inv. pendulum. (a)~$r_j$, $q_j$. (b)~Score coverage $s_j^{(i)}\!\le\! r_j$.}
    \label{fig:ip-results}
    \subonly{\vspace{-.4cm}}
\end{figure}
\begin{figure}[!tbp]
    \centering
    \captionsetup{font=small,skip=2pt}
    \begin{subfigure}[t]{0.52\linewidth}
        \includegraphics[width=\linewidth]{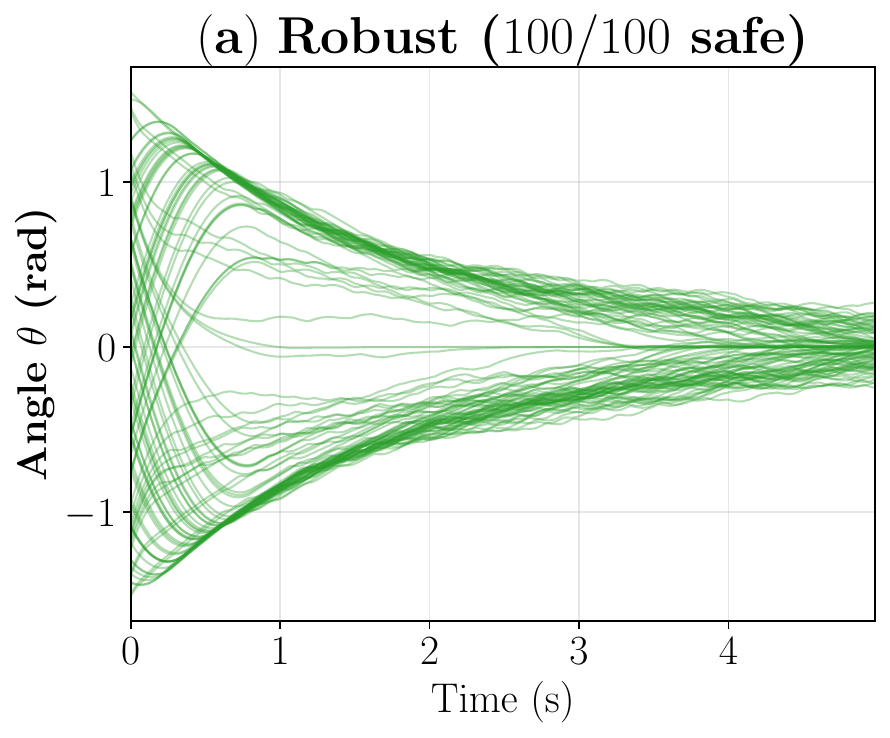}
    \end{subfigure}\hfill
    \begin{subfigure}[t]{0.48\linewidth}
        \includegraphics[width=\linewidth]{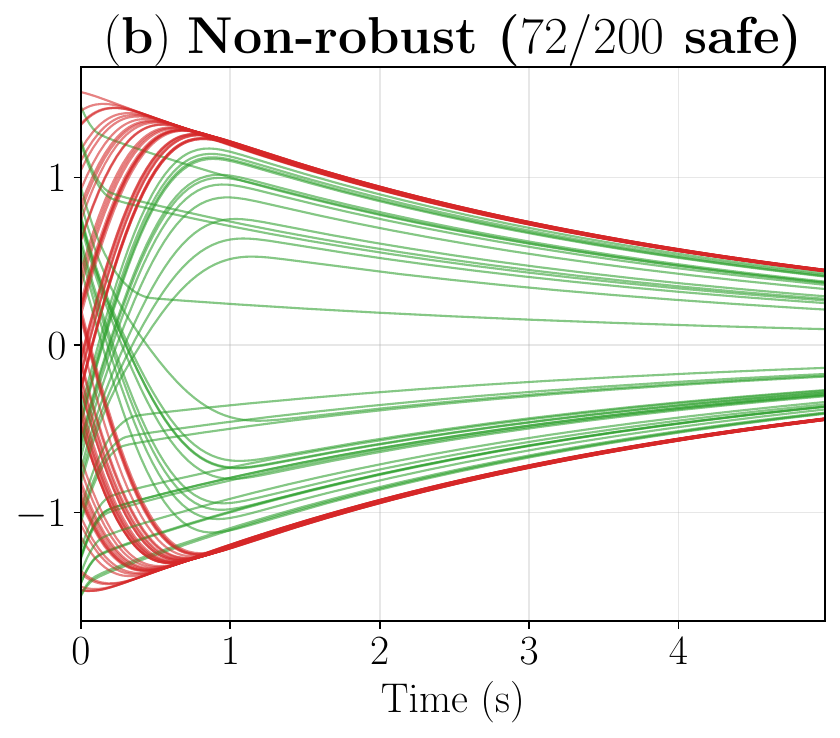}
    \end{subfigure} \\
    \begin{subfigure}[t]{0.51\linewidth}
        \includegraphics[width=\linewidth]{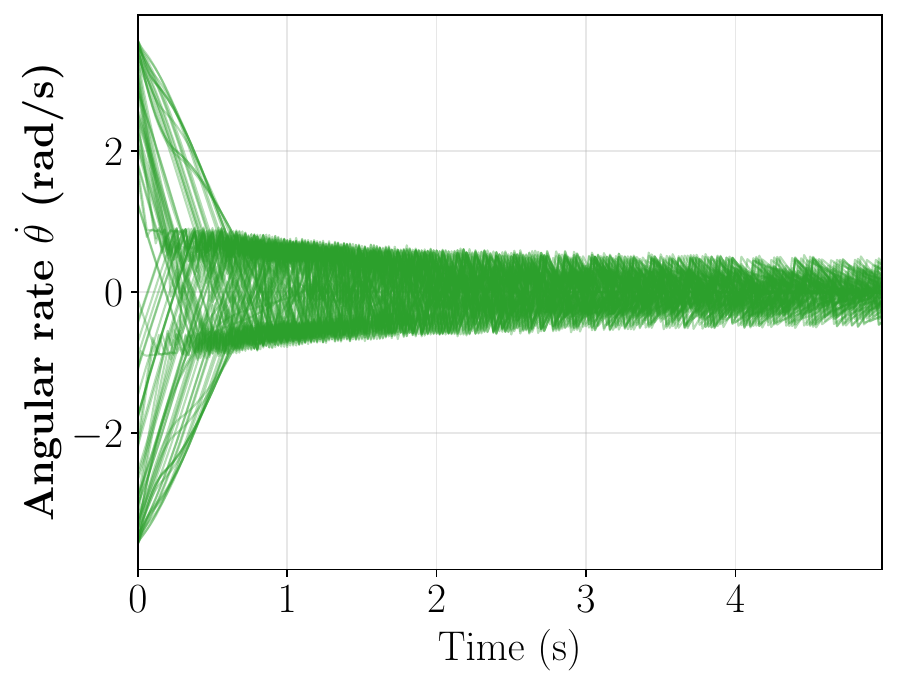}
    \end{subfigure} \hfill
    \begin{subfigure}[t]{0.47\linewidth}
        \includegraphics[width=\linewidth]{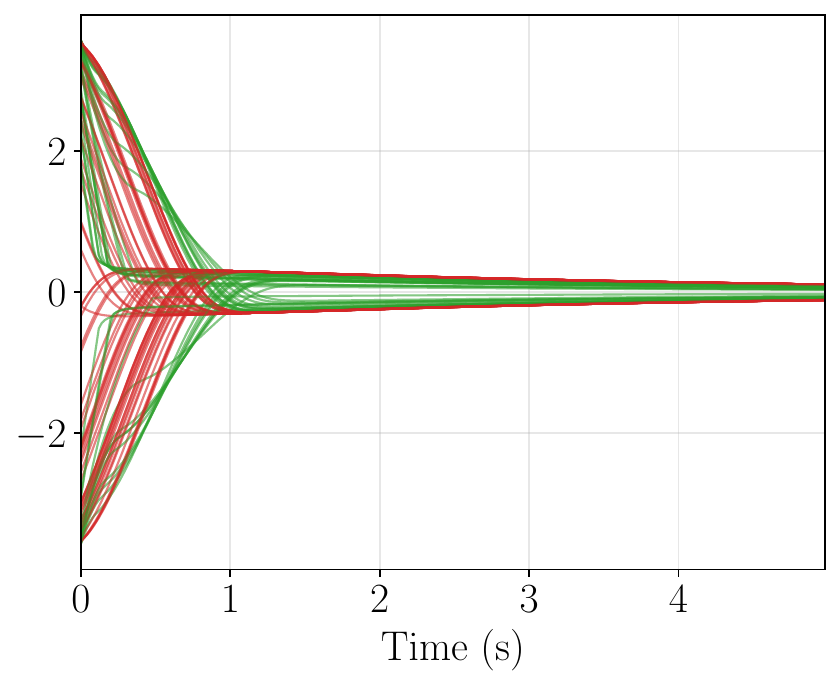}
    \end{subfigure}
    \caption{Inverted pendulum trajectories. (a)~$r\!=\!r_{\text{robust}}=0.5282$. (b)~$r\!=\!0$. Colors indicate stability violation per Theorem~\ref{thm:det_ct_clf}.}
    \label{fig:ip-trajectories}
    \subonly{\vspace{-.4cm}}
\end{figure}
\subsection{Multi-Obstacle Maze Navigation}
\label{sec:maze}
We consider a planar single-integrator system $\hat f(x,u)=u$, $f(x,u)=u+\varepsilon(x,u)$, $x\in\R^2$, $u\in\R^2$, navigating a maze of $17$ circular obstacles with centers $\{c_i\}_{i=1}^{17}$, physical radii $R_i\in[0.22,0.45]$\,m, and safety radii $R_{s,i}=1.25\,R_i$. The barrier function for obstacle~$i$ is $h_i(x)=\|x-c_i\|^2-R_{s,i}^2$, the safe set is $\Cset=\bigcap_{i}\{x:h_i(x)\ge 0\}$,
and the model error is an aggregate barrier-localized vortex field
$\varepsilon(x,u)=\sum_{i=1}^{17}\sigma_i(x)(R_{\theta_i}-I_2)u+d_{\mathrm{m}}$
, where $R_\theta:=\bigl[\begin{smallmatrix}\cos\theta & -\sin\theta\\\sin\theta & \cos\theta\end{smallmatrix}\bigr]\in\mathrm{SO}(2)$ is the 2D rotation matrix,
with $\sigma_i(x):=\exp(-\max\{h_i(x),0\}/\ell_i)$, per-obstacle rotation angles
$\theta_i\in[10^\circ,32^\circ]$, length-scales $\ell_i\in[0.15,0.50]$, and drift $d_{\mathrm{m}}=(0.001,-0.002)^\top$. The obstacles form a maze with boundary rows at $y=\pm 2$ and staggered interior obstacles creating narrow passages; the full layout is given in \subonly{Table~\ref{tab:maze_obstacles} of \cite{our_extended}.}%
\fulonly{Table~\ref{tab:maze_obstacles}.} We set goal $B_{\mathrm{m}}=(10,0)$, nominal controller $u_{\mathrm{nom}}(x):=0.6(B_{\mathrm{m}}-x)$, and ${\XX}_{\mathrm{m}}=[-5,-0.5]\times[-2.59,2.59]\cap\{x:\min_i h_i(x)\ge 0.05\}$, $T=12$, $\Delta t=0.01$, $\alpha=0.1$, $\delta=0.05$, and $\gamma=10$. The rCBF-QP~\eqref{eq:cbf_qp} enforces all $17$ constraints $\ip{\nabla h_i(x)}{u}+\gamma h_i(x)\ge\|\nabla h_i(x)\|r$ simultaneously; since $x\in\R^2$, at most two constraints can be active at the optimum, and we solve the QP exactly via a $2$-D active-set enumeration. For every episode~$j$, $n_j=200$ calibration and $N_{\mathrm{eval}}=500$ evaluation trajectories are collected, and we run $20$ episodes with $\kappa=0.3$.
\Cref{fig:maze_episodic}(a) shows that the margin converges from $r_0\approx 2.75$ to $r\approx 2.38$ within $\sim\!2$ episodes, consistent with \Cref{thm:convergence} since $\kappa=0.3<1/3$. The calibrate-once and naive baselines achieve score coverage of only $0.668$ and $0.872$, both below $1-\alpha=0.9$, while our iterative method reaches $0.998$ (\Cref{fig:maze_episodic}(b)) and always satisfies the coverage level. All three  methods maintain $100\%$ safety throughout%
\fulonly{ (see \Cref{fig:maze-episodic-large}(c) in \Cref{app:tables_figures})}%
\subonly{ (see Appendix ~\ref{app:maze_details} of \cite{our_extended})},
confirming that the iterative margin update balances coverage and robustness even in a geometrically complex multi-obstacle environment.
\Cref{fig:maze_traj} illustrates the converged SR-CR policy navigating the maze%
\fulonly{; see Tables~\ref{tab:maze_params}--\ref{tab:maze_final} for full parameters and metrics}%
\subonly{; see~\cite{our_extended} for full parameters and metrics}.
\begin{figure}[!tbp]
    \centering
    \captionsetup{font=small,skip=2pt}
    \begin{subfigure}[t]
{0.50\linewidth}      \includegraphics[width=\linewidth]{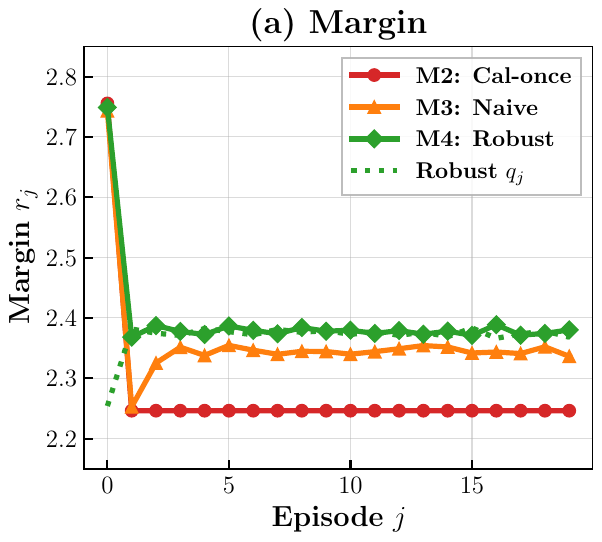}
    \end{subfigure}\hfill
        \begin{subfigure}[t]{0.50\linewidth} 
     \includegraphics[width=\linewidth]{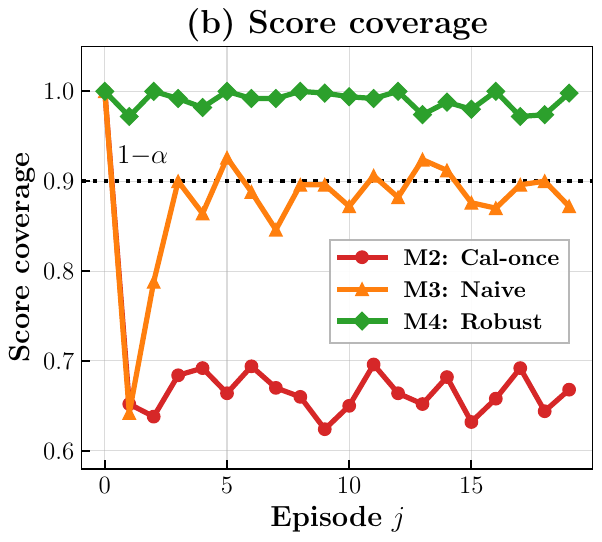}
    \end{subfigure}\hfill
    \caption{Maze: (a)~Margin $r_j$, threshold $q_j$. (b)~Score coverage. }
  \label{fig:maze_episodic}
 \subonly{\vspace{-.4cm}}
\end{figure}
  \subonly{\vspace{-.3cm}}
\begin{figure}[!tbp]
    \centering
    \captionsetup{font=small,skip=2pt}
    \begin{subfigure}[t]{\linewidth}  \includegraphics[width=\linewidth]{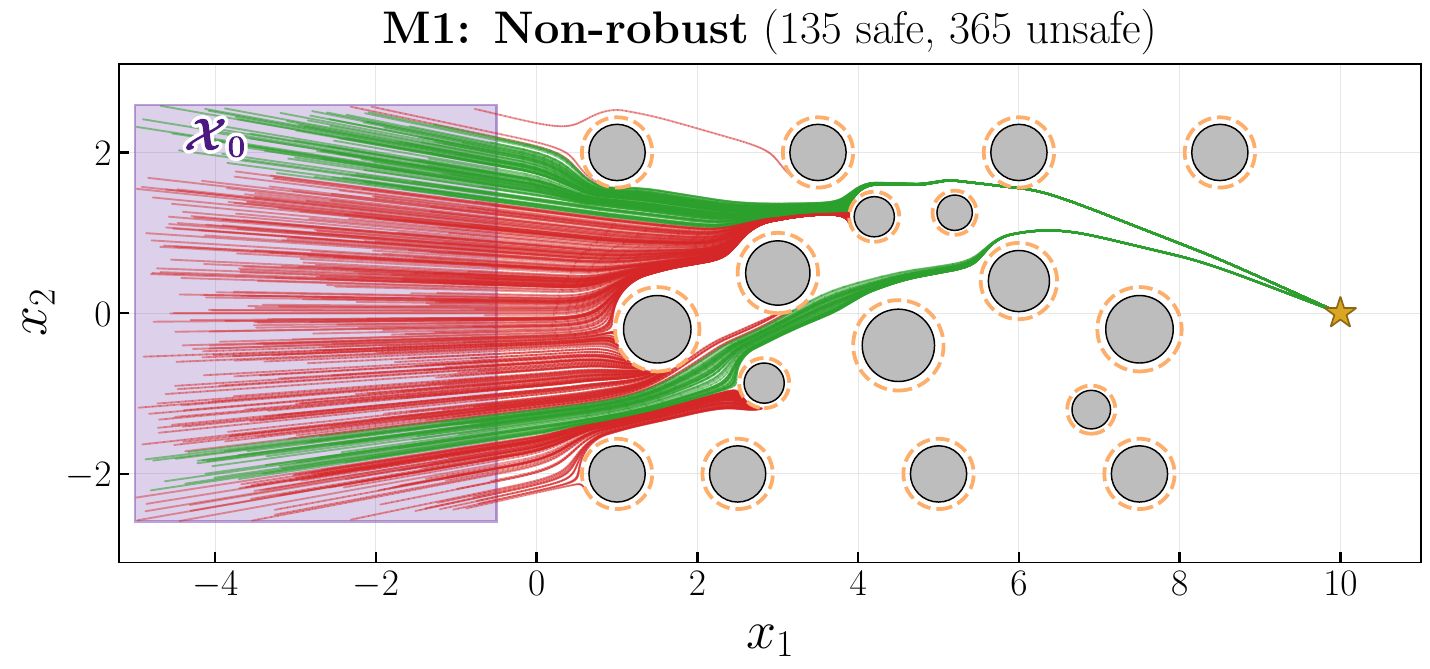}
    \end{subfigure} \\
    \begin{subfigure}[t]{\linewidth}
    \includegraphics[width=\linewidth]{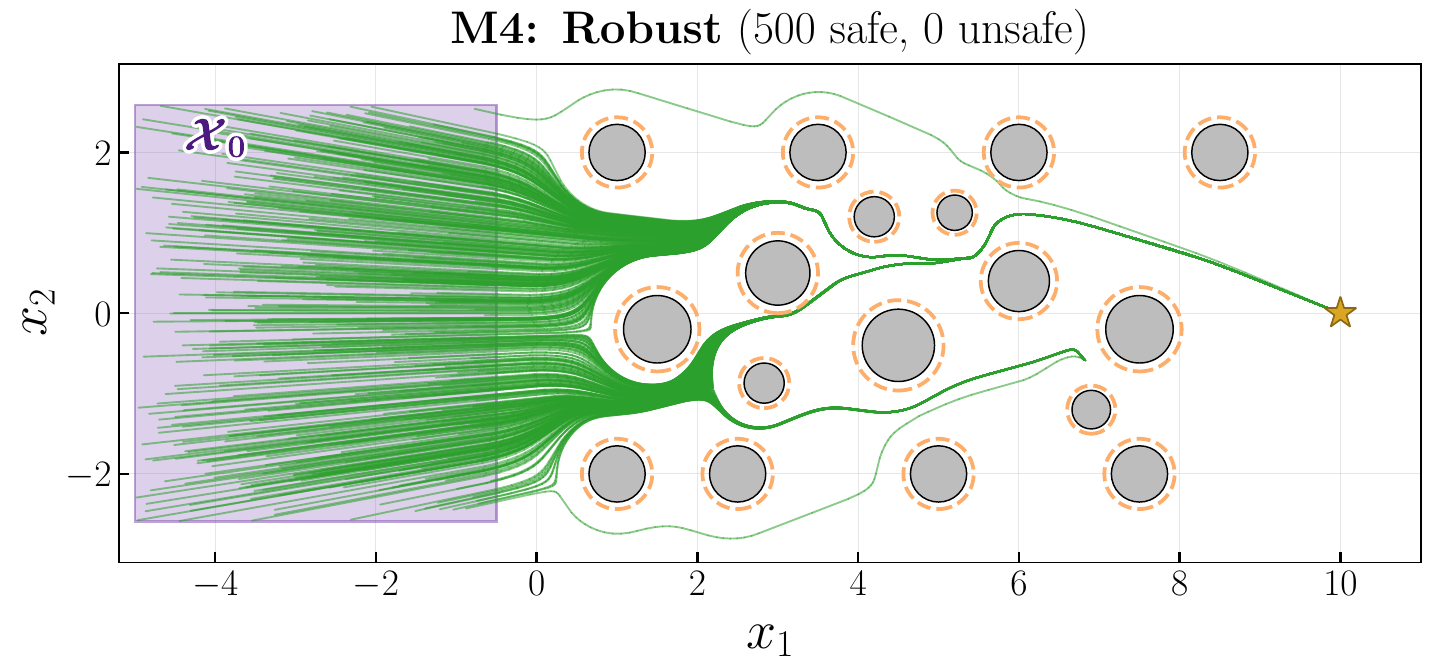}
    \end{subfigure}
    \caption{Maze: $500$ trajectories under $r_{19}$, (a)~Non-robust, (b)~Robust.}
  \label{fig:maze_traj}
  \subonly{\vspace{-.8cm}}
\end{figure}

\section{Conclusion}
\label{sec:conclusion}
We developed a framework for data-driven robust control that maintains provable safety and stability guarantees across iterative policy updates despite policy-induced distribution shift. By combining adversarially robust conformal prediction with robust CLF/CBF-QP synthesis, the framework transfers stability/safety guarantees from one episode to the next via a computable distribution shift budget derived from closed-loop sensitivity analysis. We proved per-episode finite-sample validity, probabilistic CBF safety and CLF stability certificates, and convergence of the algorithm.%
\extonly{ We demonstrated these guarantees on six benchmarks including an inverted pendulum, 2-D navigation tasks, a multi-obstacle maze, and quadrotor systems.}%
\arxonly{ We demonstrated these guarantees on an inverted pendulum, a multi-obstacle maze, and a quadrotor system.}%
\subonly{ Finally, we demonstrated these guarantees on three different systems.}%
\renewcommand*{\bibfont}{\footnotesize}
\printbibliography

@article{sontag1989universal,
  title   = {A ``Universal'' Construction of Artstein's Theorem on Nonlinear Stabilization},
  author  = {Sontag, Eduardo D.},
  journal = {Syst. Control Lett.},
  volume  = {13},
  number  = {2},
  pages   = {117--123},
  year    = {1989},
  doi     = {10.1016/0167-6911(89)90028-5}
}

@misc{our_extended,
  author = {Mirzaeedodangeh, Omid and Shekhtman, Eliot and Matni, Nikolai and Lindemann, Lars},
  title  = {Robust Conformal CBF and CLF Controllers via Iterative Policy Updates: Extended Version},
  year   = {2026},
  url    = {https://tinyurl.com/CDC1481}
}

@inproceedings{duchi2025sample,
  title     = {Sample-Conditional Coverage in Split-Conformal Prediction},
  author    = {Duchi, John C.},
  booktitle = {Adv. Neural Inf. Process. Syst.},
  year      = {2025}
}

@article{mestres2025regularity,
  title   = {Regularity properties of optimization-based controllers},
  author  = {Mestres, Pol and Allibhoy, Ahmed and Cort{\'e}s, Jorge},
  journal = {Eur. J. Control},
  volume  = {81},
  pages   = {101098},
  year    = {2025}
}

@book{khalil2002nonlinear,
  title     = {Nonlinear Systems},
  author    = {Khalil, Hassan K.},
  edition   = {3},
  year      = {2002},
  publisher = {Prentice Hall}
}

@article{xu2015robustness,
  title   = {Robustness of control barrier functions for safety critical control},
  author  = {Xu, Xiangru and Tabuada, Paulo and Grizzle, Jessy W. and Ames, Aaron D.},
  journal = {IFAC-PapersOnLine},
  volume  = {48},
  number  = {27},
  pages   = {54--61},
  year    = {2015},
  doi     = {10.1016/j.ifacol.2015.11.152}
}

@book{freeman2008robust,
  title     = {Robust Nonlinear Control Design: State-Space and Lyapunov Techniques},
  author    = {Freeman, Randy and Kokotovi{\'c}, Petar V.},
  year      = {2008},
  publisher = {Springer}
}

@article{ames2019cbf,
  title   = {Control Barrier Function Based Quadratic Programs for Safety Critical Systems},
  author  = {Ames, Aaron D. and Xu, Xiangru and Grizzle, Jessy W. and Tabuada, Paulo},
  journal = {IEEE Trans. Autom. Control},
  volume  = {62},
  number  = {8},
  pages   = {3861--3876},
  year    = {2017},
  doi     = {10.1109/TAC.2016.2638961}
}

@article{angelopoulos2023conformal,
  title   = {Conformal Prediction: A Gentle Introduction},
  author  = {Angelopoulos, Anastasios N. and Bates, Stephen},
  journal = {Found. Trends Mach. Learn.},
  volume  = {16},
  number  = {4},
  pages   = {494--591},
  year    = {2023},
  doi     = {10.1561/2200000101}
}

@article{lindemann2025formal,
  title   = {Formal Verification and Control with Conformal Prediction: Practical Safety Guarantees for Autonomous Systems},
  author  = {Lindemann, Lars and Zhao, Yiqi and Yu, Xinyi and Pappas, George J. and Deshmukh, Jyotirmoy V.},
  journal = {IEEE Control Syst. Mag.},
  volume  = {45},
  number  = {6},
  pages   = {72--122},
  year    = {2025},
  doi     = {10.1109/MCS.2025.3611545}
}

@article{hsu2025,
  title   = {Statistical Guarantees in Data-Driven Nonlinear Control: Conformal Robustness for Stability and Safety},
  author  = {Hsu, Ting-Wei and Tsukamoto, Hiroyasu},
  journal = {IEEE Control Syst. Lett.},
  volume  = {9},
  pages   = {997--1002},
  year    = {2025},
  doi     = {10.1109/LCSYS.2025.3578062}
}

@inproceedings{vovk12conditional,
  title     = {Conditional Validity of Inductive Conformal Predictors},
  author    = {Vovk, Vladimir},
  booktitle = {Proc. Asian Conf. Mach. Learn.},
  series    = {Proc. Mach. Learn. Res.},
  volume    = {25},
  pages     = {475--490},
  year      = {2012},
  publisher = {PMLR}
}

@book{vovk2005alrw,
  title     = {Algorithmic Learning in a Random World},
  author    = {Vovk, Vladimir and Gammerman, Alex and Shafer, Glenn},
  publisher = {Springer},
  year      = {2005},
  doi       = {10.1007/b106715}
}

@misc{hardt2023performativepastfuture,
  title         = {Performative Prediction: Past and Future},
  author        = {Hardt, Moritz and Mendler-D{\"u}nner, Celestine},
  year          = {2023},
  eprint        = {2310.16608},
  archivePrefix = {arXiv},
  primaryClass  = {cs.LG}
}

@article{mirzaeedodangeh2025safe,
  title={Safe Planning in Interactive Environments via Iterative Policy Updates and Adversarially Robust Conformal Prediction},
  author={Mirzaeedodangeh, Omid and Shekhtman, Eliot and Matni, Nikolai and Lindemann, Lars},
  journal={arXiv preprint arXiv:2511.10586},
  year={2025}
}

@article{massart1990dkw,
  title   = {The Tight Constant in the {D}voretzky-{K}iefer-{W}olfowitz Inequality},
  author  = {Massart, Pascal},
  journal = {Ann. Probab.},
  volume  = {18},
  number  = {3},
  pages   = {1269--1283},
  year    = {1990},
  doi     = {10.1214/aop/1176990746}
}

@book{bonnans2000perturbation,
  author    = {Bonnans, J. Fr{\'e}d{\'e}ric and Shapiro, Alexander},
  title     = {Perturbation Analysis of Optimization Problems},
  publisher = {Springer},
  year      = {2000},
  doi       = {10.1007/978-1-4612-1394-9}
}

@inproceedings{cairoli2023conformal,
  title     = {Conformal Quantitative Predictive Monitoring of {STL} Requirements for Stochastic Processes},
  author    = {Cairoli, Francesca and Paoletti, Nicola and Bortolussi, Luca},
  booktitle = {Proc. ACM Int. Conf. Hyb. Syst.: Comp. Cont.},
  pages     = {1--11},
  year      = {2023}
}

@inproceedings{lindemann2023conformal,
  title     = {Conformal Prediction for {STL} Runtime Verification},
  author    = {Lindemann, Lars and Qin, Xin and Deshmukh, Jyotirmoy V. and Pappas, George J.},
  booktitle = {Proc. ACM/IEEE Int. Conf. Cyber-Phys. Syst.},
  pages     = {142--153},
  year      = {2023},
  doi       = {10.1145/3576841.3585927}
}

@book{dontchev2014implicit,
  author    = {Dontchev, Asen L. and Rockafellar, R. Tyrrell},
  title     = {Implicit Functions and Solution Mappings: A View from Variational Analysis},
  publisher = {Springer},
  edition   = {2nd},
  year      = {2014},
  doi       = {10.1007/978-1-4939-1037-3}
}

@misc{ren2023robots,
  title         = {Robots That Ask for Help: Uncertainty Alignment for Large Language Model Planners},
  author        = {Ren, Allen Z. and Dixit, Anushri and Bodrova, Alexandra and Singh, Sumeet and Tu, Stephen and Brown, Noah and Xu, Peng and Takayama, Leila and Xia, Fei and Varley, Jake and others},
  year          = {2023},
  eprint        = {2307.01928},
  archivePrefix = {arXiv},
  primaryClass  = {cs.RO}
}

@article{wang2024probabilistically,
  title   = {Probabilistically Correct Language-Based Multi-Robot Planning Using Conformal Prediction},
  author  = {Wang, Jun and He, Guocheng and Kantaros, Yiannis},
  journal = {IEEE Robot. Autom. Lett.},
  volume  = {10},
  number  = {1},
  pages   = {160--167},
  year    = {2025}
}

@inproceedings{gibbs2021adaptive,
  title     = {Adaptive Conformal Inference Under Distribution Shift},
  author    = {Gibbs, Isaac and Cand{\`e}s, Emmanuel J.},
  booktitle = {Adv. Neural Inf. Process. Syst.},
  volume    = {34},
  pages     = {1660--1672},
  year      = {2021}
}

@inproceedings{gendler2021arcp,
  title     = {Adversarially Robust Conformal Prediction},
  author    = {Gendler, Asaf and Weng, Tsui-Wei and Daniel, Luca and Romano, Yaniv},
  booktitle = {Proc. Int. Conf. Learn. Represent. (ICLR)},
  year      = {2022}
}

@inproceedings{srinivasan2026oodmpc,
  title         = {Safety Beyond the Training Data: Robust Out-of-Distribution {MPC} via Conformalized System Level Synthesis},
  author        = {Srinivasan, Anutam and Leeman, Antoine and Chou, Glen},
  booktitle     = {Proc. Learn. Dyn. Control Conf. (L4DC)},
  year          = {2026},
  note          = {Accepted},
  eprint        = {2602.12047},
  archivePrefix = {arXiv},
  primaryClass  = {eess.SY}
}

@misc{zhu2025sparc,
  title         = {{SPARC}: Prediction-Based Safe Control for Coupled Controllable and Uncontrollable Agents with Conformal Predictions},
  author        = {Wang, Shuqi and Wang, Siqi and Li, Shaoyuan and Yin, Xiang},
  year          = {2025},
  eprint        = {2410.15660},
  archivePrefix = {arXiv},
  primaryClass  = {eess.SY}
}

@article{huang2023sample,
  title   = {On the Sample Complexity of {L}ipschitz Constant Estimation},
  author  = {Huang, Julien Walden and Roberts, Stephen and Calliess, Jan-Peter},
  journal = {Trans. Mach. Learn. Res.},
  year    = {2023}
}

@inproceedings{fazlyab2019efficient,
  title     = {Efficient and Accurate Estimation of {L}ipschitz Constants for Deep Neural Networks},
  author    = {Fazlyab, Mahyar and Robey, Alexander and Hassani, Hamed and Morari, Manfred and Pappas, George J.},
  booktitle = {Adv. Neural Inf. Process. Syst.},
  volume    = {32},
  year      = {2019}
}

@article{lindemann2023safeplanning,
  title   = {Safe Planning in Dynamic Environments Using Conformal Prediction},
  author  = {Lindemann, Lars and Cleaveland, Matthew and Shim, Gihyun and Pappas, George J.},
  journal = {IEEE Robot. Autom. Lett.},
  volume  = {8},
  number  = {8},
  pages   = {5116--5123},
  year    = {2023},
  doi     = {10.1109/LRA.2023.3292071}
}

@inproceedings{dixit2023adaptive,
  title     = {Adaptive Conformal Prediction for Motion Planning among Dynamic Agents},
  author    = {Dixit, Anushri and Lindemann, Lars and Wei, Skylar X. and Cleaveland, Matthew and Pappas, George J. and Burdick, Joel W.},
  booktitle = {Proc. Learn. Dyn. Control Conf. (L4DC)},
  series    = {Proc. Mach. Learn. Res.},
  volume    = {211},
  pages     = {300--314},
  year      = {2023},
  publisher = {PMLR}
}

@inproceedings{yang2023safeperception,
  title     = {Safe perception-based control under stochastic sensor uncertainty using conformal prediction},
  author    = {Yang, Shuo and Pappas, George J. and Mangharam, Rahul and Lindemann, Lars},
  booktitle = {Proc. IEEE Conf. Decis. Control (CDC)},
  pages     = {6072--6078},
  year      = {2023},
}

@inproceedings{zhou2024adaptivecp,
  title     = {Safety-critical control with uncertainty quantification using adaptive conformal prediction},
  author    = {Zhou, Hao and Zhang, Yanze and Luo, Wenhao},
  booktitle = {Proc. Amer. Control Conf. (ACC)},
  pages     = {574--580},
  year      = {2024},
}

@article{tayal2025cpncbf,
  title   = {Cp-ncbf: A conformal prediction-based approach to synthesize verified neural control barrier functions},
  author  = {Tayal, Manan and Singh, Aditya and Jagtap, Pushpak and Kolathaya, Shishir},
  journal = {arXiv:2503.17395},
  year    = {2025}
}

@inproceedings{cped2025ncbf,
  title     = {{CPED-NCBFs}: A Conformal Prediction for Expert Demonstration-Based Neural Control Barrier Functions},
  author    = {Sumeadh, M. S. and Dsouza, Kevin and Prakash, Ravi},
  booktitle = {Proc. Ind. Control Conf. (ICC)},
  year      = {2025},
  doi       = {10.1109/ICC69100.2025.11372429},
  eprint    = {2507.15022},
  archivePrefix = {arXiv}
}

@article{zhang2025conformal,
  title   = {Conformal Prediction in the Loop: Risk-Aware Control Barrier Functions for Stochastic Systems with Data-Driven State Estimators},
  author  = {Zhang, Junhui and Hoxha, Bardh and Fainekos, Georgios and Panagou, Dimitra},
  journal = {IEEE Control Syst. Lett.},
  volume  = {9},
  pages   = {282--287},
  year    = {2025},
  doi     = {10.1109/LCSYS.2025.3571828}
}

@inproceedings{rahaman2026environments,
  title         = {When Environments Shift: Safe Planning with Generative Priors and Robust Conformal Prediction},
  author        = {Rahaman, Kaizer and Deshmukh, Jyotirmoy V. and Hota, Ashish R. and Lindemann, Lars},
  booktitle     = {Proc. Learn. Dyn. Control Conf. (L4DC)},
  year          = {2026},
  note          = {Accepted},
  eprint        = {2602.12616},
  archivePrefix = {arXiv},
  primaryClass  = {cs.RO}
}

@inproceedings{tibshirani2019covshift,
  title     = {Conformal Prediction under Covariate Shift},
  author    = {Tibshirani, Ryan J. and Barber, Rina Foygel and Cand{\`e}s, Emmanuel J. and Ramdas, Aaditya},
  booktitle = {Adv. Neural Inf. Process. Syst.},
  volume    = {32},
  year      = {2019}
}

@misc{d2026statistical,
  title         = {Statistical Contraction for Chance-Constrained Trajectory Optimization of Non-{G}aussian Stochastic Systems},
  author        = {D'Silva, Rihan Aaron and Tsukamoto, Hiroyasu},
  year          = {2026},
  eprint        = {2603.07092},
  archivePrefix = {arXiv},
  primaryClass  = {eess.SY}
}

@article{sheng2024safe,
  title   = {Safe {POMDP} online planning among dynamic agents via adaptive conformal prediction},
  author  = {Sheng, Shili and Yu, Pian and Parker, David and Kwiatkowska, Marta and Feng, Lu},
  journal = {IEEE Robot. Autom. Lett.},
  volume  = {9},
  number  = {11},
  pages   = {9946--9953},
  year    = {2024},
}

@article{cauchois2024robust,
  title   = {Robust validation: Confident predictions even when distributions shift},
  author  = {Cauchois, Maxime and Gupta, Suyash and Ali, Alnur and Duchi, John C.},
  journal = {J. Amer. Statist. Assoc.},
  volume  = {119},
  number  = {548},
  pages   = {3033--3044},
  year    = {2024},
}

@inproceedings{zhao2024rprvshift,
  title     = {Robust Conformal Prediction for {STL} Runtime Verification under Distribution Shift},
  author    = {Zhao, Yiqi and Hoxha, Bardh and Fainekos, Georgios and Deshmukh, Jyotirmoy V. and Lindemann, Lars},
  booktitle = {Proc. ACM/IEEE Int. Conf. Cyber-Phys. Syst.},
  pages     = {169--179},
  year      = {2024},
  doi       = {10.1109/ICCPS61052.2024.00022},
  eprint    = {2311.09482},
  archivePrefix = {arXiv},
  primaryClass  = {cs.SY}
}

@inproceedings{contreras2024sodampc,
  title     = {Safe, Out-of-Distribution-Adaptive {MPC} with Conformalized Neural Network Ensembles},
  author    = {Contreras, Jose Leopoldo and Shorinwa, Ola and Schwager, Mac},
  booktitle = {Proc. Learn. Dyn. Control Conf. (L4DC)},
  series    = {Proc. Mach. Learn. Res.},
  volume    = {283},
  pages     = {194--207},
  year      = {2025},
  publisher = {PMLR},
  eprint    = {2406.02436},
  archivePrefix = {arXiv},
  primaryClass  = {cs.RO}
}

@misc{zhou2025safeRLacp,
  title         = {Computationally and Sample Efficient Safe Reinforcement Learning Using Adaptive Conformal Prediction},
  author        = {Zhou, Hao and Zhang, Yanze and Luo, Wenhao},
  year          = {2025},
  eprint        = {2503.17678},
  archivePrefix = {arXiv},
  primaryClass  = {eess.SY}
}

@misc{prinster2026cpc,
  title         = {Conformal Policy Control},
  author        = {Prinster, Drew and Fannjiang, Clara and Park, Ji Won and Cho, Kyunghyun and Liu, Anqi and Saria, Suchi and Stanton, Samuel},
  year          = {2026},
  eprint        = {2603.02196},
  archivePrefix = {arXiv},
  primaryClass  = {stat.ML}
}

@inproceedings{li2025performativerisk,
  title     = {Performative Risk Control: Calibrating Models for Reliable Deployment under Performativity},
  author    = {Li, Victor and Chen, Baiting and Mao, Yuzhen and Lei, Qi and Deng, Zhun},
  booktitle = {Adv. Neural Inf. Process. Syst.},
  year      = {2025},
  eprint    = {2505.24097},
  archivePrefix = {arXiv},
  primaryClass  = {stat.ML}
}

@inproceedings{huang2024confpolicylearning,
  title     = {Conformal policy learning for sensorimotor control under distribution shifts},
  author    = {Huang, Huang and Sharma, Satvik and Loquercio, Antonio and Angelopoulos, Anastasios and Goldberg, Ken and Malik, Jitendra},
  booktitle = {Proc. IEEE Int. Conf. Robot. Autom. (ICRA)},
  pages     = {16285--16291},
  year      = {2024},
}

@misc{quadswarm,
  title         = {{QuadSwarm}: A Modular Multi-Quadrotor Simulator for Deep Reinforcement Learning with Direct Thrust Control},
  author        = {Huang, Zhehui and Batra, Sumeet and Chen, Tao and Krupani, Rahul and Kumar, Tushar and Molchanov, Artem and Petrenko, Aleksei and Preiss, James A. and Yang, Zhaojing and Sukhatme, Gaurav S.},
  year          = {2023},
  eprint        = {2306.09537},
  archivePrefix = {arXiv},
  primaryClass  = {cs.RO},
}

@book{kallenberg2002foundations,
  author    = {Kallenberg, Olav},
  title     = {Foundations of Modern Probability},
  edition   = {2},
  publisher = {Springer},
  year      = {2002},
  series    = {Probability and Its Applications}
}
%
%
%
%
\ifsubmit
  \useRomanappendicesfalse
  \begin{refsection}
  \appendices
  \makeatletter
  %
  %
  \long\def\protected@write#1#2#3{%
    \begingroup
      \let\thepage\relax
      #2%
      \let\protect\@unexpandable@protect
      \edef\reserved@a{\immediate\write#1{#3}}%
      \reserved@a
    \endgroup
  }%
  %
  %
  \renewenvironment{figure}[1][]{\def\@captype{figure}}{}%
  \renewenvironment{figure*}[1][]{\def\@captype{figure}}{}%
  \renewenvironment{table}[1][]{\def\@captype{table}}{}%
  \renewenvironment{table*}[1][]{\def\@captype{table}}{}%
  \renewenvironment{algorithm}[1][]{\def\@captype{algorithm}}{}%
  %
  %
  \let\clearpage\relax
  \let\cleardoublepage\relax
  \let\newpage\relax
  %
  %
  \newsavebox{\@appendixdump}%
  \begin{lrbox}{\@appendixdump}\begin{minipage}{\textwidth}%
  \makeatother
\else
  \clearpage
  \useRomanappendicesfalse
  \appendices
\fi

\section{Forward Invariance and Stability via Deterministic rCBF and rCLF}
\label{app:det}

This appendix provides invariance and stability results that follow from the rCBF and rCLF formulations in \Cref{def:cbf} and \Cref{def:clf}, respectively. These results build on the analysis from \cite{xu2015robustness} for CBFs and \cite{freeman2008robust} for CLFs. We present results for both continuous-time and discrete-time systems. The results below follow from elementary comparison arguments; we provide complete proofs since the exact statements are specific to our formulation.

\subsection{Continuous-Time Systems}
\label{sec:cont_time_re}
The next results are stated for systems of the form \eqref{eq:ct_true}.
\begin{theorem}[Forward Invariance via Deterministic rCBFs]\label{thm:det_ct_cbf}

Let $h(x)$ be a robust CBF on $\Kset$ with decay rate $\gamma$ and margin $r$. Assume that $\mathcal{C}$ is compact and that $\norm{\res(x,u)}\le r$ for all $(x,u)\in \Kset\times\U$. Furthermore, let $u(x)=\pi(x)$ be a continuous function that enforces the CBF constraint \eqref{eq:ct_cbf} for all $x\in\Kset$.  Then, $x(t)\in\Cset$  for all $t\in[0,T]$ if $x(0)\in \mathcal{C}$.
\end{theorem}

\begin{proof}
Since $h$ and $x$ are continuously differentiable, $h(x(t))$ is differentiable and, using \eqref{eq:ct_true}, such that
\begin{align*}
\dot h(x(t))
&= \ip{\nabla h(x(t))}{\hat f(x(t),\pi(x(t)))}\\
&\quad + \ip{\nabla h(x(t))}{\res(x(t),\pi(x(t)))} .
\end{align*}
The rCBF constraint \eqref{eq:ct_cbf} implies
\begin{equation*}
\ip{\nabla h(x(t))}{\hat f(x(t),\pi(x(t)))}+\gamma h(x(t))
\ge \norm{\nabla h(x(t))}\,r .
\end{equation*}
Since $\norm{\res(x,\pi(x))}\le r$ for all $x\in\Kset$, the Cauchy--Schwarz inequality gives
\begin{equation*}
-\ip{\nabla h(x(t))}{\res(x(t),\pi(x(t)))}
\le \norm{\nabla h(x(t))}\,r .
\end{equation*}
Combining the previous two inequalities yields $\dot h(x(t))+\gamma h(x(t))\ge0$ in case that $x(t)\in\Kset$. Hence
\begin{equation*}
\frac{\dd}{\dd t}\bigl[e^{\gamma t}h(x(t))\bigr]
=e^{\gamma t}\bigl(\dot h(x(t))+\gamma h(x(t))\bigr)\ge0 .
\end{equation*}
Thus, $h(x(t))\ge e^{-\gamma t}h(x(0))$ as long as the trajectory $x(t)$ remains in $\Kset$. If $x(0)\in\Cset$, then $x(t)\in\Cset$. Since $\Cset$ is compact and $\Cset\subseteq\Kset$, standard continuation arguments, e.g., \cite[Theorem~3.3]{khalil2002nonlinear}, extend the solution $x(t)$ and the previous reasoning holds for all $t\in[0,T]$. 
\end{proof}

\begin{theorem}[Stability via Deterministic rCLFs]\label{thm:det_ct_clf}
Let $V(x)$ be a robust CLF on $\Kset$  with decay rate $c$ and margin
$r$. Assume that $\mathcal{V}$ is compact and that $\norm{\res(x,u)}\le r$ for all $(x,u)\in \Kset\times\U$. Furthermore, let $u(x)=\pi(x)$ be a continuous function that enforces the CLF constraint \eqref{eq:ct_clf} for all $x\in\Kset$. Then, $\|x(t)\|\le\underline\alpha_V^{-1}\bigl(e^{-ct}\overline\alpha_V(\|x(0)\|)\bigr)$ for all $t\in[0,T]$ if $x(0)\in~\mathcal{V}$.
\end{theorem}

\begin{proof}
Similar to the proof of \Cref{thm:det_ct_cbf}, using \eqref{eq:ct_true} gives
\begin{align*}
\dot V(x(t))
&= \ip{\nabla V(x(t))}{\hat f(x(t),\pi(x(t)))}\\
&\quad + \ip{\nabla V(x(t))}{\res(x(t),\pi(x(t)))} .
\end{align*}
The rCLF constraint \eqref{eq:ct_clf} and the Cauchy--Schwarz inequality imply $\dot V(x(t))+cV(x(t))\le0$ in case that $x(t)\in\Kset$. Therefore
\begin{equation*}
\frac{\dd}{\dd t}\bigl[e^{ct}V(x(t))\bigr]
=e^{ct}\bigl(\dot V(x(t))+cV(x(t))\bigr)\le0,
\end{equation*}
and $V(x(t))\le e^{-ct}V(x(0))$. If $x(0)\in\mathcal V$, then $V(x(t))\le V(x(0))\le V_{\max}$, so the trajectory $x(t)$ remains in the compact set $\mathcal V\subseteq\Kset$ and the previous reasoning holds for all $t\in[0,T]$. The assumptions on $\underline\alpha_V(\cdot)$ and $\overline\alpha_V(\cdot)$ in \Cref{def:clf} give
\begin{equation*}
\underline\alpha_V(\norm{x(t)})\le V(x(t))\le e^{-ct}\overline\alpha_V(\norm{x(0)}),
\end{equation*}
and monotonicity of $\underline\alpha_V(\cdot)$ yields $\|x(t)\|\le\underline\alpha_V^{-1}\bigl(e^{-ct}\overline\alpha_V(\|x(0)\|)\bigr)$ for all $t\in[0,T]$.
\end{proof}

\subsection{Discrete-Time Systems}

The next results are stated for discrete-time systems 
\begin{align}
    x_{t+1}=f(x_t,u_t)=\hat f(x_t,u_t)+\res(x_t,u_t)
    \label{eq:dt_true_det}
\end{align}
where $t=0,\ldots,T-1$ now denotes discrete time. For a set $\Kset\subseteq\X$ and robustness margin $r\ge0$, let $\Kset_r^+\subseteq\X$ be a compact set satisfying
\begin{extendedonly}
\begin{equation}
\Kset_r^+\supseteq
\bigl\{\hat f(x,u)+w:\ (x,u)\in\Kset\times\U,\ \norm{w}\le r\bigr\}.
\label{eq:dt_one_step_tube_det}
\end{equation}
\end{extendedonly}
\begin{compactonly}
\begin{multline}
\Kset_r^+\supseteq
\bigl\{\hat f(x,u)+w:\ (x,u)\in\Kset\times\U,
\norm{w}\le r\bigr\}.
\label{eq:dt_one_step_tube_det}
\end{multline}
\end{compactonly}
This set is introduced only to specify the domain on which the Lipschitz bounds from the definitions below are valid; in particular, it contains both the nominal model $\hat f(x,u)$ and every perturbed nominal model $\hat f(x,u)+w$ with $\norm{w}\le r$.

\begin{definition}[Discrete-time robust CBF]\label{def:dt_cbf}
Let $\Kset\subseteq\X$ be a set and $h:\X\to\R$ be a Lipschitz continuous function. Define the set $\Cset:=\{x\in\X:h(x)\ge0\}$ and assume that $\Cset\subseteq\Kset$. Let $r\ge0$, $\Kset_r^+$ be a compact set satisfying \eqref{eq:dt_one_step_tube_det}, and $h$ be Lipschitz continuous on $\Kset_r^+$ with Lipschitz constant $L_h$. We say that $h(x)$ is a discrete-time robust control barrier function (DT-rCBF) on $\Kset$ with decay rate $\gamma\in(0,1]$ and robustness margin $r$ if, for every $x\in\Kset$, there exists $u\in\U$ satisfying
\begin{equation}
    h\bigl(\hat f(x,u)\bigr)-L_h r\ge (1-\gamma)h(x).
    \label{eq:dt_rcbf_det}
\end{equation}
\end{definition}

\begin{definition}[Discrete-time robust CLF]\label{def:dt_clf}
Let $\Kset\subseteq\X$ be a set, $V:\X\to\R_{\ge0}$ be a Lipschitz continuous and positive-definite function, $V_{\max}>0$ be a constant, and $\underline\alpha_V(\cdot),\overline\alpha_V(\cdot):\R\to\R$ be locally Lipschitz continuous class-$\mathcal K_\infty$ functions such that $\underline\alpha_V(\norm{x})\le V(x)\le\overline\alpha_V(\norm{x})$ for all $x\in\Kset$. Define the set $\mathcal V:=\{x\in\X:V(x)\le V_{\max}\}$ and assume that $\mathcal V\subseteq\Kset$. Let $r\ge0$, $\Kset_r^+$ be a compact set satisfying \eqref{eq:dt_one_step_tube_det}, and $V$ be Lipschitz continuous on $\Kset_r^+$ with Lipschitz constant $L_V$. We say that $V(x)$ is a discrete-time robust control Lyapunov function (DT-rCLF) on $\Kset$ with decay rate $c\in(0,1)$ and robustness margin $r$ if, for every $x\in\Kset$, there exists $u\in\U$ satisfying
\begin{equation}
    V\bigl(\hat f(x,u)\bigr)+L_V r\le (1-c)V(x).
    \label{eq:dt_rclf_det}
\end{equation}
\end{definition}

Similar to the continuous-time results, we assume that the robustness margin $r$ bounds the model error as
\begin{equation}
    \norm{\res(x,u)}\le r,
    \qquad \forall (x,u)\in\Kset\times\U .
    \label{eq:dt_residual_uniform_det}
\end{equation}

\begin{theorem}[Forward invariance via deterministic discrete-time rCBFs]
\label{thm:det_dt_cbf}
Let $h(x)$ be a DT-rCBF on $\Kset$ with decay rate $\gamma$ and margin $r$. Assume that the error bound \eqref{eq:dt_residual_uniform_det} holds. Furthermore, let $u_t=\pi(x_t)$ be a policy that enforces the DT-rCBF constraint \eqref{eq:dt_rcbf_det} for all $x_t\in\Kset$. Then, for the discrete-time system \eqref{eq:dt_true_det}, it holds that $x_t\in\Cset$ for all $t=0,\dots,T$ if $x_0\in\Cset$.
\end{theorem}

\begin{proof}
We argue by induction. The claim is immediate for $t=0$. Suppose that $x_t\in\Cset$ for some $t\in\{0,\dots,T-1\}$. Since $\Cset\subseteq\Kset$, the policy enforces \eqref{eq:dt_rcbf_det} at $x_t$ with $u_t=\pi(x_t)$. Moreover, \eqref{eq:dt_residual_uniform_det} implies $\norm{\res(x_t,u_t)}\le r$, so both $\hat f(x_t,u_t)$ and $x_{t+1}=\hat f(x_t,u_t)+\res(x_t,u_t)$ belong to $\Kset_r^+$ by \eqref{eq:dt_one_step_tube_det}. The Lipschitz continuity of $h$ on $\Kset_r^+$ with constant $L_h$ gives
\begin{align*}
    h(x_{t+1})
    &\ge h\bigl(\hat f(x_t,u_t)\bigr)-L_h\norm{\res(x_t,u_t)}\\
    &\ge h\bigl(\hat f(x_t,u_t)\bigr)-L_h r.
\end{align*}
Applying \eqref{eq:dt_rcbf_det} yields $h(x_{t+1})\ge(1-\gamma)h(x_t)$. Iterating from $t=0$ gives $h(x_t)\ge(1-\gamma)^t h(x_0)$ for all $t=0,\dots,T$. Since $x_0\in\Cset$ implies $h(x_0)\ge0$ and $\gamma\in(0,1]$, it follows that $h(x_t)\ge0$ for all $t=0,\dots,T$, and hence $x_t\in\Cset$ for all $t=0,\dots,T$.
\end{proof}

\begin{theorem}[Stability via deterministic discrete-time rCLFs]
\label{thm:det_dt_clf}
Let $V(x)$ be a DT-rCLF on $\Kset$ with decay rate $c$ and margin $r$. Assume  that the error bound \eqref{eq:dt_residual_uniform_det} holds. Furthermore, let $u_t=\pi(x_t)$ be a policy that enforces the DT-rCLF constraint \eqref{eq:dt_rclf_det} for all $x_t\in\Kset$. Then, for the discrete-time system \eqref{eq:dt_true_det}, it holds that $\norm{x_t}\le \underline\alpha_V^{-1}\!\left((1-c)^t\overline\alpha_V(\norm{x_0})\right)$ for all $t=0,\dots,T$ if $x_0\in\mathcal V$. 
\end{theorem}

\begin{proof}
The proof again proceeds by induction. The claim is immediate for $t=0$. Suppose that $x_t\in\mathcal V$ for some $t\in\{0,\dots,T-1\}$. Since $\mathcal V\subseteq\Kset$, the policy enforces \eqref{eq:dt_rclf_det} at $x_t$ with $u_t=\pi(x_t)$. Moreover, \eqref{eq:dt_residual_uniform_det} implies $\norm{\res(x_t,u_t)}\le r$, so both $\hat f(x_t,u_t)$ and $x_{t+1}=\hat f(x_t,u_t)+\res(x_t,u_t)$ belong to $\Kset_r^+$ by \eqref{eq:dt_one_step_tube_det}. The Lipschitz continuity of $V$ on $\Kset_r^+$ with constant $L_V$ gives
\begin{align*}
    V(x_{t+1})
    &\le V\bigl(\hat f(x_t,u_t)\bigr)+L_V\norm{\res(x_t,u_t)}\\
    &\le V\bigl(\hat f(x_t,u_t)\bigr)+L_V r.
\end{align*}
Applying \eqref{eq:dt_rclf_det} gives $V(x_{t+1})\le(1-c)V(x_t)$. Iterating from $t=0$ gives $V(x_t)\le(1-c)^tV(x_0)$ for all $t=0,\dots,T$. Since $x_0\in\mathcal V$ implies $V(x_0)\le V_{\max}$ and $c\in(0,1)$, it follows that $V(x_t)\le V_{\max}$ for all $t=0,\dots,T$, and hence $x_t\in\mathcal V$ for all $t=0,\dots,T$. Finally, the class-$\mathcal K_\infty$ bounds in \Cref{def:dt_clf} give
\begin{equation*}
\underline\alpha_V(\norm{x_t})\le V(x_t)\le (1-c)^t\overline\alpha_V(\norm{x_0}),
\end{equation*}
and monotonicity of $\underline\alpha_V$ yields the stated norm bound.
\end{proof}

\section{Continuous-time scores and discretization}
\label{app:ct_meas}

\begin{lemma}[Measurability of $s^{\ct}$]
\label{lem:score_measurable}
If the functions $x(t)$, $u(t)$, and $\res(x,u)$ are continuous, then the nonconformity score $s^{\ct}(\tau)=\sup_{t\in[0,T]}\norm{\res(x(t),u(t))}$ is measurable.
\end{lemma}

\begin{proof}
Let $\mathcal T:=C([0,T];\X\times\U)$ denote the set of continuous functions mapping from the domain $[0,T]$ to the domain $\X\times\U$. Furthermore, let $\mathcal T$ be equipped with the supremum norm. For each fixed $t\in[0,T]$, the evaluation function $e_t:\mathcal T\to\X\times\U$ given by $e_t(\tau):=(x(t),u(t))$ is continuous. Since the function $\res$ is continuous, it also follows that the function
\begin{equation*}
    g_t(\tau):=\norm{\res(e_t(\tau))}=\norm{\res(x(t),u(t))}
\end{equation*}
is continuous and hence also measurable for each fixed $t$. Moreover, for every trajectory $\tau\in\mathcal T$, the function $t\mapsto g_t(\tau)$ is continuous on the compact interval $[0,T]$. Therefore, using the density of $[0,T]\cap\mathbb Q$ in $[0,T]$, we have that
\begin{equation*}
    s^{\ct}(\tau)=\sup_{t\in[0,T]}g_t(\tau)
    =\sup_{t\in[0,T]\cap\mathbb Q}g_t(\tau).
\end{equation*}
The right-hand side is a countable supremum of measurable functions, and is therefore measurable; see, e.g., standard closure properties of measurable maps in \cite[Ch.~1]{kallenberg2002foundations}. Thus $s^{\ct}:\mathcal T\to\R_{\ge0}$ is measurable. 
\end{proof}
The previous result implies the following: if $\tau$ is a random trajectory, then $s^{\ct}(\tau)$ is a random variable by composition. Consequently, if $\tau,\tau^{(1)},\dots,\tau^{(n)}$ are i.i.d. random trajectories, then applying the same map $s^{\ct}$ to each trajectory preserves independence and identical distribution, so that $s^{\ct}(\tau),s^{\ct}(\tau^{(1)}),\dots,s^{\ct}(\tau^{(n)})$ are i.i.d. random variables. 

In practice, one rarely has access to datasets that contain $x(t)$ and $u(t)$ for all continuous times $t\in[0,T]$. Indeed, one often only has access to $x(t_k)$ and $u(t_k)$ at sampling times $\{t_k\}_{k=0}^N$ with sampling period $\Delta t$.  When the nonconformity score $s^{\ct}(\cdot)$ is evaluated over such a sampled discrete-time domain
and $\norm{\res(x(t),u(t))}$ is Lipschitz continuous with Lipschitz constant $L$ (later shown to be guaranteed under \Cref{ass:ct_lip}), then a one-step interpolation gives $s^{\ct}(\tau)\le \max_k\norm{\res(x(t_k),u(t_k))} + L\Delta t$, since for any time $t$ there exists a sampling time $t_k$ with $|t-t_k|\le\Delta t$. Under \Cref{ass:ct_lip}, a valid Lipschitz constant is $L = (L_{\res,x} + L_{\res,u}L_\pi)\sup_{(x,u)\in\Oset\times\U}\norm{f(x,u)}$ where $L_{\res,x},L_{\res,u},L_\pi$ and $\Oset$ are explained in Section \ref{sec:algorithm}.

\section{Split CP and Calibration-conditional Adversarially-Robust CP}
\label{app:arcp}

\begin{proof}[Proof of \Cref{thm:split_conformal_validity}]
Recall that $k:=\lceil(1-\alpha)(n+1)\rceil$. By assumption, the test and calibration trajectories are i.i.d. draws from $\mathcal{D}_\pi$ and are defined over $[0,T]$. Hence their nonconformity scores{, defined by \eqref{eq:score_ct},} are well-defined and i.i.d. The rank of the test score $s$ among $\{s,s^{(1)},\dots,s^{(n)}\}$ is uniform on $\{1,\dots,n+1\}$. If this rank is at most $k$, then $s\le s^{[k]}$, and therefore $\Prob_{n+1}(s\le s^{[k]})\ge k/(n+1)\ge 1-\alpha$.
\end{proof}

\begin{proof}[Proof of \Cref{lem:arcp}]
Let $q:=s^{[k]}$ with $k=\lceil(1-\bar\alpha)n\rceil$. By the conditional split conformal guarantee of \cite{duchi2025sample}, stated in the notation of this paper in \cite[Lemma~2]{lindemann2025formal}, we have
\begin{equation}
    \Prob_n\!\left\{\Prob(s\le q\mid D^{\mathrm{cal}})\ge 1-\alpha\right\}\ge 1-\delta .
    \label{eq:sample_conditional_cp_app}
\end{equation}
Let $\mathcal G:=\{D^{\mathrm{cal}}:\Prob(s\le q\mid D^{\mathrm{cal}})\ge1-\alpha\}$. By \eqref{eq:sample_conditional_cp_app}, $\Prob_n\{\mathcal G\}\ge1-\delta$. Fix any calibration dataset in $\mathcal G$. Conditional on this dataset, both $q$ and $M(D^{\mathrm{cal}})$ are deterministic. Since $\tilde s\le s+M(D^{\mathrm{cal}})$ almost surely, the event inclusion
\begin{equation*}
    \{s\le q\}\subseteq\{\tilde s\le q+M(D^{\mathrm{cal}})\}
\end{equation*}
holds conditionally on $D^{\mathrm{cal}}$. Therefore,
\begin{align*}
\Prob\!\left(\tilde s\le s^{[k]}+M(D^{\mathrm{cal}})\mid D^{\mathrm{cal}}\right)
&\ge \Prob(s\le s^{[k]}\mid D^{\mathrm{cal}})\nonumber\\
&\ge 1-\alpha .
\end{align*}
Thus $\mathcal G$ is contained in the event appearing in \Cref{lem:arcp}, and the desired calibration-conditional ARCP statement follows from $\Prob_n\{\mathcal G\}\ge1-\delta$.
\end{proof}

\section{Fixed-policy and continuous-time certificate proofs}
\label{app:event}
Unlike Appendix~\ref{app:det}, which assumes that $\norm{\res(x,u)}\le r$ for all $(x,u)\in\Kset\times\U$, the results in this appendix use only the trajectory-level event $\{s^{\ct}(\tau)\le r\}$ generated by the nonconformity score $s^{\ct}$ in~\eqref{eq:score_ct} and the threshold $r:=s^{[k]}$.

\begin{proof}[Proof of \Cref{thm:baseline_cr}]
Recall that $r:=s^{[k]}$ with $k=\lceil(1-\alpha)(n+1)\rceil$. Next, define the event
\begin{equation*}
    \mathcal E:=\{s^{\ct}(\tau)\le r\}.
\end{equation*}
By \Cref{thm:split_conformal_validity}, we have that $\Prob_{n+1}(\mathcal E)\ge1-\alpha$. On the event $\mathcal E$, the definition of the nonconformity score \eqref{eq:score_ct} gives
\begin{equation*}
    \norm{\res(x(t),u(t))}\le r,
    \qquad t\in[0,T].
\end{equation*}

\emph{Safety.} Since we assume $\mu_{\XX}(\Cset)=1$, it holds that $x(0)\in\Cset$ almost surely. Additionally, since $h(x)$ is a robust CBF and  the policy $\pi$ enforces the CBF constraint  \eqref{eq:ct_cbf}, the same proof of \Cref{thm:det_ct_cbf} --- here with the trajectory-level bound $\norm{\res(x(t),u(t))}\le r$ for all $t\in[0,T]$ --- gives $h(x(t))\ge e^{-\gamma t}h(x(0))\ge 0$ for all $t\in[0,T]$, but now only on the event $\mathcal E$. Hence, we have $x(t)\in\Cset$ for all $t\in[0,T]$ on the event $\mathcal E$ so that we can conclude
\begin{equation*}
\Prob_{n+1}\bigl(x(t)\in\Cset,\ \forall t\in[0,T]\bigr)
\ge\Prob_{n+1}(\mathcal E)\ge1-\alpha .
\end{equation*}

\emph{Stability.} Since we assume $\mu_{\XX}(\mathcal V)=1$, it holds that $x(0)\in\mathcal V$ almost surely. Additionally, since $V(x)$ is a robust CLF and  $\pi$ enforces the CLF constraint \eqref{eq:ct_clf}, the same proof of \Cref{thm:det_ct_clf} --- here with the trajectory-level bound $\norm{\res(x(t),u(t))}\le r$ for all $t\in[0,T]$ --- gives
$\norm{x(t)}\le
\underline\alpha_V^{-1}\bigl(e^{-ct}\overline\alpha_V(\norm{x(0)})\bigr)$
for all $t\in[0,T]$, but now only on the event $\mathcal E$. Hence, we can conclude 
\begin{align*}
    \Prob_{n+1}&(\|x(t)\|\le\underline\alpha_V^{-1}\bigl(e^{-ct}\overline\alpha_V(\|x(0)\|)\bigr),\  \forall t\in[0,T])\\
    &\ge \Prob_{n+1}(\mathcal E)\ge1-\alpha.
\end{align*}
\end{proof}

\begin{proof}[Proof of \Cref{thm:srcr_validity}]
We first define the ``coupled'' nonconformity score
\begin{equation*}
    \tilde s_{j+1}:=s^{\ct}(x_j(0),\pi_{j+1}),
\end{equation*}
at episode $j+1$ using the same initial condition $x_j(0)$ as the nonconformity score $s_j=s^{\ct}(x_j(0),\pi_j)$ at episode $j$.  According to \Cref{ass:domination} we have that
\begin{equation*}
    \tilde s_{j+1}\le s_j+M_{j+1}
\end{equation*}
almost surely with $M_{j+1}:=\rho(\pi_{j+1},\pi_j)$. \Cref{ass:within_episode_iid} guarantees that $s_j,s_j^{(1)},\ldots,s_j^{(n_j)}\sim \mathcal S_{\pi_j}$ are i.i.d. random variables. Since $M_{j+1}$ is a nonnegative, bounded function of $D_j^{\mathrm{cal}}$, \Cref{lem:arcp} implies that
\begin{equation*}
\Prob_{n_j}\Bigl\{
    \Prob\bigl(\tilde s_{j+1}\le q_j+M_{j+1}\mid D_j^{\mathrm{cal}}\bigr)
    \ge1-\alpha
\Bigr\}\ge1-\delta .
\end{equation*}
If $r_{j+1}\ge q_j+M_{j+1}$, it further holds that
\begin{equation*}
\Prob_{n_j}\Bigl\{
    \Prob\bigl(\tilde s_{j+1}\le r_{j+1}\mid D_j^{\mathrm{cal}}\bigr)
    \ge1-\alpha
\Bigr\}\ge1-\delta .
\end{equation*}
Finally, we note that $x_j(0)$ and $x_{j+1}(0)$ are independent of $D_j^{\mathrm{cal}}$ and follow the same distribution  $\mu_{\XX}$ by \Cref{ass:within_episode_iid}. Hence $\tilde s_{j+1}$ and $s_{j+1}=s^{\ct}(x_{j+1}(0),\pi_{j+1})$ follow the same distribution, implying that 
\begin{align*}
    \Prob\bigl(\tilde s_{j+1}\le r_{j+1}\mid D_j^{\mathrm{cal}}\bigr)=\Prob\bigl( s_{j+1}\le r_{j+1}\mid D_j^{\mathrm{cal}}\bigr)
\end{align*}
as well as 
\begin{align*}
    \Prob\bigl(\tilde s_{j+1}\le q_j+M_{j+1}\mid D_j^{\mathrm{cal}}\bigr)=\Prob\bigl( s_{j+1}\le q_j+M_{j+1}\mid D_j^{\mathrm{cal}}\bigr),
\end{align*}
which concludes the proof.
\end{proof}

\begin{proof}[Proof of \Cref{thm:ct_cbf_certificate}]
Define the event $\mathcal E_{j+1}:=\{s_{j+1}\le r_{j+1}\}$. Since the margin $r_{j+1}$ satisfies $r_{j+1} \ge q_j + M_{j+1}$, equation \eqref{eq:srcr_validity} gives
\begin{equation*}
\Prob_{n_j}\Bigl\{
    \Prob\bigl(\mathcal E_{j+1}\mid D_j^{\mathrm{cal}}\bigr)\ge1-\alpha
\Bigr\}\ge1-\delta.
\end{equation*}
On the event $\mathcal E_{j+1}$, the definition of the nonconformity score \eqref{eq:score_ct} gives
\begin{equation*}
\norm{\res(x_{j+1}(t),\pi_{j+1}(x_{j+1}(t)))}\le r_{j+1},
\qquad t\in[0,T].
\end{equation*}

\emph{Safety.} Since we assume $\mu_{\XX}(\Cset)=1$, it holds that $x_{j+1}(0)\in\Cset$ almost surely. Additionally, since $h(x)$ is a robust CBF with margin $r_{j+1}$ and the  policy $\pi_{j+1}$ enforces the CBF constraint \eqref{eq:ct_cbf}, the same proof of \Cref{thm:baseline_cr} gives $h(x_{j+1}(t))\ge e^{-\gamma t}h(x_{j+1}(0))$ for all $t\in[0,T]$ on the event $\mathcal E_{j+1}$. Hence, we have  $x_{j+1}(t)\in\Cset$ for all $t\in[0,T]$ on the event $\mathcal E_{j+1}$, which concludes the proof. 

\emph{Stability.}  Since we assume $\mu_{\XX}(\mathcal V)=1$, it holds that $x_{j+1}(0)\in\mathcal V$ almost surely. Additionally, since $V(x)$ is a robust CLF with margin $r_{j+1}$ and $\pi_{j+1}$ enforces the CLF constraint \eqref{eq:ct_clf}, the same proof of \Cref{thm:baseline_cr} gives
\begin{equation*}
\norm{x_{j+1}(t)}\le
\underline\alpha_V^{-1}\bigl(e^{-ct}\overline\alpha_V(\norm{x_{j+1}(0)})\bigr),
\qquad t\in[0,T],
\end{equation*}
on the event $\mathcal E_{j+1}$. 
\end{proof}
\section{Derivations for \Cref{prop:counterexample_main}}
\label{app:counterexample}

In \Cref{prop:counterexample_main}, the conformal threshold $r$ is computed from trajectories generated by a fixed policy $\pi$ and then used  to obtain the CBF-QP policy $\pi_r^{\mathrm{cbf}}$. Initial states are sampled from a uniform distribution $x(0)\sim\mu_{\XX}=\mathrm{Unif}[0,1]$.

\emph{Safety of the fixed policy $\pi$.} Under $\pi(x)\equiv -u_0$, the true dynamics $\dot x(t)=f(x(t),u(t))=\hat{f}(x(t),u(t))+\varepsilon(x(t),u(t))$ reduce to $\dot x(t)=(1+x(t))u_0$, resulting in trajectories $x(t)=(1+x(0))e^{u_0 t}-1$. Note that the trajectory $x(t)$ is increasing with time, meaning that $x(t)\ge x(0)\ge 0$ for $t\ge 0$ so that every trajectory $x(t)$ remains in $\Cset$ since $x(0)\sim\mathrm{Unif}[0,1]$. Consequently, under the policy $\pi(x)\equiv -u_0$, it holds that $\Prob_{\tau\sim\mathcal D_\pi}\bigl(x(t)\in\Cset,\ \forall t\in[0,T]\bigr)=1$.

\emph{Computation of $r$.} Since the trajectory $x(t)$ under $\pi(x)\equiv -u_0$ is increasing and $\varepsilon(x,u)=-(2+x)u$, the nonconformity score~\eqref{eq:score_ct} is attained at $t=T$ so that $s^{\ct}(\tau) = u_0(2+x(T)) = u_0(1+(1+x(0))e^{u_0T})$. Because $x(0)\sim\mathrm{Unif}[0,1]$, the nonconformity score $s^{\ct}(\tau)$ is uniformly distributed on $[u_0(1+e^{u_0T}),\, u_0(1+2e^{u_0T})]$. Consequently, the analytical $(1-\alpha)$-quantile is $r=u_0(1+(2-\alpha)e^{u_0T})$.

\emph{CBF-QP synthesis.} For fixed values of $x\in\R$, $r\ge 0$, and nominal input $\bar u\in\R$, the  CBF-QP $\min_{u\in\R}\tfrac12(u-\bar u)^2$ subject to $u+\gamma x\ge r$ has feasible set $[r-\gamma x,\infty)$. It therefore follows that the unique minimizer is
\begin{equation}
  u^\star(x)=\max\{\bar u,\ r-\gamma x\}.
  \label{eq:ce_scalar_qp_solution}
\end{equation}

Since $h(x)=x$, $\nabla h(x)=1$, and $\hat f(x,u)=u$, the rCBF constraint~\eqref{eq:ct_cbf} reduces to $u+\gamma x\ge r$. Applying~\eqref{eq:ce_scalar_qp_solution} with $\bar u=-u_0$ gives $\pi_r^{\mathrm{cbf}}(x)=\max\{-u_0,\, r-\gamma x\}$. When $0<\gamma<r/2$, we have $r-\gamma x\ge r-\gamma > r/2 > 0 > -u_0$ for all $x\in[0,1]$, so that $\pi_r^{\mathrm{cbf}}(x)=r-\gamma x$ for $x\in [0,1]$.

\emph{Nonconformity score violation.} The nonconformity score $s^{\ct}(\tau)$ under $\pi_r^{\mathrm{cbf}}(x)$ exceeds $r$ already at time $t=0$. Indeed, since $u(0):=\pi_r^{\mathrm{cbf}}(x(0))=r-\gamma x(0)$ for $x(0)\in[0,1]$,
\begin{align*}
    s^{\ct}(\tau)
    &\ge\bigl|\varepsilon\bigl(x(0),u(0)\bigr)\bigr|\\
    &=(2+x(0))(r-\gamma x(0))\\
    &\ge 2(r-\gamma)>r,
\end{align*}
where the last step uses $\gamma<r/2$. Since this holds for every $x(0)\in[0,1]$, it follows that $\Prob_{\tau\sim\mathcal D_{\pi_r^{\mathrm{cbf}}}}(s^{\ct}(\tau)\le r)=0$ under the policy $\pi_r^{\mathrm{cbf}}(x)$.

\emph{Safety violation.} Let $t_{\mathrm{hit}}:=\inf\{t\ge 0: x(t)\le 0\}$ denote the first time the trajectory reaches the boundary $h(x)=0$. If $x(0)=0$, then $\dot x(0)=-r<0$, so that the trajectory leaves $\Cset$ immediately. Now, consider $x(0)\in(0,1]$. For every $t<t_{\mathrm{hit}}$, we have $x(t)>0$ by the definition of $t_{\mathrm{hit}}$. While $x(t)\in(0,1]$, the closed-loop dynamics satisfy
\begin{equation*}
    \dot x(t)=-(1+x(t))(r-\gamma x(t))<-(r-\gamma),
\end{equation*}
because $(1+x)(r-\gamma x)>r-\gamma$ for every $x\in(0,1]$ and $0<\gamma<r/2$. Therefore, we see that $x(t)<x(0)-(r-\gamma)t$ for all $t<t_{\mathrm{hit}}$, which implies that
\begin{equation*}
    t_{\mathrm{hit}}<\frac{x(0)}{r-\gamma}\le \frac{1}{r-\gamma}.
\end{equation*}
If $T(r-\gamma)\ge1$, then every trajectory with $x(0)\in[0,1]$ leaves $\Cset$ during $[0,T]$. Consequently, it follows that
$\Prob_{\tau\sim\mathcal D_{\pi_r^{\mathrm{cbf}}}}(x(t)\in\Cset,\,\forall t\in[0,T])=0$ 
under the policy $\pi_r^{\mathrm{cbf}}(x)$

\textcolor{black}{We refer to \Cref{fig:prop1-large} for plots of the trajectories, nonconformity scores, and model errors.}

\section{Controller Sensitivity via Parametric QP Analysis}
\label{app:qp}

This appendix justifies the constants $L_\pi$ and $L_U$ used in \Cref{ass:ct_lip,ass4}. The robust CBF-QP in \eqref{eq:cbf_qp} and the robust CLF-QP in \eqref{eq:clf_qp} are instances of the parametric quadratic program
\begin{equation}
\begin{aligned}
    u^*(x,r)\in\arg\min_{u\in\R^{m}}\;&\frac12u^\top H(x)u+c(x)^\top u\\
    \mathrm{s.t.}\;&A(x)u\le b(x)+r d(x),
\end{aligned}
\label{eq:parametric_qp_app}
\end{equation}
where $x\in\Oset$, $r\in\Rset$, $H(x)\in\R^{m\times m}$, $A(x)\in\R^{p\times m}$, and $b(x),d(x)\in\R^p$. The policy is $\pi_r(x):=u^*(x,r)$. For an index set $\mathcal J\subseteq\{1,\ldots,p\}$, $A_{\mathcal J}(x)$, $b_{\mathcal J}(x)$, and $d_{\mathcal J}(x)$ denote the rows or components indexed by $\mathcal J$.

For a primal-dual solution $(u^*(x,r),\lambda^*(x,r))$, define
\begin{align*}
    c_\ell(x,r)&:=A_\ell(x)u^*(x,r)-b_\ell(x)-r d_\ell(x),\nonumber\\
    \mathcal A(x,r)&:=\{\ell\in\{1,\ldots,p\}:c_\ell(x,r)=0\}.
\end{align*}
On a region where the active set is a fixed set $\mathcal J$, the active KKT equations are
\begin{equation}
\begin{aligned}
&\begin{bmatrix}
H(x)&A_{\mathcal J}(x)^\top\\
A_{\mathcal J}(x)&0
\end{bmatrix}
\begin{bmatrix}
u^*(x,r)\\ \lambda_{\mathcal J}^*(x,r)
\end{bmatrix}\\
&\quad+
\begin{bmatrix}
c(x)\\ -b_{\mathcal J}(x)-r d_{\mathcal J}(x)
\end{bmatrix}=0 .
\end{aligned}
\label{eq:active_kkt_system_app}
\end{equation}
We write
\begin{equation*}
    K_{\mathcal J}(x):=
\begin{bmatrix}
H(x)&A_{\mathcal J}(x)^\top\\
A_{\mathcal J}(x)&0
\end{bmatrix}
\end{equation*}
for the active-set KKT matrix.

\begin{assumption}[Strong regularity of the parametric QP]
\label{ass:qp_strong_regularity_app}
On the compact domain $\Oset\times\Rset$, the following hold: (i) there exists $\mu_H>0$ such that $H(x)\succeq\mu_H I$ for all $x\in\Oset$; (ii) the data $H,c,A,b,d$ are Lipschitz continuous and uniformly bounded on $\Oset$; and (iii) LICQ and strict complementarity hold at every solution of \eqref{eq:parametric_qp_app}. Equivalently, the KKT system \eqref{eq:active_kkt_system_app} is strongly regular uniformly over $\Oset\times\Rset$.
\end{assumption}
These are standard sufficient conditions for local single-valuedness and Lipschitz continuity of parametric QP solution maps; see \cite[Ch.~4]{bonnans2000perturbation}, \cite[Thm.~2.1]{dontchev2014implicit}, and \cite{mestres2025regularity}.

\begin{proposition}[Lipschitz continuity of $\pi_r(x)$ in $x$]
\label{prop:Lux}
Under \Cref{ass:qp_strong_regularity_app}, there exists $L_\pi\ge0$ such that, for every $r\in\Rset$ and all $x,x'\in\Oset$, it holds that
\begin{equation}
    \norm{u^*(x,r)-u^*(x',r)}\le L_\pi\norm{x-x'} .
    \label{eq:qp_lipschitz_x_app}
\end{equation}
\end{proposition}

\begin{proof}
By \Cref{ass:qp_strong_regularity_app}, the KKT generalized equation associated with \eqref{eq:parametric_qp_app} is strongly regular at every $(x,r)\in\Oset\times\Rset$. The cited sensitivity results for strongly regular generalized equations imply that the primal-dual solution map is locally single-valued and locally Lipschitz in the parameter $x$. Compactness of $\Oset\times\Rset$ yields a finite subcover of these local neighborhoods, and hence a uniform Lipschitz constant $L_\pi$ for the primal component. This proves \eqref{eq:qp_lipschitz_x_app}.
\end{proof}

\begin{proposition}[Lipschitz continuity of $\pi_r(x)$ in $r$]
\label{prop:LU_def}
Under \Cref{ass:qp_strong_regularity_app}, there exists $L_U\ge0$ such that, for all $r,r'\in\Rset$, it holds that
\begin{equation}
    \CzNorm{\pi_r-\pi_{r'}}{\Oset}\le L_U |r-r'| .
    \label{eq:qp_lipschitz_policy}
\end{equation}
\end{proposition}

\begin{proof}
Fix $x\in\Oset$ and consider an interval in $r$ on which the active set is constant and equal to $\mathcal J$. Differentiating the KKT system \eqref{eq:active_kkt_system_app} with respect to $r$ gives
\begin{equation*}
    K_{\mathcal J}(x)
    \begin{bmatrix}
        \partial_r u^*(x,r)\\
        \partial_r\lambda_{\mathcal J}^*(x,r)
    \end{bmatrix}
    =
    \begin{bmatrix}
        0\\ d_{\mathcal J}(x)
    \end{bmatrix} .
\end{equation*}
LICQ and $H(x)\succeq\mu_HI$ imply that $K_{\mathcal J}(x)$ is nonsingular for every active set that occurs under \Cref{ass:qp_strong_regularity_app}. Hence
\begin{equation}
    \norm{\partial_r u^*(x,r)}
    \le
    \norm{K_{\mathcal J}(x)^{-1}}
    \norm{\begin{bmatrix}0\\ d_{\mathcal J}(x)\end{bmatrix}} .
    \label{eq:dr_u_bound_app}
\end{equation}
Let $\mathfrak A(x,r)$ denote the finite collection of active sets of the strongly regular regions whose closures contain $(x,r)$. Define
\begin{equation}
    L_U:=\sup_{(x,r)\in\Oset\times\Rset}\;\sup_{\mathcal J\in\mathfrak A(x,r)}
    \norm{K_{\mathcal J}(x)^{-1}}
    \norm{\begin{bmatrix}0\\ d_{\mathcal J}(x)\end{bmatrix}} .
    \label{eq:LU_constant_app}
\end{equation}
The quantity in \eqref{eq:LU_constant_app} is finite because the domain is compact, $d$ is bounded, and strong regularity prevents singular active-set KKT matrices on active regions. Along the segment between any $r,r'\in\Rset$, the map $r\mapsto u^*(x,r)$ is continuous and piecewise continuously differentiable with finitely many active-set changes. Integrating \eqref{eq:dr_u_bound_app} over those regions yields
\begin{equation*}
    \norm{u^*(x,r)-u^*(x,r')}\le L_U|r-r'|,
    \qquad x\in\Oset .
\end{equation*}
Taking the supremum over $x\in\Oset$ gives \eqref{eq:qp_lipschitz_policy}.
\end{proof}

\Cref{prop:Lux} provides $L_\pi$ for \Cref{ass:ct_lip}, while \Cref{prop:LU_def} provides $L_U$ for \Cref{ass4}. Together with the trajectory-sensitivity constant $\beta_T$ in \eqref{eq:betaT_main}, these constants define $\kappa=\beta_TL_U$ in Appendix~\ref{app:sensitivity}.

\section{Nonconformity Score Sensitivity Analysis and Derivation of \texorpdfstring{$\beta_T$}{βT} }
\label{app:sensitivity}

This appendix derives the nonconformity score sensitivity constant $\beta_T$ based on which we later define $\kappa=\beta_TL_U$. Throughout, $\Oset$ denotes the compact set from \Cref{ass:ct_lip} and controllers belong to $\{\pi_r\}_{r\in\Rset}$. For two trajectories sharing the same initial condition $x_0\in\XX$ under policies $\pi_r$ and $\pi_{r'}$, we write $x(t):=x(t;x_0,\pi_r)$ and $x'(t):=x(t;x_0,\pi_{r'})$ with corresponding inputs $u(t):=\pi_r(x(t))$ and $u'(t):=\pi_{r'}(x'(t))$. Both trajectories exist on $[0,T]$ and remain in $\Oset$ by \Cref{ass:ct_lip}. We also use the notation $s(x_0,\pi_r):=s^{\ct}(\tau(x_0,\pi_r))$ for the nonconformity score~\eqref{eq:score_ct}.

\subsection{Closed-loop state deviation}

\begin{lemma}\label{lem:state_dev}
Under \Cref{ass:ct_lip}, let $x_0\in \XX$, $r,r'\in \Rset$, and $\Lambda_x:=L_x+L_uL_\pi$. Then for all $t\in[0,T]$, we have
\begin{equation}
  \norm{x'(t)-x(t)}
  \le \frac{L_u}{\Lambda_x}\bigl(e^{\Lambda_x t}-1\bigr)\,\CzNorm{\pi_{r'}-\pi_r}{\Oset},
  \label{eq:ct_state_sensitivity}
\end{equation}
with the convention $(e^{\Lambda_x t}-1)/\Lambda_x=t$ when $\Lambda_x=0$.
\end{lemma}

\begin{proof}
Define the state deviation $\Delta(t):=x'(t)-x(t)$. Since both trajectories satisfy the system dynamics~\eqref{eq:ct_true}, we have $\dot\Delta(t)= f(x'(t),u'(t))- f(x(t),u(t))$. The Lipschitz property of $f$ on $\Oset\times\U$ (see \Cref{ass:ct_lip}(i)) gives
\[
\norm{\dot\Delta(t)}\le L_x\norm{x'(t)-x(t)}+L_u\norm{u'(t)-u(t)}.
\]
For the deviation in control inputs, we obtain
\begin{extendedonly}
\[
\norm{u'(t)-u(t)}=\norm{\pi_{r'}(x'(t))-\pi_r(x(t))}\le \norm{\pi_{r'}(x'(t))-\pi_r(x'(t))}+\norm{\pi_r(x'(t))-\pi_r(x(t))}.
\]
\end{extendedonly}
\begin{compactonly}
\begin{multline*}
\norm{u'(t)\!-\!u(t)}=\norm{\pi_{r'}(x'(t))\!-\!\pi_r(x(t))}\\
\le \norm{\pi_{r'}(x'(t))\!-\!\pi_r(x'(t))}\!+\!\norm{\pi_r(x'(t))\!-\!\pi_r(x(t))}.
\end{multline*}
\end{compactonly}
Since $x'(t)\in\Oset$, the first term satisfies $\norm{\pi_{r'}(x'(t))-\pi_r(x'(t))}\le \sup_{x\in\Oset}\norm{\pi_{r'}(x)-\pi_r(x)}=\CzNorm{\pi_{r'}-\pi_r}{\Oset}$. The Lipschitz property of $\pi_r$ on $\Oset$ (see \Cref{ass:ct_lip}(iii)) bounds the second term as $\norm{\pi_r(x'(t))-\pi_r(x(t))}\le L_\pi\norm{x'(t)-x(t)}$. Substituting both bounds into the dynamics inequality yields
\begin{extendedonly}
\[
\norm{\dot\Delta(t)}\le (L_x+L_uL_\pi)\norm{\Delta(t)}+L_u\CzNorm{\pi_{r'}-\pi_r}{\Oset}=\Lambda_x\norm{\Delta(t)}+L_u\CzNorm{\pi_{r'}-\pi_r}{\Oset}.
\]
\end{extendedonly}
\begin{compactonly}
\begin{align*}
\norm{\dot\Delta(t)}&\le (L_x+L_uL_\pi)\norm{\Delta(t)}+L_u\CzNorm{\pi_{r'}-\pi_r}{\Oset}\\
&= \Lambda_x\norm{\Delta(t)}+L_u\CzNorm{\pi_{r'}-\pi_r}{\Oset}.
\end{align*}
\end{compactonly}
Since $\Delta(0)=0$, Gr\"onwall's inequality gives
\begin{align*}
\norm{\Delta(t)}
&\le L_u\CzNorm{\pi_{r'}-\pi_r}{\Oset}
   \int_0^t e^{\Lambda_x(t-\tau)}\,\dd\tau \\
&= \frac{L_u}{\Lambda_x}(e^{\Lambda_x t}-1)\CzNorm{\pi_{r'}-\pi_r}{\Oset},
\end{align*}
which is exactly~\eqref{eq:ct_state_sensitivity} and completes the proof.
\end{proof}

\begin{remark}[Discrete-time counterpart]\label{rem:dt_state_dev}
In discrete time, we define $\Delta_{t+1}:=f(x'_t,u'_t)-f(x_t,u_t)$. The same Lipschitz continuity and triangle-inequality argument gives $\norm{\Delta_{t+1}}\le \Lambda_x\norm{\Delta_t}+L_u\CzNorm{\pi_{r'}-\pi_r}{\Oset}$ with $\Delta_0=0$. Unrolling the recursion gives $\norm{\Delta_t}\le L_u\CzNorm{\pi_{r'}-\pi_r}{\Oset}\sum_{m=0}^{t-1}\Lambda_x^{m}$.
\end{remark}

\subsection{Nonconformity score sensitivity}

\begin{proposition}\label{prop:betaT_summary}
Under \Cref{ass:ct_lip}, for every $x_0\in\XX$ and $r,r'\in\Rset$, we have
\begin{equation*}
|s(x_0,\pi_{r'})-s(x_0,\pi_r)|\le \beta_T\,\CzNorm{\pi_{r'}-\pi_r}{\Oset},
\end{equation*}
where $\beta_T := L_{\res,u} +(L_{\res,x}+L_{\res,u}L_\pi)\frac{L_u}{\Lambda_x}(e^{\Lambda_x T}-1)$.

\end{proposition}

\begin{proof}
Since $|\sup_t a(t) - \sup_t b(t)|\le \sup_t |a(t)-b(t)|$ for any two bounded functions $a(\cdot)$ and $b(\cdot)$, we have  $|s(x_0,\pi_{r'})-s(x_0,\pi_r)| \le \sup_{t\in[0,T]}|\norm{\res(x'(t),u'(t))}-\norm{\res(x(t),u(t))}|$. The reverse triangle inequality then gives us $|\norm{\res(x'(t),u'(t))}-\norm{\res(x(t),u(t))}|\le \norm{\res(x'(t),u'(t))-\res(x(t),u(t))}$. Since $\res(x,u)=f(x,u)-\hat f(x,u)$, \Cref{ass:ct_lip}(i)--(ii) imply that $\res(x,u)$ is Lipschitz continuous on $\Oset\times\U$ with Lipschitz constants $L_{\res,x}\le L_x+L_{\hat f,x}$ and $L_{\res,u}\le L_u+L_{\hat f,u}$, so that
\begin{extendedonly}
\[
\norm{\res(x'(t),u'(t))-\res(x(t),u(t))}\le L_{\res,x}\norm{x'(t)-x(t)}+L_{\res,u}\norm{u'(t)-u(t)}.
\]
\end{extendedonly}
\begin{compactonly}
\begin{multline*}
\norm{\res(x'(t),u'(t))-\res(x(t),u(t))}\\
\le L_{\res,x}\norm{x'(t)\!-\!x(t)}+L_{\res,u}\norm{u'(t)\!-\!u(t)}.
\end{multline*}
\end{compactonly}
For the deviation in control inputs, the same  argument as in the proof of \Cref{lem:state_dev} gives us $\norm{u'(t)-u(t)}\le \CzNorm{\pi_{r'}-\pi_r}{\Oset}+L_\pi\norm{x'(t)-x(t)}$. Substituting this results in
\begin{extendedonly}
\[
\norm{\res(x'(t),u'(t))-\res(x(t),u(t))}\le (L_{\res,x}+L_{\res,u}L_\pi)\norm{x'(t)-x(t)}+L_{\res,u}\CzNorm{\pi_{r'}-\pi_r}{\Oset}.
\]
\end{extendedonly}
\begin{compactonly}
\begin{align*}
&\norm{\res(x'(t),u'(t))-\res(x(t),u(t))}\\
&\quad\le (L_{\res,x}+L_{\res,u}L_\pi)\norm{x'(t)-x(t)}\\
&\qquad+L_{\res,u}\CzNorm{\pi_{r'}-\pi_r}{\Oset}.
\end{align*}
\end{compactonly}
Inserting the state deviation bound~\eqref{eq:ct_state_sensitivity} and noting that $t\mapsto (e^{\Lambda_x t}-1)/\Lambda_x$ is nondecreasing, the supremum of the previous expression over $t\in[0,T]$ is attained at $t=T$, i.e.,
\begin{extendedonly}
\[
\sup_{t\in[0,T]}\norm{\res(x'(t),u'(t))-\res(x(t),u(t))}\le \Bigl[L_{\res,u}+(L_{\res,x}+L_{\res,u}L_\pi)\frac{L_u}{\Lambda_x}(e^{\Lambda_x T}-1)\Bigr]\CzNorm{\pi_{r'}-\pi_r}{\Oset},
\]
\end{extendedonly}
\begin{compactonly}
\begin{align*}
&\sup_{t\in[0,T]}\norm{\res(x'(t),u'(t))-\res(x(t),u(t))}\\
&\quad\le \Bigl[L_{\res,u}+(L_{\res,x}+L_{\res,u}L_\pi)\tfrac{L_u}{\Lambda_x}(e^{\Lambda_x T}\!-\!1)\Bigr]\\
&\qquad\cdot\CzNorm{\pi_{r'}-\pi_r}{\Oset},
\end{align*}
\end{compactonly}
which equals $\beta_T\CzNorm{\pi_{r'}-\pi_r}{\Oset}$. 
\end{proof}

\begin{remark}[Discrete-time counterpart]
In discrete time, we proceed similarly but instead use the discrete-time state deviation bound from \Cref{rem:dt_state_dev} while replacing the supremum by the maximum over $t=0,\dots,T-1$. The discrete-time counterpart of \Cref{prop:betaT_summary} then simply replaces $(e^{\Lambda_x T}-1)/\Lambda_x$ by $\sum_{j=0}^{T-1}\Lambda_x^j$. 
\end{remark}

\subsection{From \texorpdfstring{$\beta_T$}{betaT} and \texorpdfstring{$L_U$}{LU} to \texorpdfstring{$\kappa$}{kappa}}

The margin-to-policy bound $\CzNorm{\pi_r-\pi_{r'}}{\Oset}\le L_U|r-r'|$ from Appendix~\ref{app:qp} (\Cref{prop:LU_def}) gives $\CzNorm{\pi_r-\pi_{r'}}{\Oset}\le L_U|r-r'|$. Combining this with \Cref{prop:betaT_summary} gives
\[
|s(x_0,\pi_{r'})-s(x_0,\pi_r)|\le \beta_TL_U|r'-r|=\kappa\,|r'-r|,
\]
which is the bound in~\eqref{eq:kappa_score_shift} with  constant $\kappa := \beta_TL_U$.
\ifextended\input{appendices_short/J_joint_clf_cbf_qp}\fi                 

\section{Discrete-time certificates}
\label{app:dt_certificates}

This appendix states the discrete-time safety and stability counterparts of \Cref{thm:ct_cbf_certificate}. For a discrete-time trajectory $\tau=(x_0,\ldots,x_T,u_0,\ldots,u_{T-1})$ generated by the discrete-time system  in \eqref{eq:dt_true_det} under a policy $u_t=\pi(x_t)$ for $t=0,\hdots,T-1$, define the discrete-time nonconformity score
\begin{equation}
    s^{\dt}(\tau):=\max_{t=0,\ldots,T-1}\norm{\res(x_t,u_t)}.
    \label{eq:score_dt}
\end{equation}
At episode $j+1$, we write $s_{j+1}^{\dt}:=s^{\dt}(\tau_{j+1})$ for the test nonconformity score under the deployed policy $\pi_{j+1}$. We make the same assumptions as in \Cref{thm:srcr_validity}, but now instead for the discrete-time nonconformity score \eqref{eq:score_dt}, and select the radius $r_{j+1}$ such that $r_{j+1}\ge q_j+M_{j+1}$ so that again
\begin{equation}
\Prob_{n_j}\Bigl\{
\Prob\bigl(s_{j+1}^{\dt}\le r_{j+1}\mid D_j^{\mathrm{cal}}\bigr)
\ge 1-\alpha
\Bigr\}\ge 1-\delta .
\label{eq:dt_srcr_validity}
\end{equation}
On the event $\{s_{j+1}^{\dt}\le r_{j+1}\}$, the error bound $\norm{\res(x_t,u_t)}\le r_{j+1}$ holds for every $t=0,\ldots,T-1$, so that \Cref{thm:det_dt_cbf,thm:det_dt_clf} apply trajectory-wise, similar to continuous time.

\subsection{DT iterative conformal rCBF safety certificate}

\begin{theorem}[Discrete-time conformal safety certificate]
\label{thm:dt_cbf_certificate}
Let $h(x)$ be a DT-rCBF on $\Kset$ with decay rate $\gamma\in(0,1]$ and margin $r_{j+1}$.  Furthermore, let $u_t=\pi_{j+1}(x_t)$ be a function that enforces the DT-rCBF constraint \eqref{eq:dt_rcbf_det} for all $x_t\in\Kset$. If $\mu_{\XX}(\Cset)=1$ and $r_{j+1}$ satisfies \eqref{eq:dt_srcr_validity}, then
\begin{extendedonly}
\begin{equation}
  \Prob_{n_j}\Bigl\{
    \Prob\bigl(x_{j+1,t}\in\Cset,\;\forall t=0,\ldots,T \mid D_j^{\mathrm{cal}}\bigr)
    \ge 1-\alpha
  \Bigr\}\ge 1-\delta .
  \label{eq:dt_cbf_certificate_set}
\end{equation}
\end{extendedonly}
\begin{compactonly}
\begin{multline}
  \Prob_{n_j}\Bigl\{
    \Prob\bigl(x_{j+1,t}\in\Cset,\;\forall t=0,\ldots,T \mid D_j^{\mathrm{cal}}\bigr)\\
    \ge 1-\alpha
  \Bigr\}\ge 1-\delta .
  \label{eq:dt_cbf_certificate_set}
\end{multline}
\end{compactonly}
\end{theorem}

\begin{proof}
Define $\mathcal E_{j+1}^{\dt}:=\{s_{j+1}^{\dt}\le r_{j+1}\}$. On the event $\mathcal E_{j+1}^{\dt}$, the definition \eqref{eq:score_dt} gives $\norm{\res(x_t,u_t)}\le r_{j+1}$ for all times $t=0,\ldots,T-1$. Since $\mu_{\XX}(\Cset)=1$, the initial condition satisfies $x_{j+1,0}\in\Cset$ almost surely. Therefore the proof of \Cref{thm:det_dt_cbf}, with the uniform error bound replaced by the trajectory-level error bound on $\mathcal E_{j+1}^{\dt}$ and applied with margin $r_{j+1}$ and policy $\pi_{j+1}$, gives $x_{j+1,t}\in\Cset$ for all $t=0,\ldots,T$ on $\mathcal E_{j+1}^{\dt}$. Hence, it holds that
\begin{equation*}
\mathcal E_{j+1}^{\dt}\subseteq
\{x_{j+1,t}\in\Cset,\;\forall t=0,\ldots,T\}.
\end{equation*}
Applying \eqref{eq:dt_srcr_validity} yields the statement in \eqref{eq:dt_cbf_certificate_set}.
\end{proof}

\subsection{DT iterative conformal rCLF stability certificate}

\begin{theorem}[Discrete-time conformal stability certificate]
\label{thm:dt_clf_certificate}
Let $V(x)$ be a DT-rCLF on $\Kset$ with decay rate $c\in(0,1)$ and margin $r_{j+1}$.  Furthermore, let $u_t=\pi_{j+1}(x_t)$ be a function that enforces the DT-rCLF constraint \eqref{eq:dt_rclf_det} for all $x_t\in\Kset$. If $\mu_{\XX}(\mathcal V)=1$ and $r_{j+1}$ satisfies \eqref{eq:dt_srcr_validity}, then
\begin{extendedonly}
\begin{equation}
\Prob_{n_j}\Bigl\{
\Prob\bigl(\norm{x_{j+1,t}}\le \underline\alpha_V^{-1}((1-c)^t\overline\alpha_V(\norm{x_{j+1,0}})),\;\forall t=0,\ldots,T \mid D_j^{\mathrm{cal}}\bigr)
\ge 1-\alpha
\Bigr\}\ge 1-\delta .
\label{eq:dt_clf_certificate_state}
\end{equation}
\end{extendedonly}
\begin{compactonly}
\begin{multline}
\Prob_{n_j}\Bigl\{
\Prob\bigl(\norm{x_{j+1,t}}\le \underline\alpha_V^{-1}\!\bigl((1-c)^t\overline\alpha_V(\norm{x_{j+1,0}})\bigr),\\
\forall t=0,\ldots,T \mid D_j^{\mathrm{cal}}\bigr)
\ge 1-\alpha
\Bigr\}\ge 1-\delta .
\label{eq:dt_clf_certificate_state}
\end{multline}
\end{compactonly}
\end{theorem}

\begin{proof}
Define $\mathcal E_{j+1}^{\dt}:=\{s_{j+1}^{\dt}\le r_{j+1}\}$. On the event $\mathcal E_{j+1}^{\dt}$, the definition \eqref{eq:score_dt} gives $\norm{\res(x_t,u_t)}\le r_{j+1}$ for all $t=0,\ldots,T-1$. Since $\mu_{\XX}(\mathcal V)=1$, the initial condition satisfies $x_{j+1,0}\in\mathcal V$ almost surely. Therefore the proof of \Cref{thm:det_dt_clf}, with the uniform residual bound replaced by the trajectory-level residual bound on $\mathcal E_{j+1}^{\dt}$ and applied with margin $r_{j+1}$ and policy $\pi_{j+1}$, gives
\begin{equation*}
V(x_{j+1,t})\le(1-c)^tV(x_{j+1,0}),
\qquad t=0,\ldots,T,
\end{equation*}
on $\mathcal E_{j+1}^{\dt}$. The class-$\mathcal K_\infty$ bounds from \Cref{def:dt_clf} imply
\begin{equation*}
\underline\alpha_V(\norm{x_{j+1,t}})
\le V(x_{j+1,t})
\le (1-c)^t\overline\alpha_V(\norm{x_{j+1,0}})
\end{equation*}
on $\mathcal E_{j+1}^{\dt}$. Monotonicity of $\underline\alpha_V$ and applying \eqref{eq:dt_srcr_validity} yields the statement in \eqref{eq:dt_clf_certificate_state}.
\end{proof}


\section{Convergence analysis}
\label{app:convergence}

This appendix proves \Cref{thm:convergence} and \Cref{cor:hp_tracking_convergence}.

\subsection{Quantile perturbations}

This subsection uses the nonconformity score variables $S_r$, their population quantiles $Q_r(p)$, and their cumulative distribution functions $F_r$ as defined in \eqref{eq:score_random_variable_r}, \eqref{eq:population_quantile_r}, and \eqref{eq:cdf_score_r_main}, respectively. For the quantile perturbation lemma presented below, we use the same subscript convention: for a random variable $Z$, we let
\begin{equation*}
\begin{aligned}
    F_Z(z)&:=\Prob(Z\le z),\\
    Q_Z(p)&:=\inf\{z\in\R:F_Z(z)\ge p\},\qquad p\in(0,1).
\end{aligned}
\end{equation*}
In particular, \eqref{eq:population_quantile_r} is the specialization $Q_r(p)=Q_{S_r}(p)$.

\begin{lemma}[Quantile perturbation bound]
\label{lem:quantile_lip}
Let $X$ and $Y$ be real-valued random variables. If $|X-Y|\le c$ almost surely for some $c\ge0$, then
\begin{equation}
    |Q_X(p)-Q_Y(p)|\le c,
    \qquad p\in(0,1).
    \label{eq:quantile_lip_generic_app}
\end{equation}
\end{lemma}

\begin{proof}
The inequality $Y\le X+c$ almost surely implies $F_Y(z)\ge F_X(z-c)$ for every $z\in\R$. Hence $Q_Y(p)\le Q_X(p)+c$. Interchanging $X$ and $Y$ gives $Q_X(p)\le Q_Y(p)+c$. Combining the two inequalities proves \eqref{eq:quantile_lip_generic_app}.
\end{proof}

\begin{corollary}
\label{cor:Q_lip}
Let \Cref{ass:ct_lip,ass4} hold. Then the  population quantiles in \eqref{eq:population_quantile_r} satisfy the Lipschitz bound \eqref{eq:Q_Lipschitz}; that is, we have that
\begin{equation}
    |Q_{r'}(1-\alpha)-Q_r(1-\alpha)|\le\kappa |r'-r|,
    \qquad r,r'\in\Rset .
    \label{eq:Q_lip_app}
\end{equation}
\end{corollary}

\begin{proof}
By \eqref{eq:kappa_score_shift}, the coupled nonconformity score variables $S_r$ and $S_{r'}$ from \eqref{eq:score_random_variable_r}, generated from the same initial condition $x(0)\sim\mu_{\XX}$, satisfy $|S_{r'}-S_r|\le\kappa |r'-r|$ almost surely. Applying \Cref{lem:quantile_lip} with $X=S_{r'}$, $Y=S_r$, $c=\kappa |r'-r|$, and $p=1-\alpha$ gives \eqref{eq:Q_lip_app}.
\end{proof}

\subsection{Quantile error control}

For episode $j$, let us define the empirical cumulative distribution function as
\begin{equation*}
    F_{n_j,j}(z):=\frac{1}{n_j}\sum_{i=1}^{n_j}\ind\{s_j^{(i)}\le z\},
\end{equation*}
where $z\in\R$ and where the calibration nonconformity scores $s_j^{(1)},\ldots,s_j^{(n_j)}$ are i.i.d. copies of $S_{r_j}$ and hence have cumulative distribution function $F_{r_j}$. Let us define
\begin{equation}
    p_j:=1-\bar\alpha_j,
    \qquad
    k_j:=\lceil p_j n_j\rceil,
    \label{eq:pj_kj_definition}
\end{equation}
so that $q_j=s_j^{[k_j]}$ is the empirical quantile at the tightened quantile level $p_j$. We recall that the level $p_j$ is selected according to the ARCP tightening from \Cref{thm:srcr_validity}; the tolerance $\Delta$ used below is independent of $p_j$ except for the requirement that $p_j-\Delta$ and $p_j+\Delta$ are valid quantile levels. By the order-statistic definition of $q_j$ in \eqref{eq:pj_kj_definition}, we know that
\begin{equation}
\begin{aligned}
    F_{n_j,j}(q_j)&\ge \frac{k_j}{n_j}\ge p_j,\\
    F_{n_j,j}(z)&\le \frac{k_j-1}{n_j}<p_j,
    \qquad \forall z<q_j .
\end{aligned}
\label{eq:empirical_quantile_rank_facts}
\end{equation}
The strict inequality in equation \eqref{eq:empirical_quantile_rank_facts} follows from $k_j=\lceil p_j n_j\rceil$, which implies that $k_j-1<p_j n_j$.

\begin{lemma}[Empirical quantile bracketing]
\label{lem:quantile_bracketing}
Let \Cref{ass:within_episode_iid} hold. Given an episode $j$, let $\Delta>0$ satisfy $p_j-\Delta,p_j+\Delta\in(0,1)$. Then
\begin{equation}
\begin{gathered}
\Prob_{n_j}\Bigl\{Q_{r_j}(p_j-\Delta)\le q_j
\le Q_{r_j}(p_j+\Delta)\Bigr\}\\
\ge 1-2\exp(-2n_j\Delta^2).
\end{gathered}
\label{eq:quantile_bracketing_prob}
\end{equation}
\end{lemma}

\begin{proof}
Define the DKW event
\begin{equation*}
    \mathcal A_j(\Delta):=\Bigl\{\sup_{z\in\R}\bigl|F_{n_j,j}(z)-F_{r_j}(z)\bigr|\le \Delta\Bigr\}.
\end{equation*}
Massart's sharp DKW inequality from~\cite{massart1990dkw} gives
\begin{equation}
    \Prob_{n_j}\{\mathcal A_j(\Delta)\}\ge 1-2\exp(-2n_j\Delta^2).
    \label{eq:dkw_probability_bound_quantile_bracketing}
\end{equation}
It remains to prove that $\mathcal A_j(\Delta)$ implies the bracketing event in \eqref{eq:quantile_bracketing_prob}. On $\mathcal A_j(\Delta)$, the first inequality in \eqref{eq:empirical_quantile_rank_facts} gives
\begin{equation}
    F_{r_j}(q_j)\ge F_{n_j,j}(q_j)-\Delta\ge p_j-\Delta.
    \label{eq:lower_bracket_cdf_step}
\end{equation}
By the quantile definition in \eqref{eq:population_quantile_r}, equation \eqref{eq:lower_bracket_cdf_step} implies $Q_{r_j}(p_j-\Delta)\le q_j$.

For the upper bracket, fix any $z<q_j$. On $\mathcal A_j(\Delta)$, the second inequality in \eqref{eq:empirical_quantile_rank_facts} gives
\begin{equation*}
    F_{r_j}(z)\le F_{n_j,j}(z)+\Delta < p_j+\Delta.
\end{equation*}
Therefore, no point $z<q_j$ belongs to the set $\{z\in\R:F_{r_j}(z)\ge p_j+\Delta\}$ that defines $Q_{r_j}(p_j+\Delta)$ in \eqref{eq:population_quantile_r}. Hence the infimum of that set cannot be smaller than $q_j$, and so $q_j\le Q_{r_j}(p_j+\Delta)$. Combining this containment argument with \eqref{eq:dkw_probability_bound_quantile_bracketing} proves \eqref{eq:quantile_bracketing_prob}.
\end{proof}

\begin{lemma}[Local inverse-CDF bound]
\label{lem:local_inverse_cdf_bound}
Let \Cref{ass:quantile_regularity} hold. Fix $r\in\mathcal R_\star$ and $p_1,p_2\in(0,1)$ such that $Q_r(p_1),Q_r(p_2)\in I$. Then
\begin{equation}
    |Q_r(p_2)-Q_r(p_1)|\le \frac{|p_2-p_1|}{m}.
    \label{eq:local_inverse_cdf_bound}
\end{equation}
\end{lemma}

\begin{proof}
It suffices to consider $p_2\ge p_1$, since the other case follows by interchanging the indices. The monotonicity of the function $Q_r(p)$ gives $Q_r(p_1)\le Q_r(p_2)$. Because $F_r(z)$ admits the density function $f_r$ on the open interval $I$ and both endpoint quantiles lie in $I$, the function $F_r$ is continuous at these quantiles. Moreover, if $F_r(Q_r(p_\ell))>p_\ell$ for some $\ell\in\{1,2\}$, then continuity on $I$ would give a point $z<Q_r(p_\ell)$ such that $F_r(z)\ge p_\ell$, contradicting the definition of $Q_r(p_\ell)$. Hence $F_r(Q_r(p_\ell))=p_\ell$ for $\ell\in\{1,2\}$. Therefore, we have
\begin{equation}
\begin{aligned}
    p_2-p_1
    &=F_r(Q_r(p_2))-F_r(Q_r(p_1)) \\
    &=\int_{Q_r(p_1)}^{Q_r(p_2)} f_r(z)\,\dd z
      \ge m\bigl(Q_r(p_2)-Q_r(p_1)\bigr),
\end{aligned}
\label{eq:inverse_cdf_integral_argument}
\end{equation}
where the last inequality uses the lower density bound in \Cref{ass:quantile_regularity}. Rearranging \eqref{eq:inverse_cdf_integral_argument} proves \eqref{eq:local_inverse_cdf_bound}.
\end{proof}

\begin{corollary}[Per-episode quantile-estimation error]
\label{cor:eta_dkw}
Let \Cref{ass:within_episode_iid} and \Cref{ass:quantile_regularity} hold. Fix an episode $j\ge J_0$ with $r_j\in\mathcal R_\star$.
Assume that $p_j-\Delta_j,p_j+\Delta_j\in(0,1)$ where $\Delta_j:=\sqrt{\frac{\ln(2/\delta_j)}{2n_j}}$. Then, it holds that
\begin{equation}
    \Prob_{n_j}\bigl\{|\eta_j|\le \varepsilon_j\bigr\}\ge 1-\delta_j.
    \label{eq:eta_dkw_prob}
\end{equation}
where $\varepsilon_j:=m^{-1}\!\left(\Delta_j+\alpha-\bar\alpha_j\right)$.
\end{corollary}

\begin{proof}
Applying \Cref{lem:quantile_bracketing} with $\Delta=\Delta_j$ gives
\begin{equation}
\Prob_{n_j}\Bigl\{Q_{r_j}(p_j-\Delta_j)\le q_j\le Q_{r_j}(p_j+\Delta_j)\Bigr\}
\ge 1-\delta_j.
\label{eq:bracketing_with_deltaj}
\end{equation}
On the event in \eqref{eq:bracketing_with_deltaj}, the conditions $j\ge J_0$, $r_j\in\mathcal R_\star$, and $p_j-\Delta_j,p_j+\Delta_j\in(0,1)$ imply, by \Cref{ass:quantile_regularity}, that $Q_{r_j}(1-\alpha),Q_{r_j}(p_j-\Delta_j),  Q_{r_j}(p_j), Q_{r_j}(p_j+\Delta_j)
    \in I$. Therefore, \Cref{lem:local_inverse_cdf_bound} and the two-sided bracket in \eqref{eq:bracketing_with_deltaj} give
\begin{align*}
    |\xi_j|
    &=|q_j-Q_{r_j}(p_j)| \notag\\
    &\le \max\Bigl\{
    Q_{r_j}(p_j+\Delta_j)-Q_{r_j}(p_j),\notag\\
    &\hspace{1.8cm}
    Q_{r_j}(p_j)-Q_{r_j}(p_j-\Delta_j)
    \Bigr\}
    \le \frac{\Delta_j}{m},\\
    |b_j|
    &=|Q_{r_j}(p_j)-Q_{r_j}(1-\alpha)|
    \le \frac{\alpha-\bar\alpha_j}{m}.
\end{align*}
where the second inequality uses $p_j-(1-\alpha)=\alpha-\bar\alpha_j$. The decomposition \eqref{eq:eta_decomposition_main}, the triangle inequality, and $\varepsilon_j:=m^{-1}\!\left(\Delta_j+\alpha-\bar\alpha_j\right)$ yield $|\eta_j|\le\varepsilon_j$ on the same event. Combining this implication with \eqref{eq:bracketing_with_deltaj} proves \eqref{eq:eta_dkw_prob}.
\end{proof}

\subsection{Proof of \texorpdfstring{\Cref{thm:convergence}}{Theorem~\ref{thm:convergence}}}

\begin{proof}
Recall from \eqref{eq:eta_decomposition_main} that $q_j=Q_{r_j}(1-\alpha)+\eta_j$. Recall also that the fixed point satisfies $r_\star=Q_{r_\star}(1-\alpha)$.

\emph{One-step bound.} If $q_j\ge r_j$, the update rule in~\eqref{eq:explicit_update} gives $r_{j+1}=(q_j-\kappa r_j)/(1-\kappa)$. Substituting \eqref{eq:eta_decomposition_main} and subtracting $r_\star=Q_{r_\star}(1-\alpha)$ gives
\begin{align*}
  r_{j+1}-r_\star
  &= \frac{Q_{r_j}(1-\alpha)-Q_{r_\star}(1-\alpha)}{1-\kappa}
     -\frac{\kappa(r_j-r_\star)}{1-\kappa}\\
  &\quad + \frac{\eta_j}{1-\kappa}.
\end{align*}
By equation \eqref{eq:Q_Lipschitz}, we have $|Q_{r_j}(1-\alpha)-Q_{r_\star}(1-\alpha)|\le\kappa e_j$, where $e_j:=|r_j-r_\star|$. Hence, we have
\begin{equation*}
\begin{aligned}
&\left|\bigl[Q_{r_j}(1-\alpha)-Q_{r_\star}(1-\alpha)\bigr]
-\kappa(r_j-r_\star)\right|\\
&\hspace{1.8cm}\le 2\kappa e_j.
\end{aligned}
\end{equation*}
From here, it follows that $e_{j+1}\le 2\kappa e_j/(1-\kappa)+|\eta_j|/(1-\kappa)=\lambda_\kappa e_j+B_\kappa|\eta_j|$.

If $q_j<r_j$, the update rule in~\eqref{eq:explicit_update} gives $r_{j+1}=(q_j+\kappa r_j)/(1+\kappa)$. The same decomposition gives $e_{j+1}\le 2\kappa e_j/(1+\kappa)+|\eta_j|/(1+\kappa)\le\lambda_\kappa e_j+B_\kappa|\eta_j|$. Both cases together yield \eqref{eq:contraction}.

\emph{Unrolling.} Iterating \eqref{eq:contraction} across episodes gives
\begin{equation*}
    e_{j+1}\le\lambda_\kappa^{j+1}e_0+B_\kappa\sum_{m=0}^{j}\lambda_\kappa^{j-m}|\eta_m|.
\end{equation*}
If $|\eta_m|\le C$ for all $m$ and $\kappa<1/3$, then $\lambda_\kappa<1$ and the geometric-series bound gives
\begin{equation*}
    \limsup_{j\to\infty}e_j
    \le \frac{B_\kappa C}{1-\lambda_\kappa}
    =\frac{C}{1-3\kappa}.
\end{equation*}
\end{proof}

\subsection{Proof of \texorpdfstring{\Cref{cor:hp_tracking_convergence}}{Corollary~\ref{cor:hp_tracking_convergence}}}

\begin{proof}
Define the event $E_j:=\{|\eta_j|\le\varepsilon_j\}$. By \Cref{cor:eta_dkw}, we have that $\Prob_{n_j}(E_j)\ge1-\delta_j$ for each episode $j=J_0,\ldots,J$ under the  conditions stated in \Cref{cor:hp_tracking_convergence}. Applying a union bounding argument under the probability measure $\Prob_{0:J}\{\cdot\}$ directly yields
\begin{equation*}
\Prob_{0:J}\Bigl\{\bigcap_{j=J_0}^{J}E_j\Bigr\}
\ge 1-\sum_{j=J_0}^{J}\delta_j.
\end{equation*}
On the event $\bigcap_{j=J_0}^{J}E_j$, \Cref{thm:convergence} gives \eqref{eq:contraction} with $|\eta_j|$ replaced by $\varepsilon_j$ for all $j=J_0,\ldots,J$. Together, this proves  $\Prob_{0:J}\{\mathcal H_J\}\ge 1-\sum_{j=J_0}^{J}\delta_j$ where $\mathcal H_J:=\bigcap_{j=J_0}^{J}
\left\{e_{j+1}\le
\lambda_\kappa e_j+B_\kappa\varepsilon_j\right\}$. 

In the asymptotic case, we instead obtain
\begin{equation*}
\Prob_{0:\infty}\Bigl\{\bigcap_{j=J_0}^{\infty}E_j\Bigr\}
\ge 1-\sum_{j=J_0}^{\infty}\delta_j .
\end{equation*}
On the event $\bigcap_{j=J_0}^{\infty}E_j$, the recursion in \Cref{thm:convergence}, started from episode $J_0$, gives
\begin{equation*}
\begin{aligned}
    e_{j+1}
    &\le \lambda_\kappa^{j+1-J_0}e_{J_0}
    +B_\kappa\sum_{m=J_0}^{j}\lambda_\kappa^{j-m}\varepsilon_m .
\end{aligned}
\end{equation*}
If $\sup_{j\ge J_0}\varepsilon_j\le C$, the geometric-series argument in \Cref{thm:convergence} gives $\limsup_{j\to\infty}e_j\le C/(1-3\kappa)$ so that
\begin{equation*}
\begin{gathered}
\Prob_{0:\infty}\!\left\{
\limsup_{j\to\infty}|r_j-r_\star|
\le \frac{C}{1-3\kappa}
\right\}
\ge 1-\sum_{j=J_0}^{\infty}\delta_j .
\end{gathered}
\end{equation*}
If, in addition, $\varepsilon_j\to0$, then the convolution term $B_\kappa\sum_{m=J_0}^{j}\lambda_\kappa^{j-m}\varepsilon_m $ converges to zero. Indeed, for any $\epsilon>0$, choose $M\ge J_0$ such that $\varepsilon_m\le\epsilon$ for all $m\ge M$, split the sum at $M$, let $j\to\infty$ to eliminate the finite initial sum, and then let $\epsilon\downarrow0$. Since $\lambda_\kappa^{j+1-J_0}e_{J_0}\to0$, it follows that $e_j\to0$, equivalently $r_j\to r_\star$, on the same event.
\end{proof}                
\section{Extra Plots and Details for Case Studies}
\label{app:navigation_details}

This appendix provides detailed information about all three case studies. 

\begin{arxivonly}
\subsection{Inverted Pendulum CLF}
\begin{table}[h]
    \centering
    \caption{Final-episode ($j=9$) results for the inverted-pendulum.}
    \begin{tabular}{@{}lcccc@{}}
        \toprule
         & Robust & Naive & Cal-once & Non-rob. \\
        \midrule
        Final $r_9$ & 0.5282 & 0.5272 & 0.5271 & 0.0 \\
        Score cov. & 1.0 & 0.99 & 0.98 & 0.0 \\
        Stability & 1.0 & 1.0 & 1.0 & 0.36 \\
        \bottomrule
    \end{tabular}
    \label{tab:ip-final}
\end{table}
\begin{table}[h]
\centering
\caption{Experiment parameters for the inverted pendulum.}
\label{tab:pendulum_params}
\begin{tabular}{@{}llcl@{}}
\toprule
\textbf{Symbol} & \textbf{Description} & \textbf{Value} & \textbf{Unit} \\
\midrule
\multicolumn{4}{@{}l}{\emph{Pendulum model}} \\
$g$ & Gravitational acceleration & $9.81$ & m/s$^2$ \\
$m$ & True pendulum mass & $1.0$ & kg \\
$\ell$ & True pendulum length & $1.0$ & m \\
$b$ & True damping coefficient & $0.01$ & -- \\
$u_{\min},\,u_{\max}$ & Input bounds & $-7.0,\ 7.0$ & -- \\
\midrule
\multicolumn{4}{@{}l}{\emph{CBF-QP \eqref{eq:cbf_qp}}} \\
$K_{\mathrm{fb}}$ & Initial feedback gain specification & $(6.0,\,1.0)$ & -- \\
$c_3$ & CLF decay rate & $0.5$ & s$^{-1}$ \\
\midrule
\multicolumn{4}{@{}l}{\emph{Simulation}} \\
$T$ / $\Delta t$ & Horizon / time step & $5.0$ / $0.02$ & s \\
\midrule
\multicolumn{4}{@{}l}{\emph{Conformal prediction}} \\
$\alpha$ / $\delta$ & Miscov.\ / outer confidence & $0.10$ / $0.10$ & -- \\
$\kappa$ & Robust contraction constant & $0.8$ & -- \\
$r_0$ & Initial robustness margin & $2.0$ & -- \\
$n_j$ / $N_{\mathrm{eval}}$ & Cal.\ / eval.\ trajectories per episode & $200$ / $100$ & -- \\
$J$ & Number of episodes & $10$ & -- \\
\bottomrule
\end{tabular}
\end{table}
We use the same general setup and implementation as discussed in \cite{hsu2025}.
We define the state as $x = [\theta \ \dot \theta]^\top \in \mathbb{R}^2$ where $\theta$ denotes the angular position of the pendulum, and apply torque control $u \in \mathbb{R}$ to the true dynamics $f(x,u) = \begin{bmatrix}
    \dot \theta \\ -b \dot \theta / I + mgL \sin \theta / (2 I)
\end{bmatrix} + \begin{bmatrix}
    0 \\ - 1/I
\end{bmatrix} u$, where $m=1$, $l=1$, $b=0.01$, and $I = mL^2 / 3$.
The nominal dynamics are given as $\hat{f}(x,u)= M_1 \phi(x) + M_2 \phi(x) u$, where $M_1, M_2 \in \mathbb{R}^{2 \times 10}$ are learned weight matrices and $\phi(x) = [1 \ \theta \ \dot \theta \ \theta^2 \ \theta \dot \theta \ \dot \theta^2 \ \theta^3 \cdots \dot \theta^3]$ is a feature map.
\Cref{tab:ip-final} provides exact values for the final results at episode $j=9$, while \Cref{tab:pendulum_params} summarizes all simulation parameters.
For further plots, see Figures~\ref{fig:ip-results-large}--\ref{fig:ip-trajectories-cd-large}.

\end{arxivonly}

\begin{extendedonly}
\subsection{Single-Obstacle Gaussian-Vortex Parameters}
\label{app:gaussian_vortex_details}

\begin{table}[h]
\centering
\caption{Parameters for the single-obstacle Gaussian-vortex navigation.}
\label{tab:gaussian_vortex_params}
\begin{tabular}{@{}llcl@{}}
\toprule
\textbf{Symbol} & \textbf{Description} & \textbf{Value} & \textbf{Unit} \\
\midrule
\multicolumn{4}{@{}l}{\emph{Geometry}} \\
$R_{\mathrm{obs}}$/$R_{\mathrm{safe}}$ & Obstacle / safety radius & $0.5$ / $0.6$ & m \\
$B$ & Goal position & $(1.5,\,0.52)$ & m \\
$\XX$ & IC box $\cap\{h\ge 0.05\}$ & $[-3,-0.7]\!\times\![-1.5,0.5]$ & m \\
\midrule
\multicolumn{4}{@{}l}{\emph{Gaussian-vortex residual $\varepsilon(x,u)=\alpha_g e^{-\|x\|^2/(2\sigma_g^2)}(R_{\theta_g}-I)u+d_g$}} \\
$\theta_g$ & Rotation angle & $70^\circ$ & -- \\
$\sigma_g$ & Gaussian field scale & $1.1$ & m \\
$\alpha_g$ & Field amplitude & $1.0$ & -- \\
$d_g$ & Ambient drift & $(0,\,-0.005)$ & m/s \\
\midrule
\multicolumn{4}{@{}l}{\emph{CBF-QP}} \\
$\gamma$ & CBF decay rate & $11.0$ & s$^{-1}$ \\
$k_{\mathrm{trk}}$ & Tracking gain & $1.0$ & s$^{-1}$ \\
\midrule
\multicolumn{4}{@{}l}{\emph{Simulation}} \\
$T$ / $\Delta t$ & Horizon / time step & $5.0$ / $0.01$ & s \\
\midrule
\multicolumn{4}{@{}l}{\emph{Conformal prediction}} \\
$\alpha$ / $\delta$ & Miscov.\ / outer confidence & $0.1$ / $0.05$ & -- \\
$\kappa$ & Robust contraction constant & $0.2$ & -- \\
$n_j$ / $N_{\mathrm{eval}}$ & Cal.\ / eval.\ trajectories per episode & $500$ / $1{,}000$ & -- \\
$J$ & Number of episodes & $20$ & -- \\
\bottomrule
\end{tabular}
\end{table}

\paragraph{Residual model.}
The Gaussian-vortex residual $\varepsilon(x,u)=\alpha_g e^{-\|x\|^2/(2\sigma_g^2)}(R_{\theta_g}-I)u+d_g$ differs from the barrier-localized vortex $\varepsilon(x,u)=\sigma(x)(R_\theta-I)u+d$ with $\sigma(x):=e^{-\max\{h(x),0\}/\ell}$ in two respects.
First, the localization function is a Gaussian $\alpha_g\exp(-\|x\|^2/(2\sigma_g^2))$ centered at the obstacle rather than $\exp(-\max\{h(x),0\}/\ell)$, making it smooth everywhere and independent of the barrier function~$h$.
Second, the rotation angle is $\theta_g=70^\circ$ rather than $160^\circ$, which results in a milder reversal of the radial component: $\cos(70^\circ)\approx 0.34$ versus $\cos(160^\circ)\approx -0.94$. This makes the Gaussian-vortex setting more benign---the distribution shift is present but moderate---providing a convergent regime for the robust method recursion.

\paragraph{One-step empirical results.}
\begin{table}[h]
\centering
\caption{One-step empirical results for the Gaussian-vortex case ($N=1{,}000$, $\alpha=0.1$).}
\label{tab:gaussian_onestep}
\begin{tabular}{@{}lcc@{}}
\toprule
\textbf{Metric} & $\pi_0$ ($r=0$) & $\pi_1 = \pi_r^{\mathrm{cbf}}$ ($r=\hat q$) \\
\midrule
Score coverage $\hat p_{\mathrm{score}}$ & $0.904$ & $0.885$ \\
Safety probability $\hat p_{\mathrm{safe}}$ & $0.977$ & $0.885$ \\
$\hat q_{n,\alpha}$ & \multicolumn{2}{c}{$2.056$} \\
\bottomrule
\end{tabular}
\end{table}
\end{extendedonly}

\subsection{Multi-Obstacle Maze Parameters}
\label{app:maze_details}

\begin{extendedonly}
\begin{table}[h]
\centering
\caption{Experiment parameters for the multi-obstacle maze (\Cref{sec:maze}).}
\label{tab:maze_params}
\begin{tabular}{@{}llcl@{}}
\toprule
\textbf{Symbol} & \textbf{Description} & \textbf{Value} & \textbf{Unit} \\
\midrule
\multicolumn{4}{@{}l}{\emph{Geometry}} \\
$K$ & Number of obstacles & $17$ & -- \\
$\eta$ & Safety margin factor & $0.25$ ($R_{s,i}=1.25\,R_i$) & -- \\
$B_{\mathrm{m}}$ & Goal position & $(10,\,0)$ & m \\
${\XX}_{\mathrm{m}}$ & IC box $\cap\{\min_i h_i \ge 0.05\}$ & $[-5,-0.5]\!\times\![-2.59,2.59]$ & m \\
\midrule
\multicolumn{4}{@{}l}{\emph{Disturbance $\varepsilon(x,u)=\sum_i\sigma_i(x)(R_{\theta_i}-I_2)u+d_{\mathrm{m}}$, $R_\theta\in\mathrm{SO}(2)$}} \\
$\theta_i$ & Per-obstacle rotation angle & $10^\circ$--$32^\circ$ & -- \\
$\ell_i$ & Per-obstacle length-scale & $0.15$--$0.50$ & m \\
$d_{\mathrm{m}}$ & Ambient drift & $(0.001,\,-0.002)$ & m/s \\
\midrule
\multicolumn{4}{@{}l}{\emph{CBF-QP~\eqref{eq:cbf_qp}}} \\
$\gamma$ & CBF decay rate & $10.0$ & s$^{-1}$ \\
$k_{\mathrm{trk}}$ & Tracking gain & $0.6$ & s$^{-1}$ \\
\midrule
\multicolumn{4}{@{}l}{\emph{Simulation}} \\
$T$ / $\Delta t$ & Horizon / time step & $12.0$ / $0.01$ & s \\
\midrule
\multicolumn{4}{@{}l}{\emph{Conformal prediction}} \\
$\alpha$ / $\delta$ & Miscov.\ / outer confidence & $0.1$ / $0.05$ & -- \\
$\kappa$ & Robust contraction constant & $0.3$ & -- \\
$n_j$ / $N_{\mathrm{eval}}$ & Cal.\ / eval.\ trajectories per episode & $200$ / $500$ & -- \\
$J$ & Number of episodes & $20$ & -- \\
\bottomrule
\end{tabular}
\end{table}
\end{extendedonly}
\begin{compactonly}
\begin{table}[h]
\centering
\caption{Experiment parameters for the multi-obstacle maze (\Cref{sec:maze}).}
\label{tab:maze_params}
\footnotesize
\begin{tabular}{@{}llc@{}}
\toprule
\textbf{Symbol} & \textbf{Description} & \textbf{Value} \\
\midrule
\multicolumn{3}{@{}l}{\emph{Geometry}} \\
$K$ & Number of obstacles & $17$ \\
$\eta$ & Safety margin factor & $0.25$ ($R_{s,i}\!=\!1.25\,R_i$) \\
$B_{\mathrm{m}}$ & Goal position & $(10,\,0)$ \\
${\XX}_{\mathrm{m}}$ & IC box $\cap\{\min_i h_i \!\ge\! 0.05\}$ & $[-5,\!-\!0.5]\!\times\![-2.59,2.59]$ \\
\midrule
\multicolumn{3}{@{}l}{\emph{Disturbance $\varepsilon(x,u)=\sum_i\sigma_i(x)(R_{\theta_i}-I_2)u+d_{\mathrm{m}}$, $R_\theta\in\mathrm{SO}(2)$}} \\
$\theta_i$ & Per-obstacle rotation angle & $10^\circ$--$32^\circ$ \\
$\ell_i$ & Per-obstacle length-scale & $0.15$--$0.50$ \\
$d_{\mathrm{m}}$ & Ambient drift & $(0.001,\,-0.002)$ \\
\midrule
\multicolumn{3}{@{}l}{\emph{CBF-QP~\eqref{eq:cbf_qp}}} \\
$\gamma$ & CBF decay rate & $10.0$ \\
$k_{\mathrm{trk}}$ & Tracking gain & $0.6$ \\
\midrule
\multicolumn{3}{@{}l}{\emph{Simulation}} \\
$T$ / $\Delta t$ & Horizon / time step & $12.0$ / $0.01$\,s \\
\midrule
\multicolumn{3}{@{}l}{\emph{Conformal prediction}} \\
$\alpha$ / $\delta$ & Miscov.\ / outer confidence & $0.1$ / $0.05$ \\
$\kappa$ & Robust contraction constant & $0.3$ \\
$n_j$ / $N_{\mathrm{eval}}$ & Cal.\ / eval.\ traj.\ per episode & $200$ / $500$ \\
$J$ & Number of episodes & $20$ \\
\bottomrule
\end{tabular}
\end{table}
\end{compactonly}

\paragraph{Obstacle layout.}
The $17$ obstacles are arranged to form a maze with three rows and staggered interior gaps. Table~\ref{tab:maze_obstacles} lists the center coordinates, physical radii, rotation angles, and length-scales for each obstacle. The boundary rows at $y=\pm 2$ use small rotation angles ($\theta_i=10^\circ$) and short length-scales ($\ell_i=0.15$), modeling a weak but spatially concentrated vortex near the corridor walls. The interior obstacles have larger rotation angles ($18^\circ$--$32^\circ$) and broader length-scales ($0.30$--$0.50$), creating stronger and more extended disturbance fields that dominate the navigation challenge.

\begin{table}[h]
\centering
\caption{Obstacle specifications for the multi-obstacle maze (\Cref{sec:maze}).}
\label{tab:maze_obstacles}
\footnotesize
\begin{tabular}{@{}cccccl@{}}
\toprule
\textbf{Index} & $c_i$ & $R_i$ (m) & $\theta_i$ (deg) & $\ell_i$ & \textbf{Row} \\
\midrule
 1 & $(1.0,\,2.0)$    & 0.35 & 10 & 0.15 & upper \\
 2 & $(3.5,\,2.0)$    & 0.35 & 10 & 0.15 & upper \\
 3 & $(6.0,\,2.0)$    & 0.35 & 10 & 0.15 & upper \\
 4 & $(8.5,\,2.0)$    & 0.35 & 10 & 0.15 & upper \\
 5 & $(1.0,\,-2.0)$   & 0.35 & 10 & 0.15 & lower \\
 6 & $(2.5,\,-2.0)$   & 0.35 & 10 & 0.15 & lower \\
 7 & $(5.0,\,-2.0)$   & 0.35 & 10 & 0.15 & lower \\
 8 & $(7.5,\,-2.0)$   & 0.35 & 10 & 0.15 & lower \\
 9 & $(1.5,\,-0.2)$   & 0.42 & 30 & 0.45 & interior \\
10 & $(3.0,\,0.5)$    & 0.40 & 28 & 0.50 & interior \\
11 & $(4.5,\,-0.4)$   & 0.45 & 32 & 0.45 & interior \\
12 & $(6.0,\,0.4)$    & 0.38 & 28 & 0.50 & interior \\
13 & $(7.5,\,-0.2)$   & 0.42 & 30 & 0.45 & interior \\
14 & $(2.83,\,-0.87)$ & 0.25 & 20 & 0.35 & interior \\
15 & $(4.2,\,1.2)$    & 0.25 & 18 & 0.35 & interior \\
16 & $(5.2,\,1.25)$   & 0.22 & 15 & 0.30 & interior \\
17 & $(6.9,\,-1.2)$   & 0.24 & 22 & 0.35 & interior \\
\bottomrule
\end{tabular}
\end{table}

\paragraph{QP solver.}
With $17$ rCBF constraints in two dimensions, the CBF-QP~\eqref{eq:cbf_qp} has $\binom{17}{2}+17+1 = 154$ candidate solutions (one unconstrained, $17$ single-constraint projections, and $136$ pairwise intersections). For each candidate, feasibility is checked by evaluating $17$ inner products. In practice, the unconstrained candidate $u_{\mathrm{nom}}$ is feasible in the vast majority of time steps (the agent is far from all obstacles), and the solver exits in $O(K)$ time. Near obstacles, at most a handful of constraints become active. The solver returns the feasible candidate closest to $u_{\mathrm{nom}}$ in Euclidean norm, guaranteeing global optimality.

\paragraph{Episodic convergence.}
The robust margin converges from its initial value $r_0 \approx 2.75$ (obtained from calibration at $r=0$) to a stable value $r\approx 2.38$ within two episodes. This convergence is consistent with the tracking guarantee of Theorem~\ref{thm:convergence}: since $\kappa = 0.3 < 1/3$, the contraction rate $\lambda_\kappa = 2\kappa/(1-\kappa) = 6/7 < 1$ ensures geometric convergence up to the quantile-estimation error $\eta_j$.

\paragraph{Final-episode results.}
\begin{extendedonly}
\begin{table}[h]
\centering
\caption{Final-episode ($j=19$) empirical results for the multi-obstacle maze.}
\label{tab:maze_final}
\begin{tabular}{@{}lcccc@{}}
\toprule
\textbf{Metric} & M1 (original) & M2 (cal-once) & M3 (Naive) & M4 (Robust) \\
\midrule
Final margin $r_{19}$ & $0.000$ & $2.246$ & $2.336$ & $2.380$ \\
Score coverage & -- & $0.668$ & $0.872$ & $0.998$ \\
Safety rate & $0.262$ & $1.000$ & $1.000$ & $1.000$ \\
\bottomrule
\end{tabular}
\end{table}
\end{extendedonly}
\begin{compactonly}
\begin{table}[h]
\centering
\caption{Final-episode ($j\!=\!19$) results for the multi-obstacle maze.}
\label{tab:maze_final}
\footnotesize
\begin{tabular}{@{}lcccc@{}}
\toprule
 & Non-rob. & Cal-once & Naive & Robust \\
\midrule
Final $r_{19}$ & $0.000$ & $2.246$ & $2.336$ & $2.380$ \\
Score cov. & -- & $0.668$ & $0.872$ & $0.998$ \\
Safety & $0.262$ & $1.000$ & $1.000$ & $1.000$ \\
\bottomrule
\end{tabular}
\end{table}
\end{compactonly}
Table~\ref{tab:maze_final} reports the final-episode metrics. The non-robust baseline ($r=0$) suffers a $73.8\%$ collision rate. The calibrate-once baseline achieves full safety but its fixed margin $r_{\mathrm{cal}}\approx 2.25$ leads to score coverage of only $66.8\%$, well below the required $1-\alpha=0.9$. The naive baseline reaches $87.2\%$ score coverage---closer but still below the target. Robust achieves the highest score coverage ($99.8\%$) while maintaining full safety, confirming that the iterative margin update successfully balances robustness and coverage in geometrically complex environments. For additional plots, see Figures~\ref{fig:maze-episodic-large}--\ref{fig:maze-traj-srcr-large} in \Cref{app:tables_figures}.

\begin{fullonly}
\subsection{Quadcopter Obstacle-Avoidance}
\label{sec:quadc}
We examine a quadcopter navigation task built on the open-source QuadSwarm simulator \cite{quadswarm}.
The agent is provided with a spawning region (a ball of radius 1 meter) and a goal point.  
Each agent is provided with a pre-trained RL policy which directs it to the goal point but this policy only observes the agent's current state and cannot react to the environment at large: any interactive behavior is determined entirely by a CBF.
We model the nominal dynamics as defined in the QuadSwarm paper; however, we ignore any nose or damping effects found in the true dynamics.  Additionally, the nominal dynamics don't account for environment interactions such as collisions or downwash effects between nearby quads.
These discrepancies result in context-dependent mismatches between the nominal and true dynamics.
As the simulator is implemented in discrete time, in practice we estimate the trajectory-level error as
$\varepsilon(z(t), u(t)) \approx \| (\hat x, \hat v)(t+1) - (x, v)(t+1) \| / \delta t$
where $\delta t$ denotes the length of each timestep and $z$, $x$, and $v$ denotes the quadcopter state, position, and velocity vectors.
On these nominal dynamics, we implement a 4th order ECBF-QP enforcing a radius norm distance condition.
This is done to give the CBF control of each motor's individual thrust, which only appears on the fourth derivative of the $h$-function.
For every episode $j$, $n_j =300$ trajectories are collected in estimating $q_j$ for $\alpha = 0.1$, $\delta = 0.05$ and 200 trajectories are collected with each $r_j$ for evaluation purposes. The robust baseline uses $\kappa = 0.6$ in the explicit update~\eqref{eq:explicit_update}.

\begin{table}[h]
\centering
\caption{Experiment parameters for the single-quad obstacle-avoidance setting.}
\label{tab:quad_obs_params}
\begin{tabular}{@{}llcl@{}}
\toprule
\textbf{Symbol} & \textbf{Description} & \textbf{Value} & \textbf{Unit} \\
\midrule
\multicolumn{4}{@{}l}{\emph{Quadrotor model}} \\
$g$ & Gravitational acceleration & $9.81$ & m/s$^2$ \\
$a_{\mathrm{arm}}$ & Quadrotor arm length & $0.04596$ & m \\
$\mathrm{T2W}$ & Thrust-to-weight ratio & $1.9$ & -- \\
$c_{\tau}$ & Torque-to-thrust ratio & $0.006$ & -- \\
$\tau_{\uparrow},\,\tau_{\downarrow}$ & Motor spin-up/down time & $0.15,\ 0.15$ & s \\
$\sigma_{\mathrm{thr}}$ & Thrust-noise ratio & $0.05$ & -- \\
$d_v$ & Linear velocity damping & $10^{-4}$ & -- \\
$d_{\omega}$ & Quadr. angular-rate damping & $10^{-4}$ & -- \\
\midrule
\multicolumn{4}{@{}l}{\emph{Obstacle environment}} \\
$\rho_{\mathrm{obs}}$ & Obstacle radius & $2.0$ & m \\
$R_{\mathrm{obs}}$ & Extra CBF radius & $0.1$ & m \\
$R_{\mathrm{spawn}}$ & Spawn-ball radius & $1.0$ & m \\
\midrule
\multicolumn{4}{@{}l}{\emph{CBF-QP}} \\
$\lambda_0$ & ECBF coefficient for $h$ & 256 & -- \\
$\lambda_1$ & ECBF coefficient for $\dot h$ & 256 & -- \\
$\lambda_2$ & ECBF coefficient for $\ddot h$ & 96 & -- \\
$\lambda_3$ & ECBF coefficient for $\dddot h$ & 16 & -- \\
$\lambda_4$ & ECBF coefficient for $\ddddot h$ & 1 & -- \\
\midrule
\multicolumn{4}{@{}l}{\emph{Simulation}} \\
$f_{\mathrm{sim}}$ / $\Delta t_{\mathrm{sim}}$ & Inner sim. freq. / time step & $200$ / $0.005$ & Hz / s \\
$f_{\mathrm{ctrl}}$ / $\Delta t$ & Control freq. / time step & $100$ / $0.01$ & Hz / s \\
$H$ / $T$ & Horizon / rollout duration & $700$ / $7.0$ & steps / s \\
\midrule
\multicolumn{4}{@{}l}{\emph{Conformal prediction}} \\
$\alpha$ / $\delta$ & Miscov.\ / outer confidence & $0.10$ / $0.05$ & -- \\
$\kappa$ & Robust contraction const. & $0.6$ & -- \\
$r_0$ & Initial robustness margin & $2.0$ & -- \\
$n_j$ / $N_{\mathrm{eval}}$ & Cal.\ / eval.\ trajectories & $300$ / $200$ & -- \\
$J$ & Number of episodes & $10$ & -- \\
\bottomrule
\end{tabular}
\end{table}

\begin{table}[h]
    \centering
    \caption{Environment for the quadcopter obstacle-avoidance setting.}
    \begin{tabular}{@{}lcccc@{}}
        \toprule
        Name & x & y & z & Radius \\
        \midrule
        Spawn center & -3.5 & -4.1 & 0.0 & 1.0 \\
        Obstacle 1 & -1.8 & -1.0 & -- & 2.0 \\
        Obstacle 2 & 3.9 & 0.0 & -- & 2.0 \\
        Obstacle 3 & 2.5 & -4.5 & -- & 2.0 \\
        Goal & 4.0 & 3.5 & 0.0 & 0.0 \\
        \bottomrule
    \end{tabular}
    \label{tab:obs-points}
\end{table}

We consider an environment with one ego quadcopter and a hand-designed set of fixed obstacles, each being a column of a 2m radius.  
The obstacles were placed to create a challenging track for the quad to navigate, forcing the quadcopter to traverse a corridor surrounded by these columns in order to reach the goal point.
The ECBF was designed to enforce an effective margin of 100cm from each obstacle, with $h(x) = \min_{o_i \in \mathcal{O}} \| P(x - o_i) \| - 2.1$, where $P$ represents a map from $(x,y,z) \mapsto (x,y)$ space and $\mathcal{O}$ represents the set of obstacle centers.

We provide plots depicting results comparing baselines in \Cref{fig:obs-results-large} and trajectory rollouts in \Cref{fig:obs-trajs-large}.
In this setting, the mismatch between the nominal and true dynamics is greater during rollouts of larger values of $r$.
Higher values of $r$ cause the quad to engage in stopping behavior as it attempts to make progress towards the goal but reaches areas that are not permitted given the robustification coefficient: this stopping behavior exacerbates the mismatch between the nominal and true dynamics, as velocity and acceleration damping comes into effect while the quad attempts to sharply change velocity and acceleration.
As a result, the calibrate-once trajectories fail to exploit the distribution shift resultant from changing the robustification of the CBF, exhibiting significantly worse performance in cumulative reward across episodes (\Cref{fig:obs-results-large}(b)).  
Indeed, while the initial value of $r=2$ results in the quads attempting to circumvent the maze entirely, the calibrate-once value of $r=0.671$ results in degenerate behavior where the quads attempt to make progress but realize that $r$ is still too large to allow proper pathfinding and get stuck (\Cref{fig:obs-trajs-large}(b)).
In contrast, our algorithm succeeds in safely converging on a value of $r$ which allows the quad to consistently brave the maze while maintaining the required margin from all obstacles (\Cref{fig:obs-trajs-large}(c)).

\end{fullonly}
\arxonly{\onecolumn}
\section{Collected Tables and Figures}
\label{app:tables_figures}

This section reproduces the main result figures in enlarged format for improved readability.

\subsection{Example 1}
\begin{figure}[H]
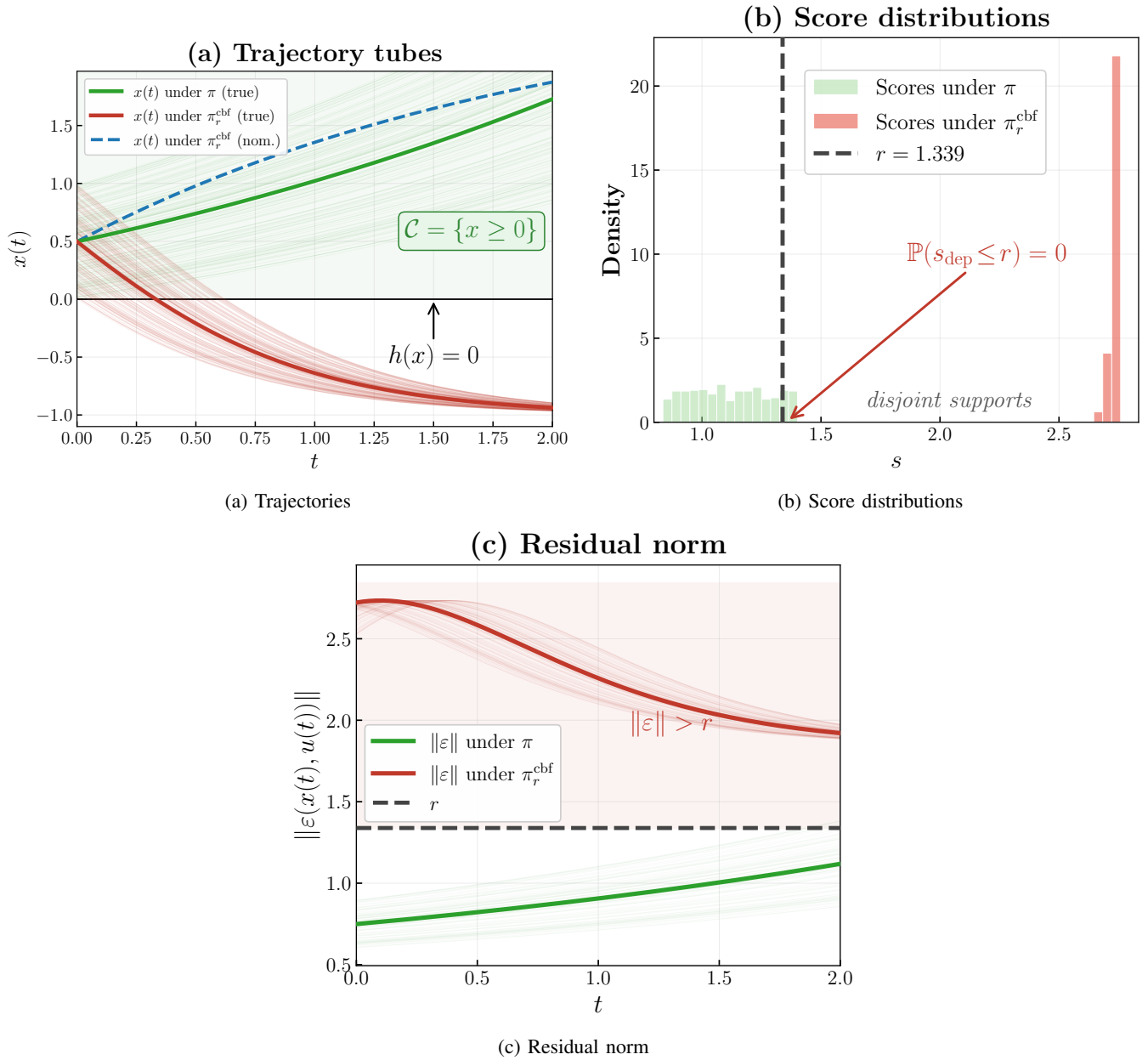

    \centering
    \begin{subfigure}[t]{0.50\textwidth}
        \includegraphics[width=\linewidth]{figures/final/ex1_trajectories.pdf}
        \caption{Trajectories}
    \end{subfigure}\hfill 
    \begin{subfigure}[t]{0.48\textwidth}
        \includegraphics[width=\linewidth]{figures/final/ex1_scores.pdf}
        \caption{Score distributions}
    \end{subfigure}\\[0.5em] 
    \begin{subfigure}[t]{0.50\textwidth}
        \includegraphics[width=\linewidth]{figures/final/ex1_residual.pdf}
        \caption{Residual norm}
    \end{subfigure}
    \caption{Example~1 (enlarged). $u_0\!=\!0.3$, $T\!=\!2$, $\alpha\!=\!0.1$, $\gamma\!=\!0.5$.}
    \label{fig:prop1-large}
\end{figure}

\subsection{Inverted Pendulum}

\begin{figure}[H]
    \centering
    \captionsetup{font=small,skip=2pt}
    \begin{subfigure}[t]{0.49\linewidth}
        \includegraphics[width=\linewidth]{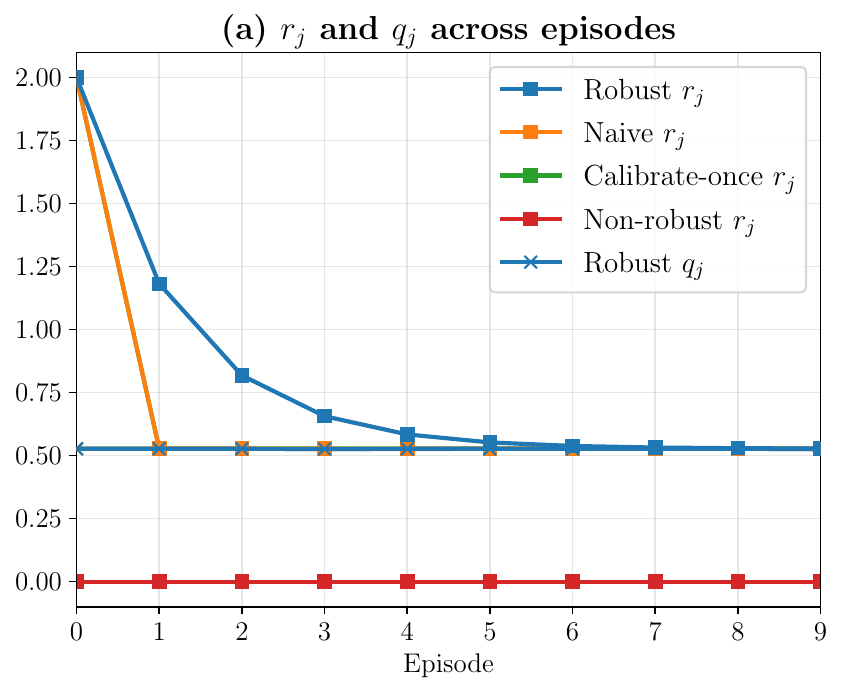}
    \end{subfigure}\hfill
    \begin{subfigure}[t]{0.49\linewidth}
        \includegraphics[width=\linewidth]{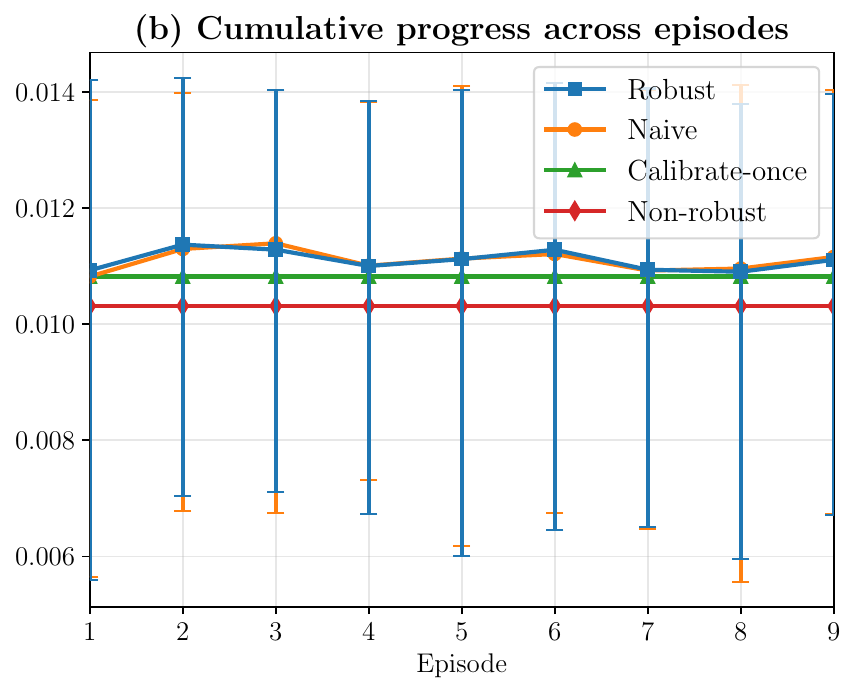}
    \end{subfigure}\hfill \\
    \begin{subfigure}[t]{0.49\linewidth}
        \includegraphics[width=\linewidth]{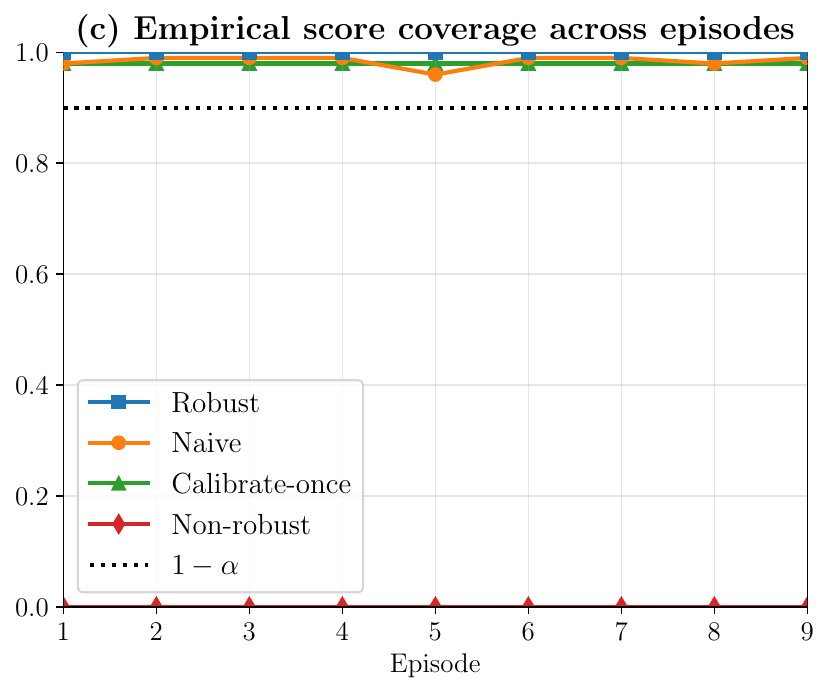}
    \end{subfigure}\hfill
    \begin{subfigure}[t]{0.49\linewidth}
        \includegraphics[width=\linewidth]{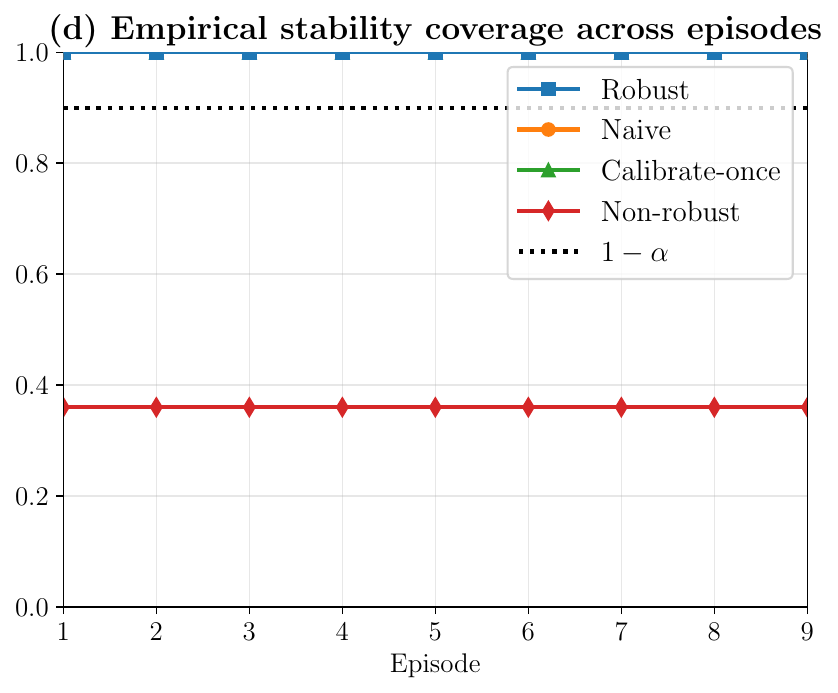}
    \end{subfigure}
    \caption{Inverted pendulum. (a)~$r_j$, $q_j$ across episodes. (b)~Cumulative progress towards $(0,0)$. (c)~Score coverage $s_j^{(i)}\!\le\! r_j$. (d)~Stability coverage per Theorem~\ref{thm:det_ct_clf}.}
    \label{fig:ip-results-large}
    \subonly{\vspace{-.2cm}}
\end{figure}

\begin{figure}[H]
    \centering
    \captionsetup{font=small,skip=2pt}
    \begin{subfigure}[t]{0.49\linewidth}
        \includegraphics[width=\linewidth]{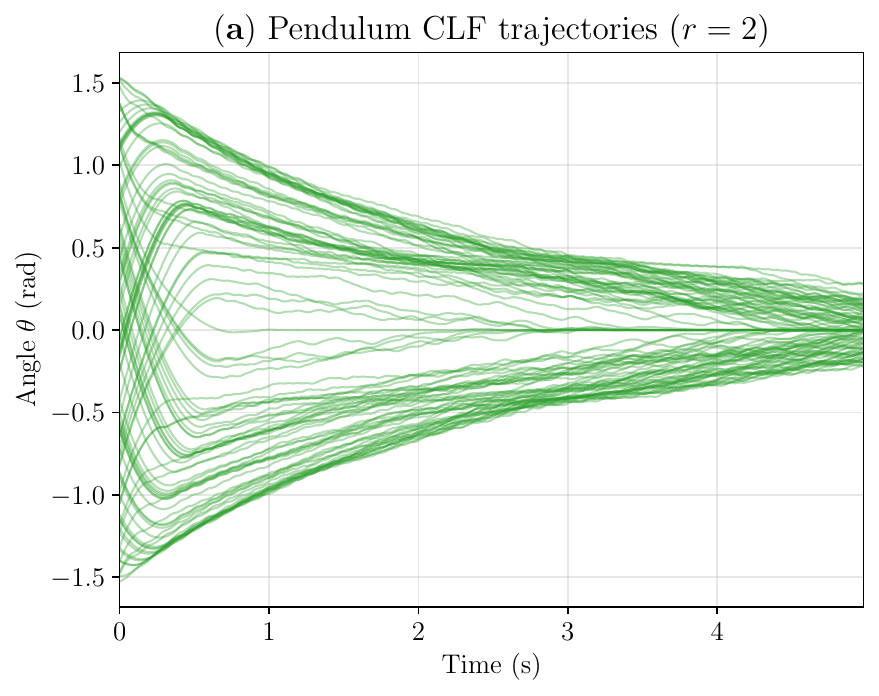}
    \end{subfigure}\hfill
    \begin{subfigure}[t]{0.49\linewidth}
        \includegraphics[width=\linewidth]{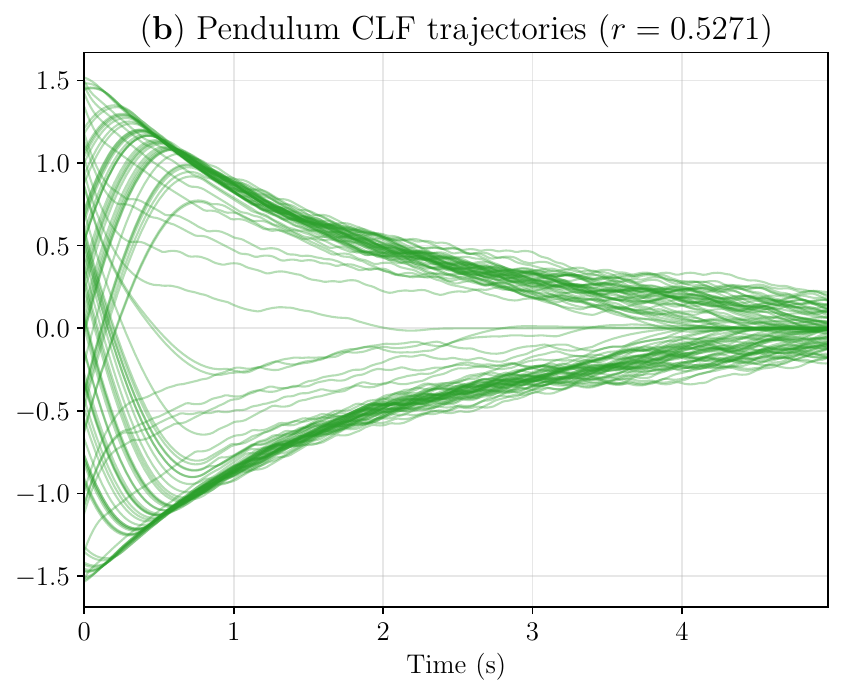}
    \end{subfigure}\hfill \\
    \begin{subfigure}[t]{0.49\linewidth}
        \includegraphics[width=\linewidth]{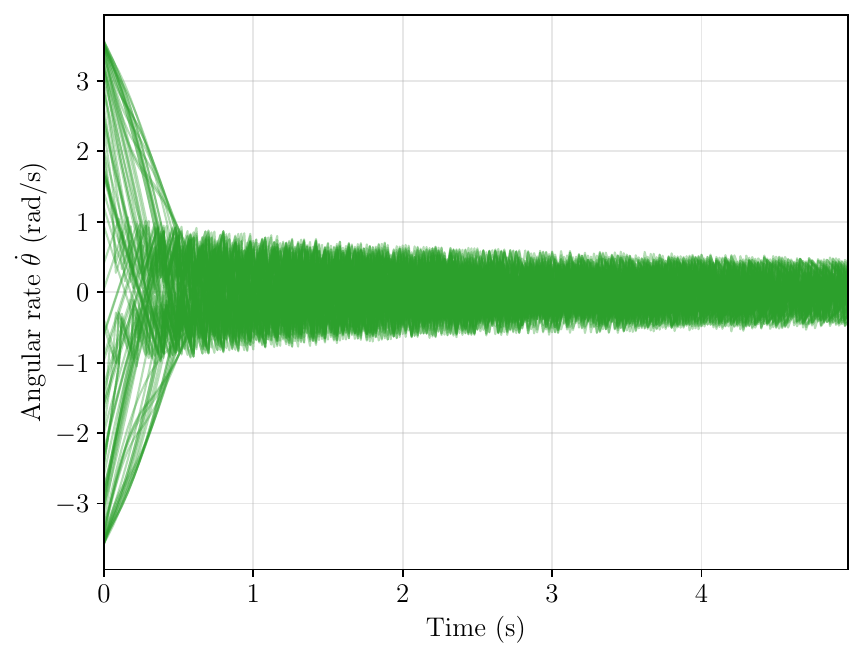}
    \end{subfigure}\hfill
    \begin{subfigure}[t]{0.49\linewidth}
        \includegraphics[width=\linewidth]{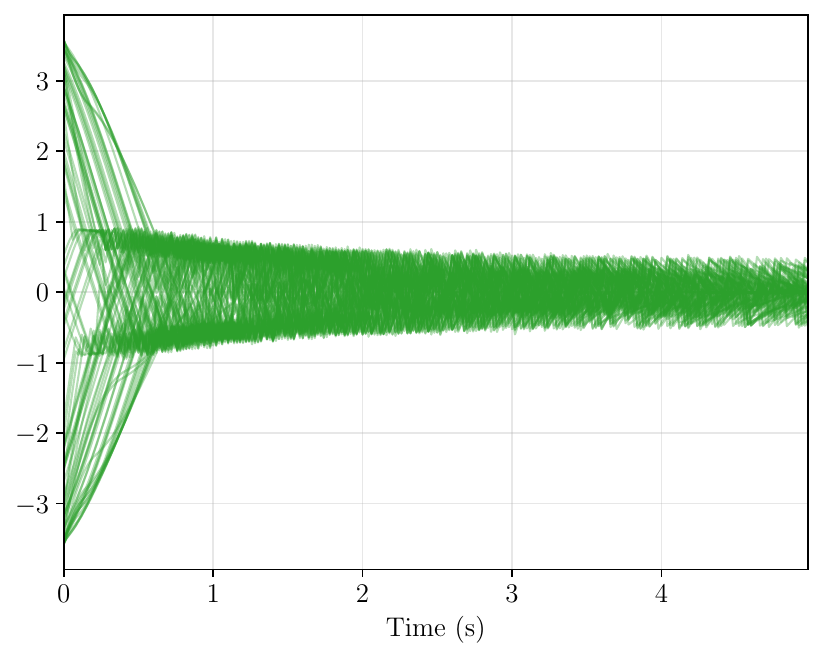}
    \end{subfigure}\hfill 
    \caption{Inverted pendulum trajectories. (a)~$r\!=\!r_0\!=\!2$. (b)~$r\!=\!r_{\text{calibrate-once}}$. Colors indicate stability violation per Theorem~\ref{thm:det_ct_clf}.}
    \label{fig:ip-trajectories-ab-large}
    \subonly{\vspace{-.8cm}}
\end{figure}

\begin{figure}[H]
    \centering
    \captionsetup{font=small,skip=2pt}
    \begin{subfigure}[t]{0.49\linewidth}
        \includegraphics[width=\linewidth]{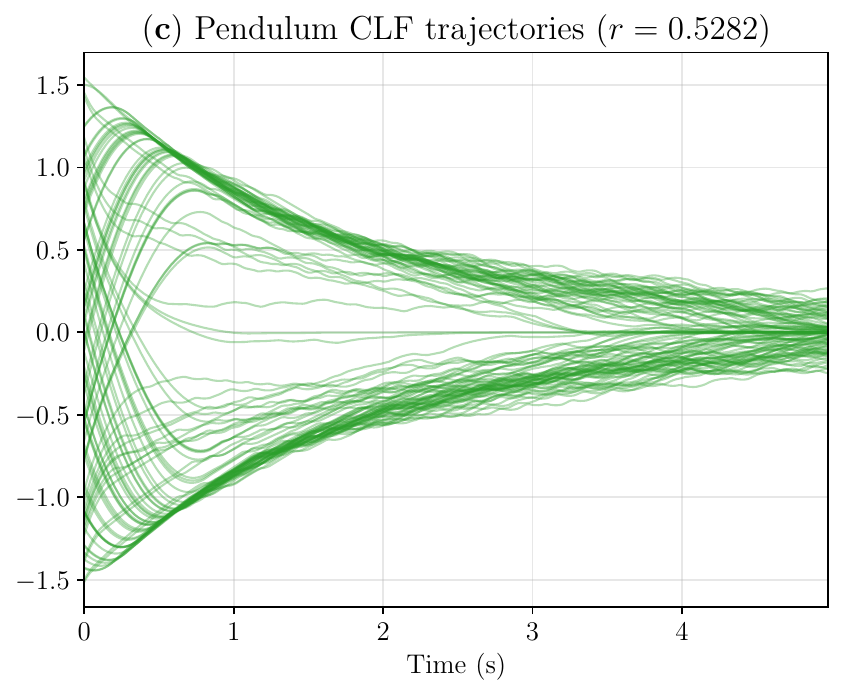}
    \end{subfigure}\hfill
    \begin{subfigure}[t]{0.49\linewidth}
        \includegraphics[width=\linewidth]{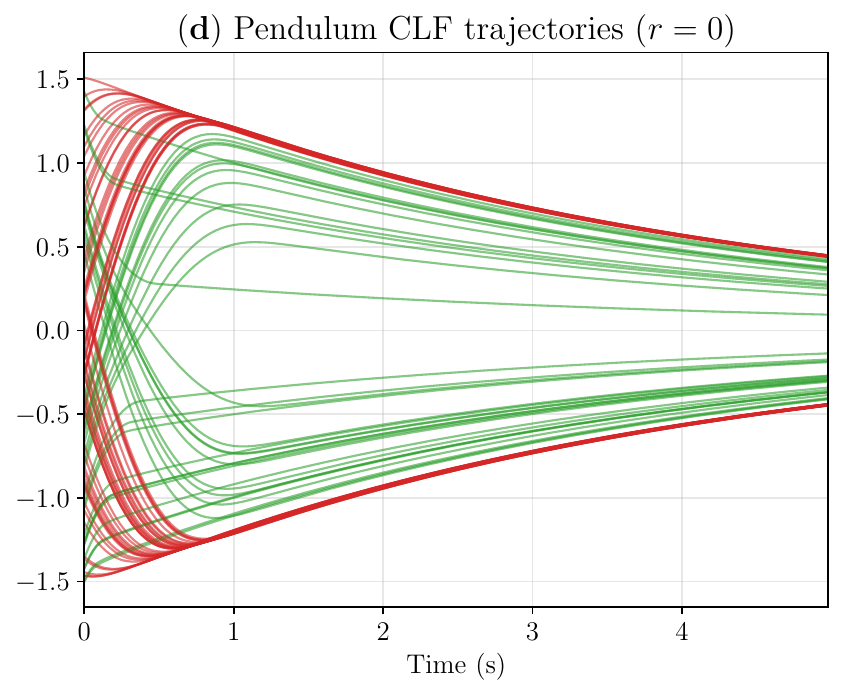}
    \end{subfigure} \\
    \begin{subfigure}[t]{0.49\linewidth}
        \includegraphics[width=\linewidth]{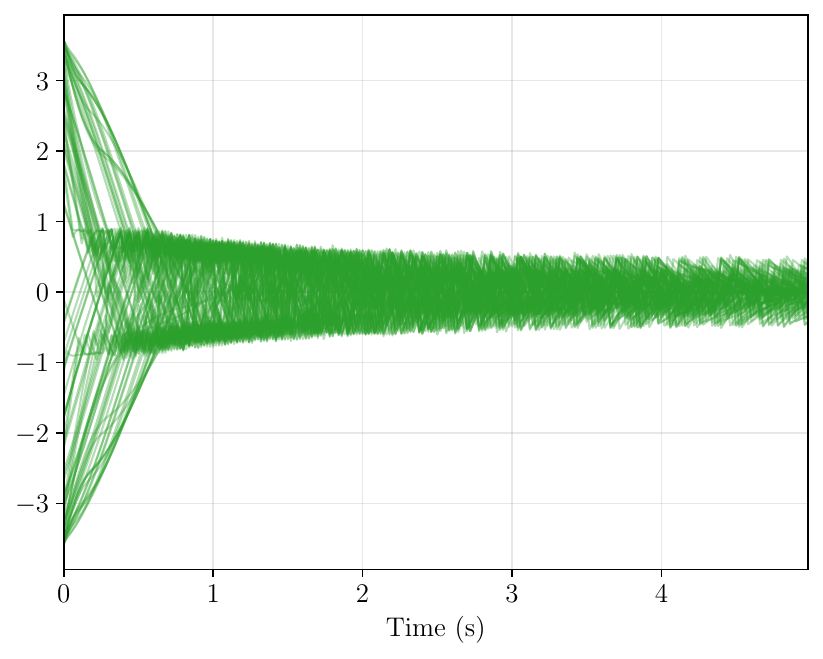}
    \end{subfigure}\hfill
    \begin{subfigure}[t]{0.49\linewidth}
        \includegraphics[width=\linewidth]{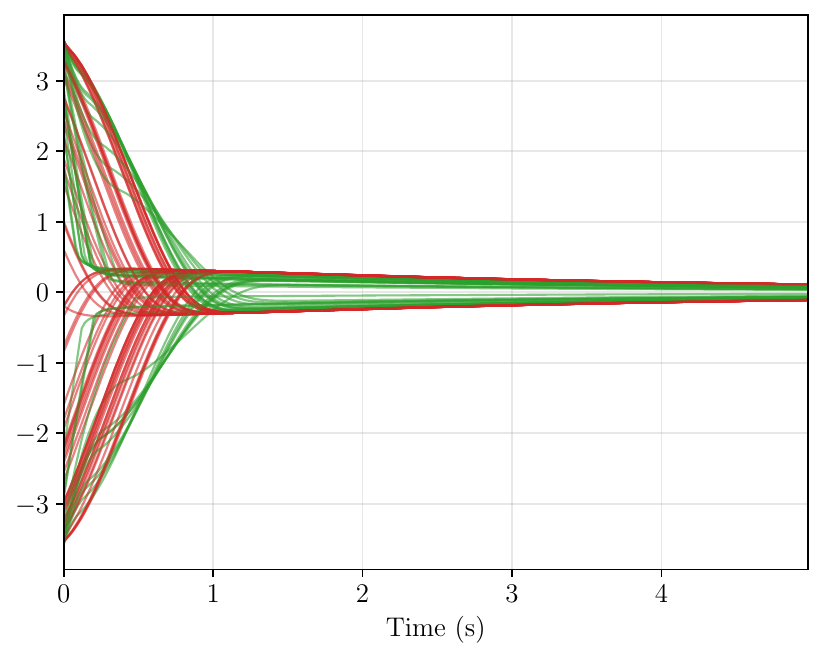}
    \end{subfigure}
    \caption{Inverted pendulum trajectories. (c)~$r\!=\!r_{\text{robust}}$. (d)~$r\!=\!0$. Colors indicate stability violation per Theorem~\ref{thm:det_ct_clf}.}
    \label{fig:ip-trajectories-cd-large}
    \subonly{\vspace{-.8cm}}
\end{figure}

\subsection{Multi-Obstacle Maze}

\begin{figure}[H]
    \centering
    \begin{subfigure}[t]{0.48\textwidth}
        \includegraphics[width=\linewidth]{figures/Final_plots/maze_margin.pdf}
        \caption{Margin $r_j$}
    \end{subfigure}\hfill
    \begin{subfigure}[t]{0.48\textwidth}
        \includegraphics[width=\linewidth]{figures/Final_plots/maze_score_cov.pdf}
        \caption{Score coverage}
    \end{subfigure}\\[0.5em]
    \begin{subfigure}[t]{0.48\textwidth}
        \includegraphics[width=\linewidth]{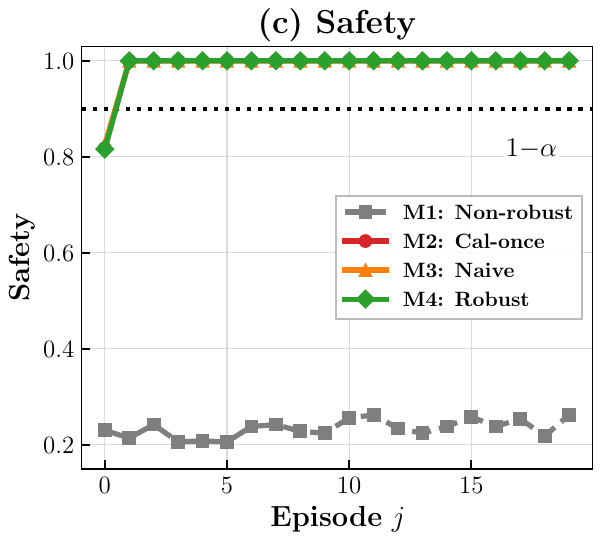}
        \caption{Safety (non-robust omitted; safety $\approx 0.23$)}
    \end{subfigure}
    \caption{Maze episodic results (enlarged).}
    \label{fig:maze-episodic-large}
\end{figure}

\begin{figure}[H]
    \centering
    \includegraphics[width=\textwidth]{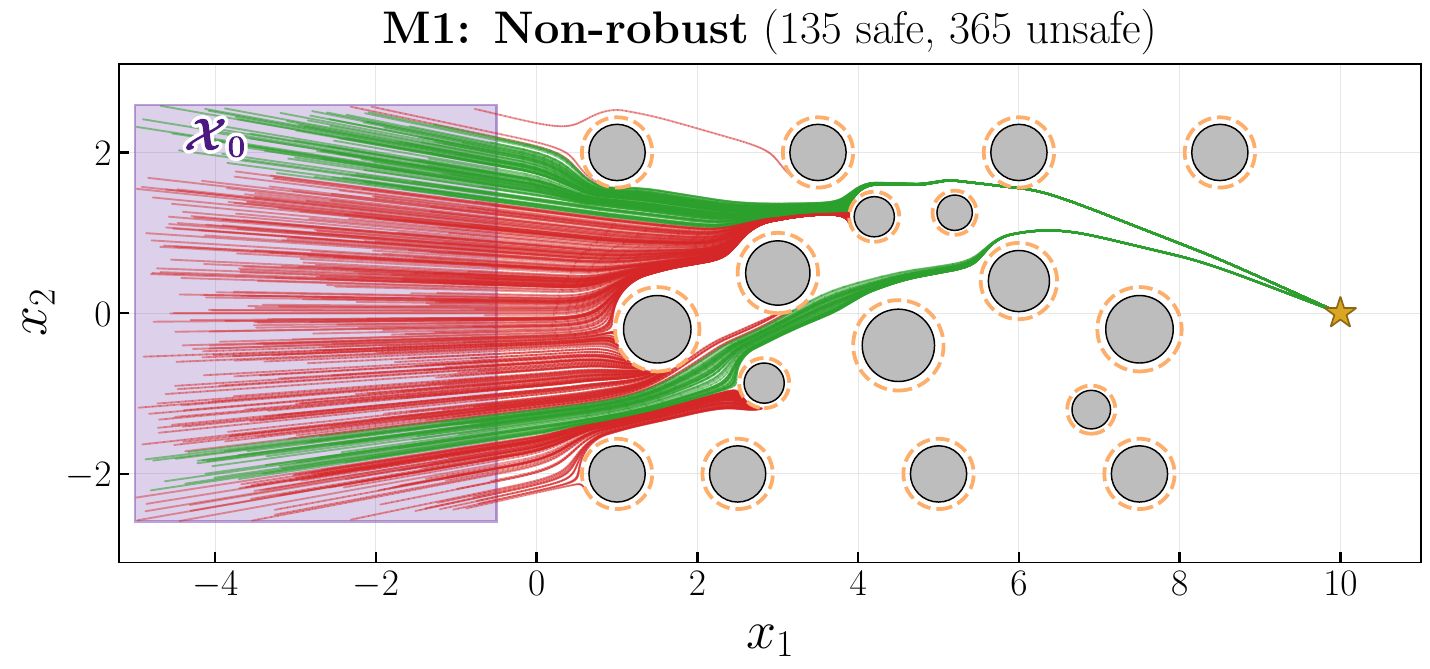}
    \caption{Maze trajectories: non-robust ($r\!=\!0$). 21 safe, 79 unsafe.}
    \label{fig:maze-traj-orig-large}
\end{figure}

\begin{figure}[H]
    \centering
    \includegraphics[width=\textwidth]{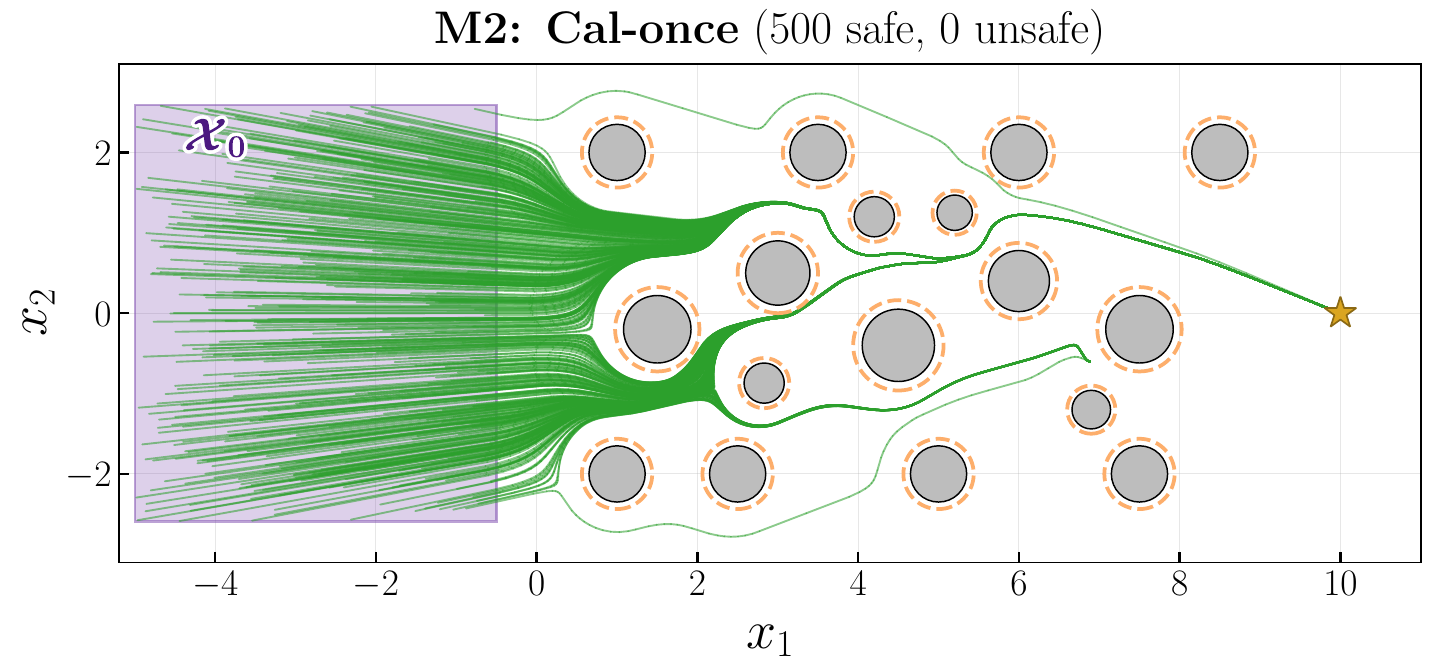}
    \caption{Maze trajectories: calibrate-once. 100 safe, 0 unsafe.}
    \label{fig:maze-traj-cal-large}
\end{figure}

\begin{figure}[H]
    \centering
    \includegraphics[width=\textwidth]{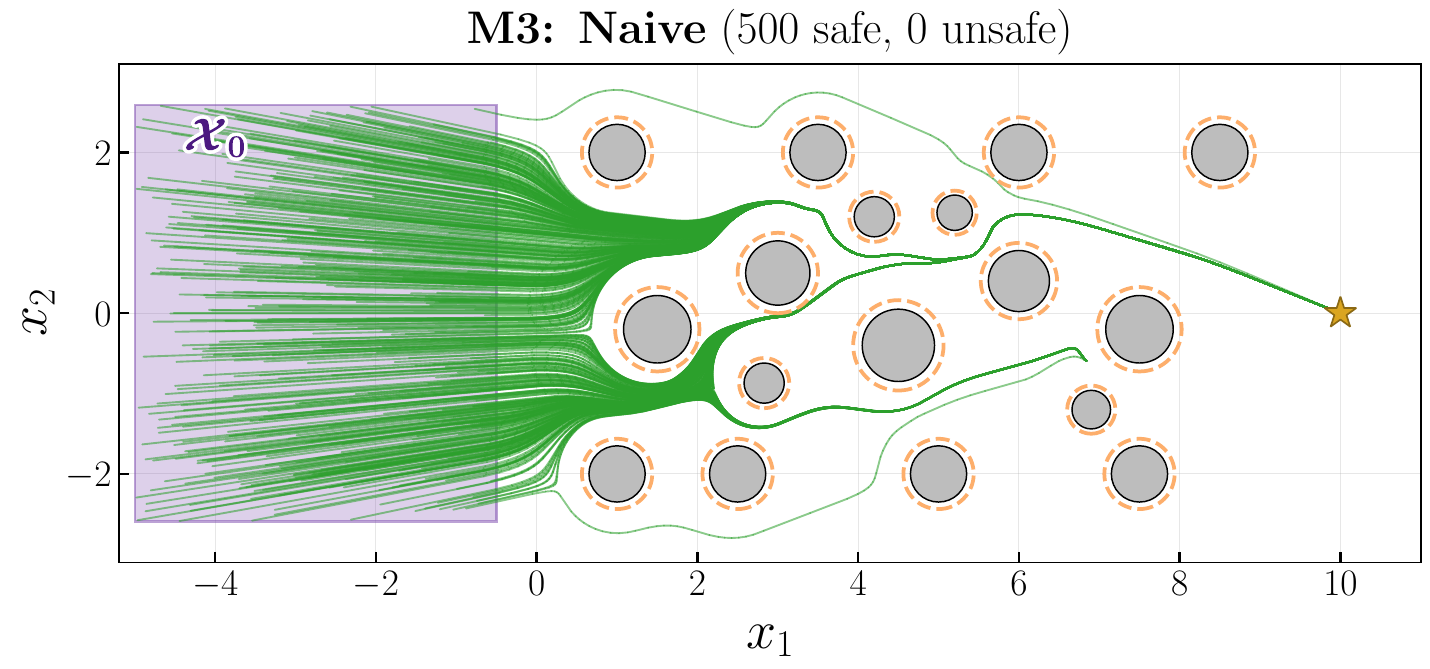}
    \caption{Maze trajectories: naive. 100 safe, 0 unsafe.}
    \label{fig:maze-traj-naive-large}
\end{figure}

\begin{figure}[H]
    \centering
    \includegraphics[width=\textwidth]{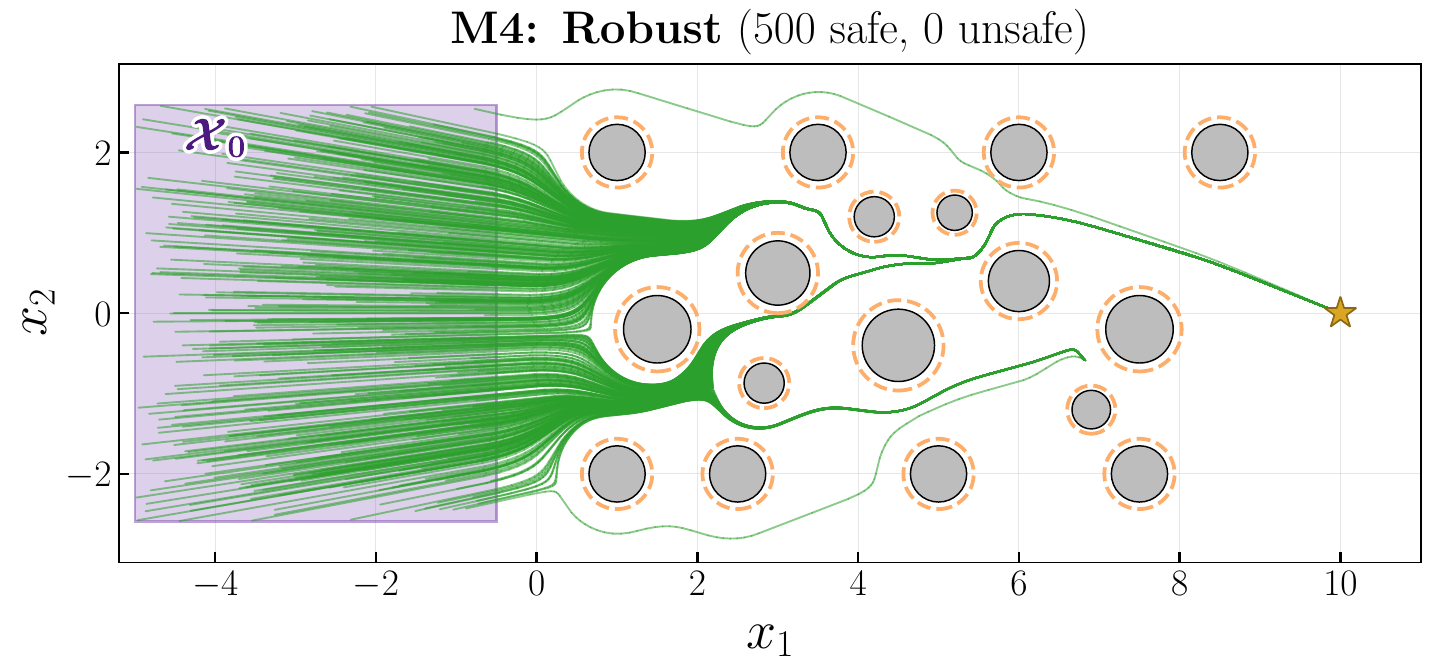}
    \caption{Maze trajectories: SR-CR (ours). 100 safe, 0 unsafe.}
    \label{fig:maze-traj-srcr-large}
\end{figure}

\subsection{Quadcopter Obstacle-Avoidance}

\begin{figure}[H]
    \centering
    \begin{subfigure}[t]{0.48\textwidth}
        \includegraphics[width=\linewidth]{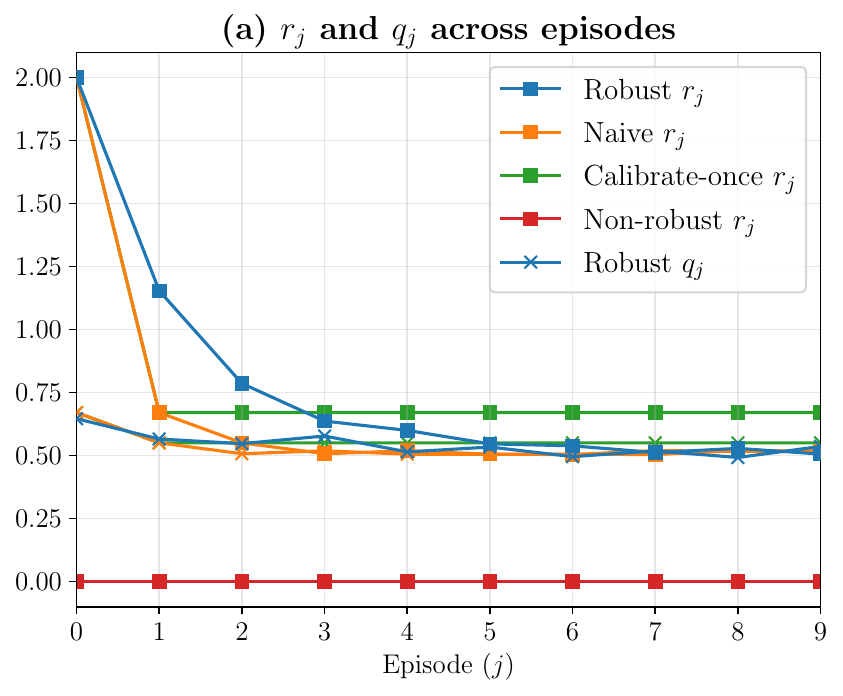}
    \end{subfigure}\hfill
    \begin{subfigure}[t]{0.48\textwidth}
        \includegraphics[width=\linewidth]{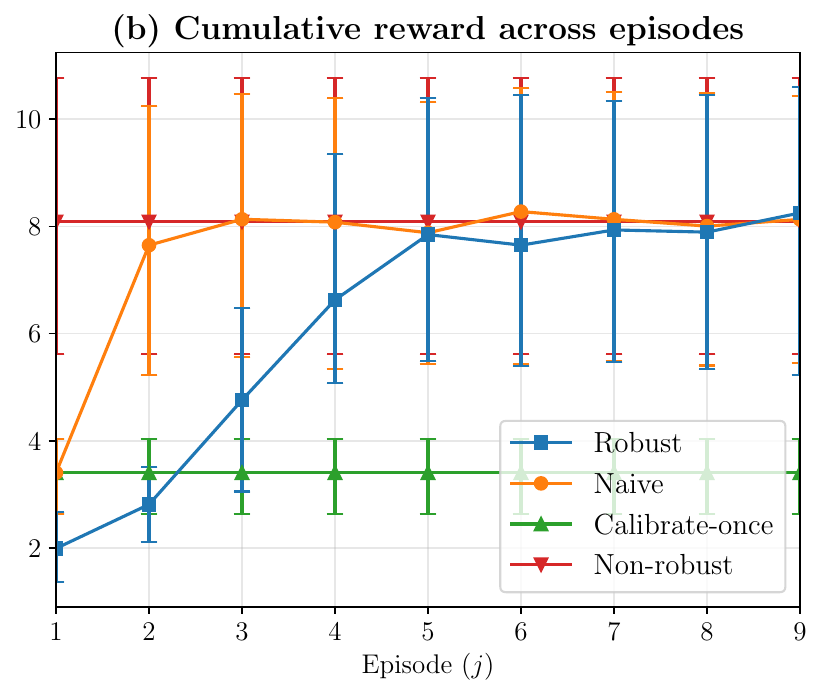}
    \end{subfigure}\\[0.5em]
    \begin{subfigure}[t]{0.48\textwidth}
        \includegraphics[width=\linewidth]{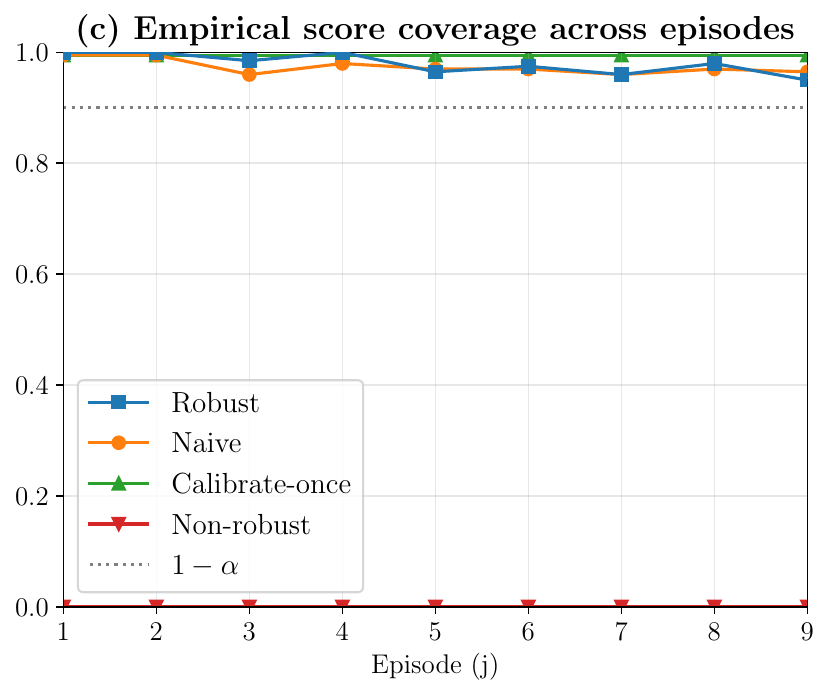}
    \end{subfigure}\hfill
    \begin{subfigure}[t]{0.48\textwidth}
        \includegraphics[width=\linewidth]{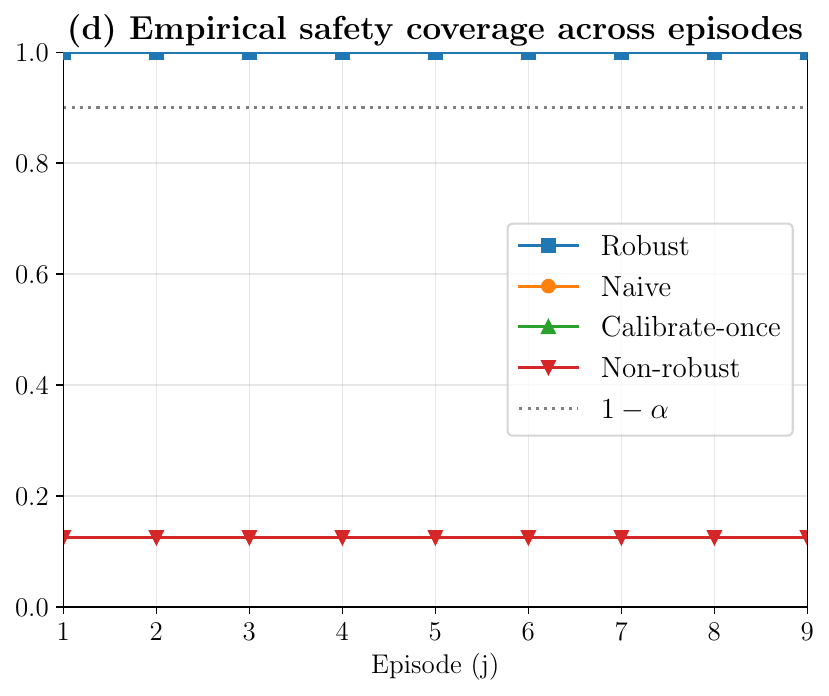}
    \end{subfigure}
    \caption{Quadcopter obstacle-avoidance (enlarged).
    (a)~$r_j$, $q_j$ across episodes. (b)~Cumulative progress towards goal. (c)~Score coverage $s_j^{(i)}\!\le\! r_j$. (d)~Safety: fraction with $\max_{t} h(x(t))\!\le\! 0$. Non-robust $q_j$ omitted from (a) ($q=250.38$, collision dynamics).}
    \label{fig:obs-results-large}
\end{figure}

\begin{figure}[H]
    \centering
    \begin{subfigure}[t]{0.48\textwidth}
        \includegraphics[width=\linewidth]{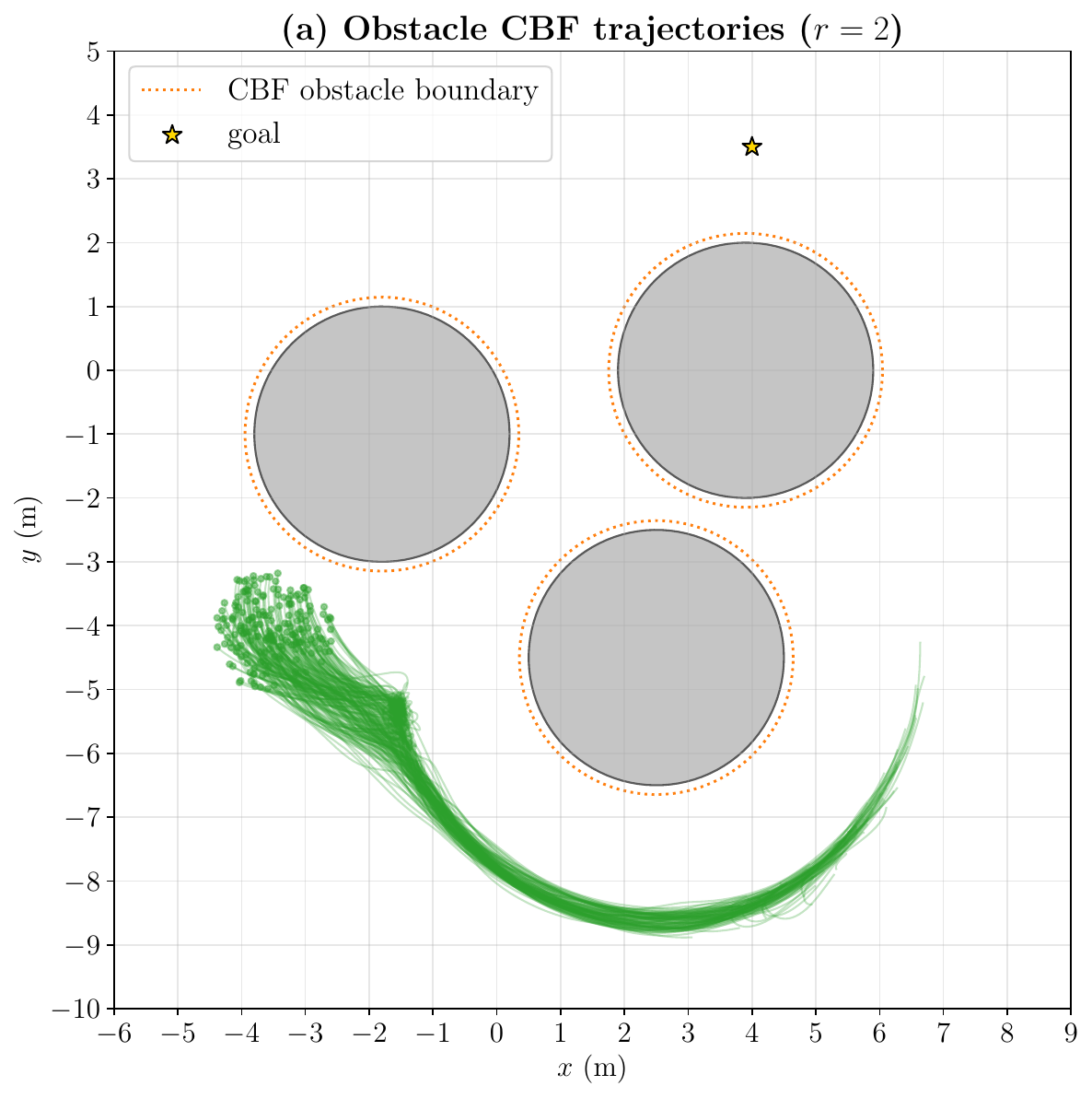}
        \caption{$r=r_0=2$}
    \end{subfigure}\hfill
    \begin{subfigure}[t]{0.48\textwidth}
        \includegraphics[width=\linewidth]{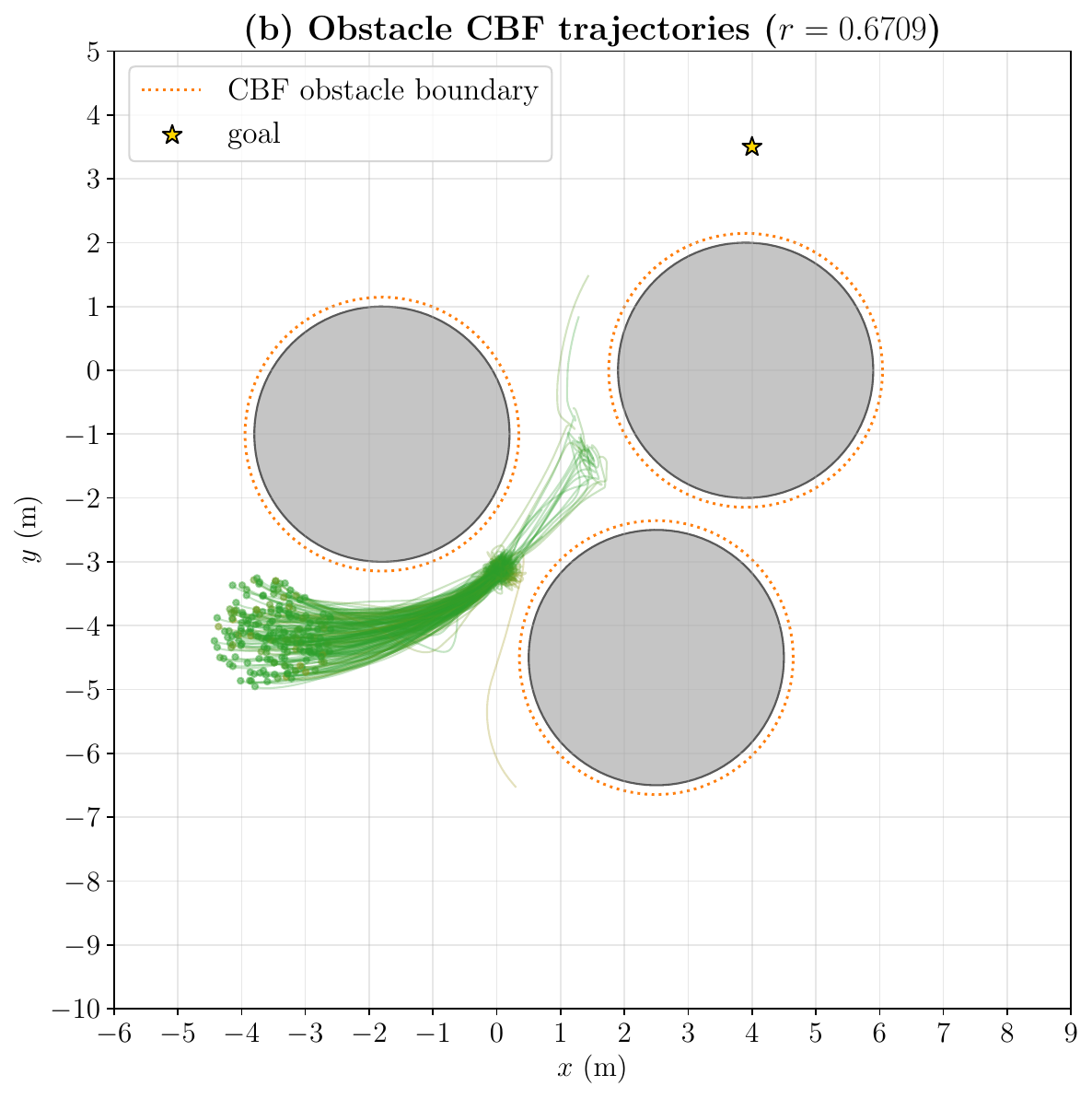}
        \caption{$r=r_{\text{cal-once}}$}
    \end{subfigure}\\[0.5em]
    \begin{subfigure}[t]{0.48\textwidth}
        \includegraphics[width=\linewidth]{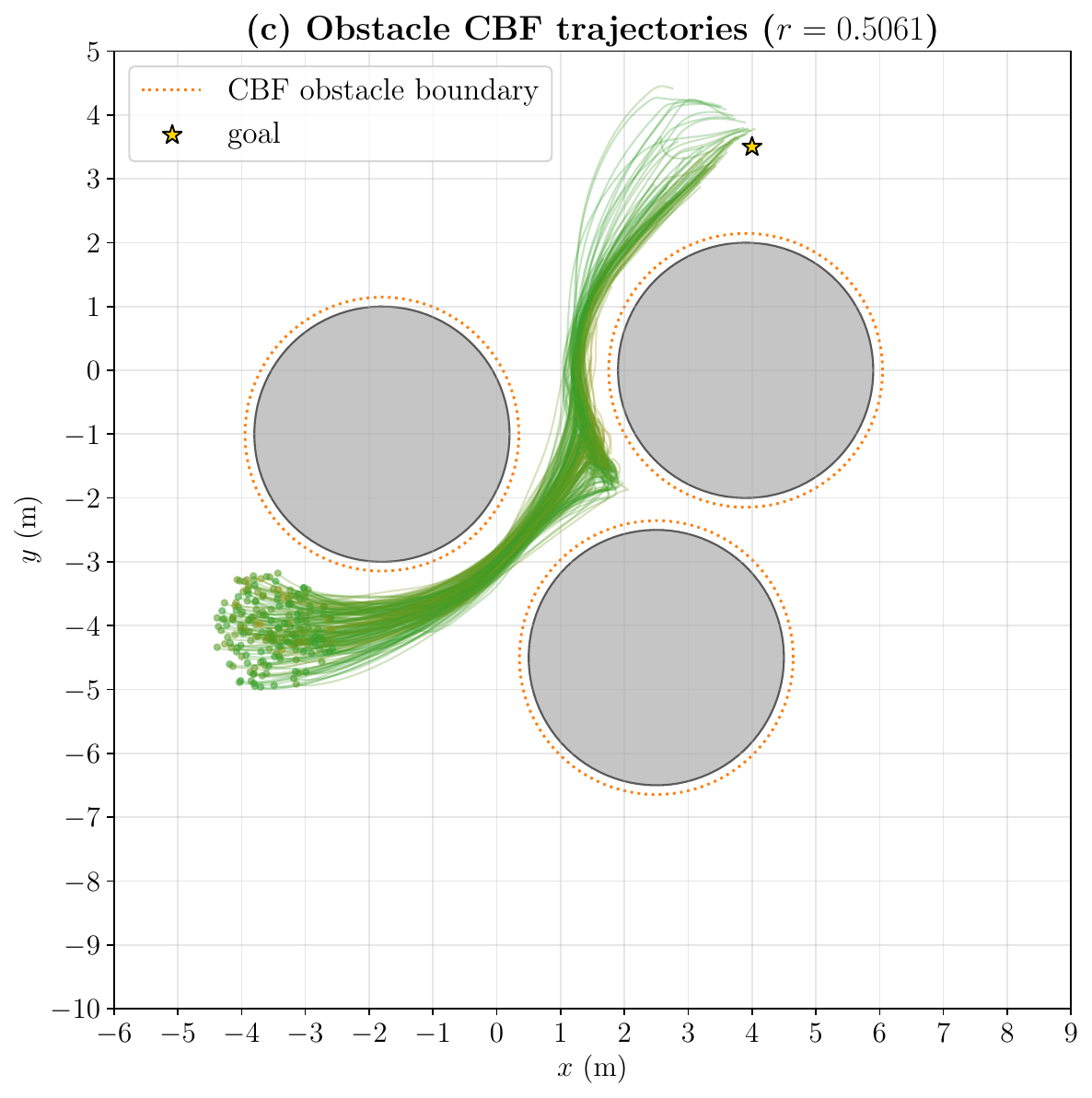}
        \caption{$r=r_{\text{robust}}$}
    \end{subfigure}\hfill
    \begin{subfigure}[t]{0.48\textwidth}
        \includegraphics[width=\linewidth]{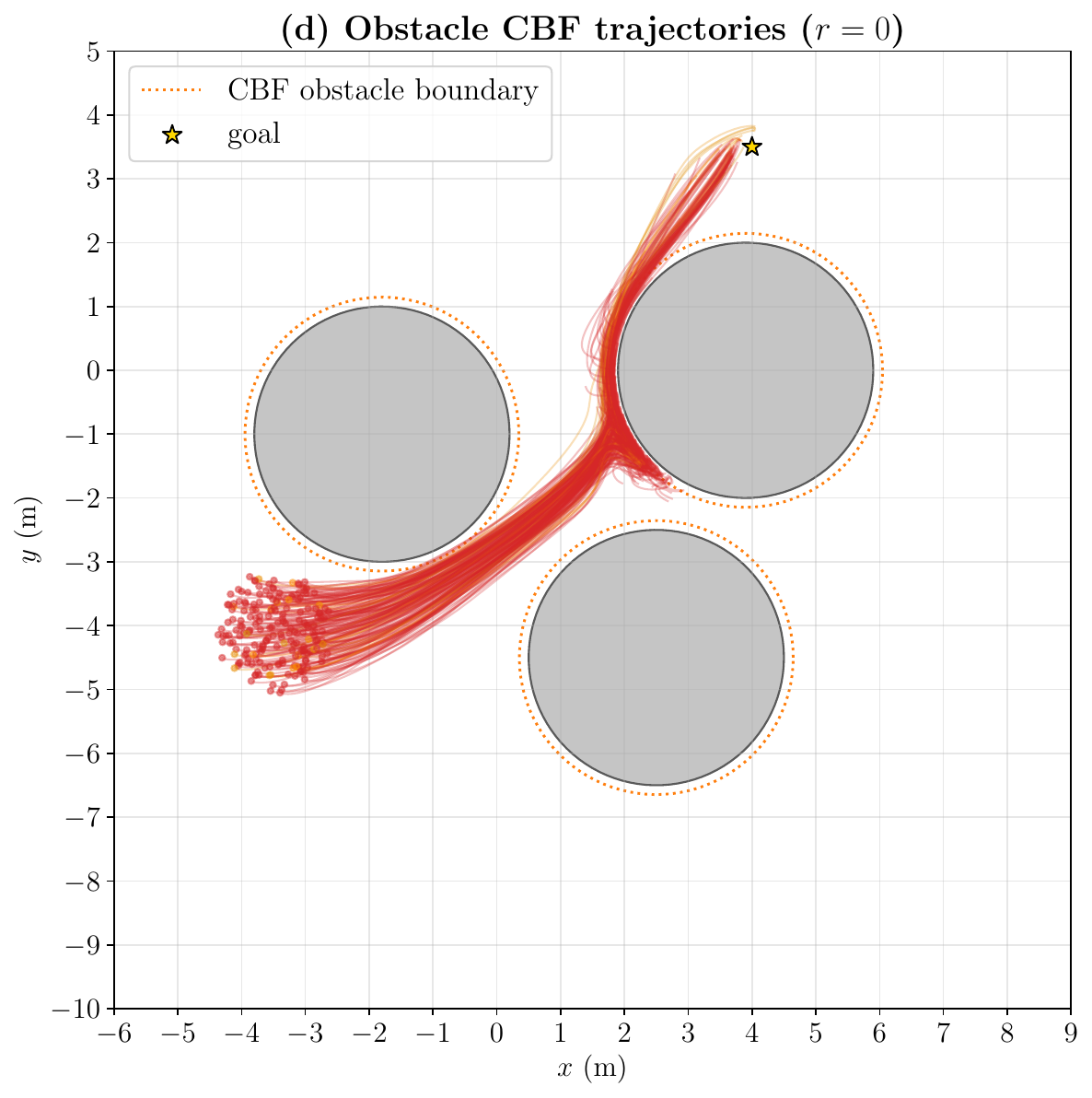}
        \caption{$r=0$}
    \end{subfigure}
    \caption{Quadcopter trajectories (enlarged).
    Orange: $h(x)\!=\!0$ boundary. Green/red: safe/unsafe.}
    \label{fig:obs-trajs-large}
\end{figure}

\arxonly{\twocolumn}

\ifsubmit
  \makeatletter
  \end{minipage}\end{lrbox}%
  \makeatother
  \end{refsection}
\fi
 
\end{document}